\documentclass[10pt,sigconf]{acmart}

\usepackage{graphicx}
\usepackage{multirow}
\usepackage{algorithm}
\usepackage{algorithmic}
\usepackage{color}
\usepackage{setspace}
\usepackage{makecell}
\usepackage{amsthm}
\usepackage{color}
\usepackage{subfig}
\usepackage{booktabs}
\usepackage{balance}
\usepackage{array}
\usepackage{flushend}
\usepackage{balance}
\usepackage{url}

\newtheorem{myDef}{Definition}
\newtheorem{myTheo}{Theorem}
\newtheorem{myLemma}{Lemma}
\newtheorem{exmp}{Example}
\newcommand{\Pset}{\mathcal{P}}

\copyrightyear{2020} 
\acmYear{2020} 
\setcopyright{acmcopyright}
\acmConference[SIGMOD '20] {2020 ACM SIGMOD International Conference on Management of Data}{June 14--19, 2020}{Portland, OR, USA}
\acmBooktitle{2020 ACM SIGMOD International Conference on Management of Data (SIGMOD'20), June 14--19, 2020, Portland, OR, USA}
\acmPrice{15.00}
\acmDOI{10.1145/3318464.3389735}
\acmISBN{978-1-4503-6735-6/20/06}

\begin{document}
\fancyhead{}
\title{Distributed Processing of {\em k} Shortest Path Queries over Dynamic Road Networks}

\affiliation{\large{\institution{Ziqiang Yu$^1$,\hspace{0.2em} Xiaohui Yu$^2$,\hspace{0.2em} Nick Koudas$^3$,\hspace{0.2em} Yang Liu$^4$,\hspace{0.2em} Yifan Li$^2$,\hspace{0.2em} Yueting Chen$^2$,\hspace{0.2em} Dingyu Yang$^5$}}}

 
\affiliation{\large{\institution{Yantai University$^1$ \hspace{0.2em} York University$^2$\hspace{0.2em} University of Toronto$^3$\hspace{0.2em} Wilfrid Laurier University$^4$\hspace{0.2em} Alibaba Group$^5$}}}

\begin{abstract}
The problem of identifying the {\em k}-shortest paths (KSPs for short) in a dynamic road network is essential to many location-based services. Road networks are dynamic in the sense that the weights of the edges in the corresponding graph constantly change over time, representing evolving traffic conditions. Very often such services have to process numerous KSP queries over large road networks at the same time, thus there is a pressing need to identify distributed solutions for this problem. However, most existing approaches are designed to identify KSPs on a static graph in a sequential manner (i.e., the $(i+1)^{th}$ shortest path is generated based on the $i^{th}$ shortest path), restricting their scalability and applicability in a distributed setting. We therefore propose KSP-DG, a distributed algorithm for identifying {\em k}-shortest paths in a dynamic graph. It is based on partitioning the entire graph into smaller subgraphs, and reduces the problem of determining KSPs into the computation of partial KSPs in relevant subgraphs, which can execute in parallel on a cluster of servers. A distributed two-level index called DTLP is developed to facilitate the efficient identification of relevant subgraphs. A salient feature of DTLP is that it indexes a set of virtual paths that are insensitive to varying traffic conditions, leading to very low maintenance cost in dynamic road networks. This is the first treatment of the problem of processing KSP queries over dynamic road networks. Extensive experiments conducted on real road networks confirm the superiority of our proposal over baseline methods.
\end{abstract}
\maketitle
\vspace{-0.5em}
\section{Introduction}\label{sec:intro}
In this work, we are concerned with identifying {\em k}-shortest paths (KSPs for short) over dynamic road networks. That is, for a given dynamic road network and a pair of origin and destination, identify the {\em k} shortest loop-less paths according to a predefined measure (e.g., travel time). The road network can be considered a graph where the intersections/endpoints are represented as vertices and roads as edges. It is dynamic in the sense that the travel time (or other similar measures such as congestion level) changes over time, corresponding to evolving weights of edges in the graph.  

Identifying KSPs in a dynamic road network is an essential building block in many location-based services: 

\begin{itemize}
\item Given a pair of origin and destination, most navigation services offer a number of candidate routes to their users, for which the shortest routes (in terms of either travel time or distance) are almost always generated along with other considerations such as avoiding toll roads or highways  \cite{li2014lpv,scano2015adaptations,liu2018finding}. As an example, Figure~\ref{google-map} shows the top-3  routes recommended by Google Maps for the trip from the Empire State Building to the Lincoln Center for the Performing Arts in New York City.  In practice, it is common for such navigation services to handle numerous concurrent queries simultaneously with a need to return the results immediately, in the presence of constantly evolving traffic conditions. 
\item In ride-sharing, services such as Uber and Lyft arrange one-time shared rides on very short notice  \cite{Ridesharing-vldb-2018,Ridesharing-Libin-2019}. When matching drivers and passengers, various factors are taken into account. For example, it is desirable to present the driver with a few alternative shortest routes from the origin to the destination to choose from. The aim is to present the driver with options to explore alternates between potential earnings (e.g., by picking up more passengers) and associated delays. With millions of users and thousands of rides to be arranged at any time, it is critical for  service providers to process $k$ shortest path queries efficiently. 
\end{itemize}

\captionsetup{font=scriptsize}
\begin{figure}[htbp]
\centering
\includegraphics[width=0.33\textwidth]{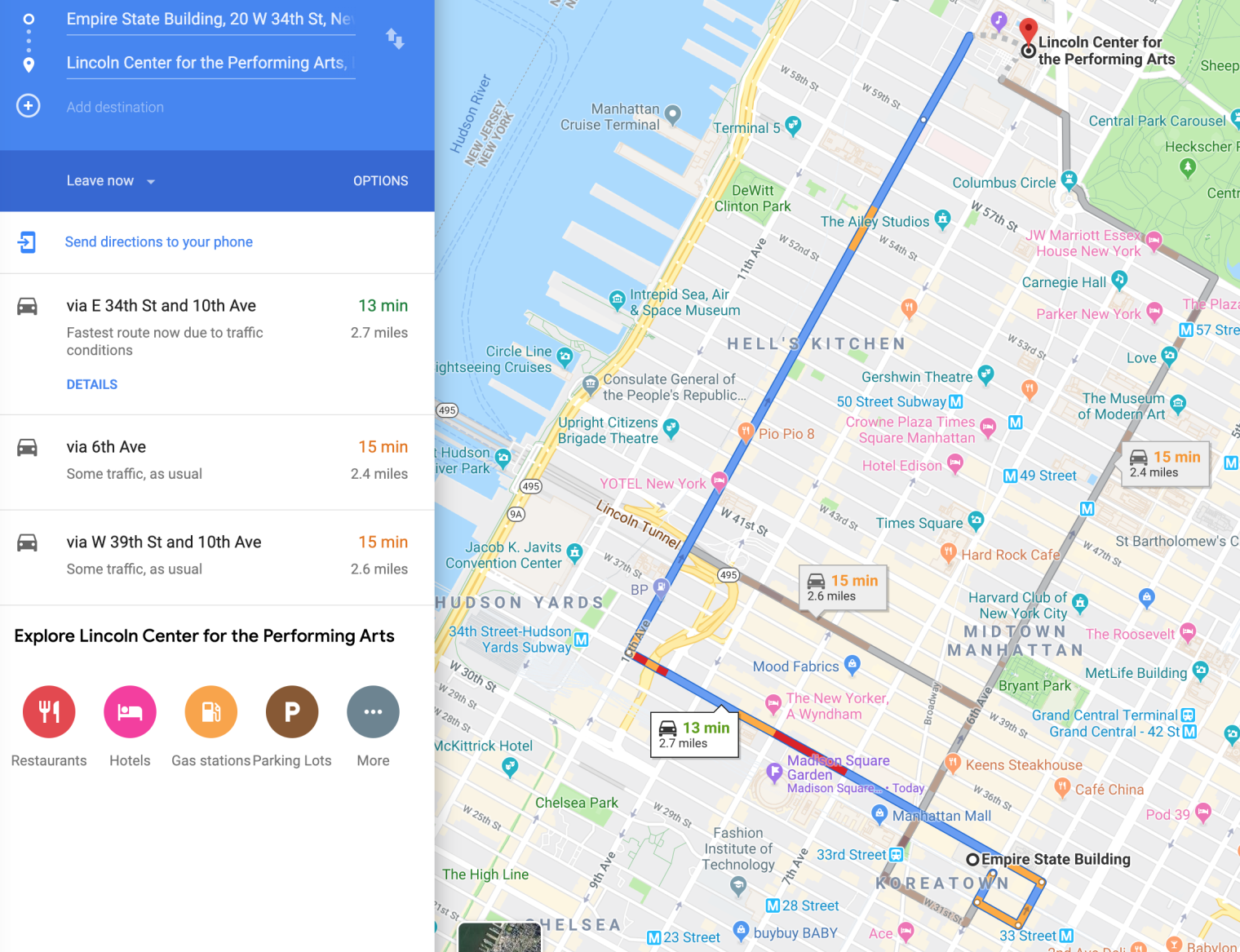}
\caption{Three Routes Recommended by Google Maps for a Trip from the Empire State Building to the Lincoln Center for the Performing Arts In NYC }\label{google-map}
\vspace{-2.35em}
\end{figure}

Most of the existing work to address the KSP problem in a road network (or more generally, in a graph) assumes a centralized approach \cite{liu2018finding,yen1971finding,katoh1982efficient,eppstein1998finding,hershberger2007finding,gao2010fast,gao2012holistic}, and is incapable of handling large volumes of concurrent queries over a dynamic graph for two main reasons. First, the query processing strategies employed are sequential, namely they generate the  $(i+1)^{th}$ shortest path based on the $i^{th}$ shortest path, which limits their scalability with respect to the number of concurrent queries. Second, some previous approaches \cite{eppstein1998finding,gao2010fast,gao2012holistic} require building a heavy-weight path index for every query, which is not only costly but also unwieldy as the index becomes invalid once the edge weights in the graph change.

Distributed algorithms have recently been developed to compute shortest paths in static graphs \cite{chandy1982distributed, DistributedSP-Baruch-1989,Elkin2017Distributed,Ghaffari2017Improved, aridhi2015mapreduce,li2017distributed, qiu2018parapll}, but directly adopting them to process KSP queries in dynamic road networks is still an uphill battle. First, they usually require building indexes beforehand to enhance search efficiency, which is impractical for dynamic graphs as the indexes will quickly become obsolete. Second, most of them identify the shortest path only, which cannot be trivially extended to determine {\em k} shortest paths in a distributed setting. A distributed algorithm CANDS \cite{yang2014cands} is proposed to find the single shortest path in a dynamic graph. Again, it is non-trivial to extend it to the case of KSPs. Moreover, the indexed shortest paths in CANDs change continuously and require frequent re-computation in a dynamic graph, which can be quite expensive. 

In order to tackle the aforementioned problems, we propose a distributed algorithm that can be deployed on a cluster of servers to handle KSP queries over large dynamic graphs. 

The primary design goal is to support large volumes of concurrent queries. At a very high level, our solution adopts a divide-and-conquer strategy: the large graph is partitioned into multiple subgraphs maintained on different servers; given a query, we first compute in parallel partial shortest paths in the subgraphs on different servers, which are subsequently merged to construct the {\em k} shortest paths. 

Although similar strategies have been adopted for distributed query evaluation on graphs \cite{Fan2012,yang2014cands}, the processing of KSP queries presents some unique challenges: (1) It is non-trivial to compute the KSPs from partial results on subgraphs. In particular it is challenging to identify the subgraphs containing the required partial shortest paths as well as determining the combination of such partial paths. 
(2) It is vital but difficult to build an effective index for identifying KSPs that is easy to maintain in the presence of constantly changing weights. Directly indexing the shortest paths between pre-selected vertices, a strategy widely used in existing work, would be too expensive in a dynamic graph due to the requirement to handle  frequent updates of edge weights. 

With these challenges in mind, we first propose a Distributed Two-Level Path (DTLP) index. This index is based on partitioning any graph $G$ into multiple subgraphs, where two subgraphs may overlap at a small number of vertices called {\em boundary vertices} but have no common edges. For any two boundary vertices in the same subgraph, we compute specific paths between them called {\em bounding paths}. The boundary vertices and the bounding paths in all subgraphs form the first level of DTLP, which provides a lower bound of the shortest distance between a pair of boundary vertices. These bounding paths do not change due to varying weights, making the index easily maintainable. The second level, is a skeleton graph $G_\lambda$, in which the vertices correspond to the boundary vertices of all subgraphs, and there exists an edge between a pair of vertices in $G_\lambda$ if and only if there is a set of bounding paths between the corresponding vertices in the original graph, with the weight of the edge computed based on that set of bounding paths. Graph $G_\lambda$ serves the purpose of supplying an approximate search direction for identifying KSPs.

{We propose an iterative algorithm, KSP-DG, to identify KSPs in  Dynamic Graphs based on DTLP}, which follows a "filter-and-refine" strategy. At each iteration of the algorithm, we first use the skeleton graph $G_\lambda$ in the filter step to compute a shortest path in $G_\lambda$ that has not been examined before. The vertices along this path correspond to the boundary vertices in $G$, and the subgraphs of $G$ covering these boundary vertices are selected for further examination. In the refine step, $k$ shortest paths between the boundary vertices are generated from each subgraph, which are combined to form complete paths in $G$. These paths are then used to update the list of shortest paths in $G$ that have been obtained so far. This process terminates when there is no more change in the list, and the final result is guaranteed to be the KSPs. The algorithm is distributed by design, in that the partitioning of the graph allows the refine step to run in parallel over the cluster, and the small footprint of the skeleton graph $G_\lambda$ lends itself well to be replicated to any node in the cluster where it is needed for generating the shortest paths in $G_\lambda$. 

Experiments are conducted on datasets generated based on three real road networks in three metropolitan centers to (1) compare the performance of the proposed method against that of baseline methods, and (2) conduct a sensitivity analysis of the proposed method with respect to varying parameters of interest (e.g., number of servers in the cluster, data characteristics). Our results demonstrate orders of magnitude performance improvement over other applicable approaches.

 It is worth noting that although the techniques developed here are motivated by identifying KSPs in road networks, they are potentially applicable in other applications involving KSPs on graphs with evolving edge weights\cite{kou2016social}. For example, in a sensor network, one may want to identify a set of shortest paths with the lowest energy cost from the source to the destination and dynamically re-route these paths in a probabilistic fashion to even out the power usage on different nodes in the network \cite{shah2002energy,zhang2013finite,meng2016consensus,li2017variance}.

In summary, we make the following contributions. 
\begin{itemize}
\item We study the problem of identifying KSPs in dynamic road networks, which is an important function of many location-based services. To the best of our knowledge, our work is the first to investigate this problem and to propose a suite of solutions.

\item We devise DTLP, a two-level path index suitable for deployment in a distributed setting to support the processing of KSP queries over dynamic graphs.

\item We propose KSP-DG for processing KSP queries based on DTLP. Diverging from existing centralized methods, KSP-DG decomposes the problem of identifying KSPs in the entire graph into searching for partial KSPs in different subgraphs in parallel, which can be implemented in a distributed fashion on a cluster of servers.

\item We conduct extensive experiments on real road networks to evaluate the performance of the proposed approach, confirming its effectiveness and superiority over other approaches across a variety of settings.

\end{itemize}
The rest of the paper is organized as follows. Section \ref{sec:overview} defines the problem of finding KSPs. Section \ref{sec:dtlp} presents the DTLP index, and KSP-DG is discussed in Section \ref{sec:KSP-DG}. Section \ref{sec:exp} experimentally evaluates the performance of our proposal. Section \ref{sec:related-work} discusses  related work, and Section \ref{sec:con} concludes this paper.

\section{Problem Definition}\label{sec:overview} 
In this section, we define terminology for our ensuing discussion and formally present the KSPs identification problem.  

\newcommand{\V}{\mathcal{V}}
\newcommand{\E}{\mathcal{E}}
\newcommand{\W}{\mathcal{W}}
\begin{myDef}[ Dynamic undirected graph] A dynamic graph $G$ = ($\mathcal{V}$, $\mathcal{E}$, $\mathcal{W}$) consists of (1) a finite set of vertices $\V$, (2) a set of edges $\E \subseteq \V\times \V$, where $e_{i,j}$ ($e_{i,j}\in \E$) denotes an edge between vertices $v_i$ and $v_j$, and (3) a set of non-negative weights $\W$, where $w_{i,j}\in \W$ is the weight of edge $e_{i,j}$, which may change by a negative or non-negative value $\Delta w$ at any time point.
\end{myDef}
A dynamic road network as discussed before can be considered a special case of a dynamic undirected graph. As such, in what follows, we present the terminologies and solutions in the context of graphs instead of road networks, in view of their potential applicability in other scenarios involving dynamic undirected graphs.

\vspace{-0.05in}
\begin{myDef}
[Subgraph] A graph ${SG}$=($\V'$, $\E'$, $\W'$) is a subgraph of the graph $G$= ($\V$, $\E$, $\W$) iff
\begin{itemize}
    \item $\V'\subseteq \V$,
    \item $\E'\subseteq \E$ $\land (e_{m, n}\in \E'\to v_m, v_n\in \V')$, 
    \item $\W'\subseteq \W$ $\land (w_{m, n}\in \W'\to e_{m, n}\in \E')$.
\end{itemize}
\end{myDef}

\vspace{-0.07in}
\begin{myDef}[Path, Distance of Path] Path $P(s, t)$ from one vertex $v_s$  (the source vertex) to another vertex $v_t$ (the destination vertex) in graph $G$ is a sequence of vertices $\langle \texttt{v}_0$= $v_s$,$\cdots$ $\texttt{v}_l$,$\cdots$,$\texttt{v}_n$ =$v_t\rangle$ such that $\forall l\in [1,n]$, $e_{i,j}\in \E$ if $\texttt{v}_{l-1}=v_i$ and $\texttt{v}_l=v_j$. 
In this work, we only consider simple paths, i.e., paths with no repeat vertices.
The distance of $P(s, t)$ is defined as $D(P(s, t))$=$\sum\limits_{i=1}^n w_{i-1,i}$.
\end{myDef}

\begin{myDef}[{\em k}-shortest path query (KSP query)] For a given pair of source and destination vertices, $v_s$ and $v_t$, the $k$-shortest path query, $q(v_s, v_t)$, identifies the set of $k$ paths from $v_s$ to $v_t$ in graph $G$, $\Pset_{s,t}=\lbrace {P_1}(s,t), \ldots, P_k(s,t) \rbrace$, such that $D(P_i(s, t)) \le D(P_{i+1}(s, t)) (i\in[1, k-1]$), and $\forall P(s,t)\notin \Pset_{s,t}$,\ $D(P(s,t)) \ge D(P_k(s,t))$. 
\end{myDef}

To ensure timely processing of the queries, we assume that query processing takes place in main memory. Moreover, as graph $G$ constantly evolves as queries arrive,  we use a buffer $G_{curr}$ to model the current version of $G$ to ensure unambiguous semantics of the query answers. This buffer is updated continuously and asynchronously as the graph changes. At fixed time intervals, a snapshot $G_{curr}$ is taken, and the answer to an incoming KSP query is processed against the most recent snapshot. 
Each query answer has a timestamp indicating the moment at which the answer is exact.

\section{Distributed Two-Level Path Index}\label{sec:dtlp}
A naive approach to process KSP queries is to compute from scratch the shortest paths on graph $G$ directly each time a query arrives. In practice, however, both the size and the dynamic nature of the road network, as well the vast volume of concurrent queries, make it infeasible to process queries this way. For this reason we consider an index structure that could scale and speed up the processing. In this section, we discuss the desiderata for the index, and present our proposal, the Distributed Two-Level Path (DTLP) index . 

\subsection{Desiderata}
Our problem setting necessitates the following desirable properties in an index to support processing KSP queries. 

{\bf (1) Low maintenance cost for dynamic graphs.} Some existing index structures for processing shortest path queries maintain the shortest path between each pair of vertices in the graph (or pairs of selected ``landmark" vertices). However, when the edge weights of the graph constantly change, this becomes extremely expensive to maintain. As such, one of our requirements is that the index structure must be easy to maintain for dynamic graphs and imposes as little overhead as possible when the graph evolves. 

{\bf (2) Suitable for deployment in a distributed setting.} As the size of the graph and the number of concurrent queries keep increasing, a distributed index becomes a more attractive solution than a centralized one due to its ability to scale out. Since the graph $G$ itself may have to be partitioned and stored across the cluster, we should be able to maintain  parts of the index on different nodes in a cluster corresponding to  subgraphs of $G$. 

{\bf (3) Supporting distributed algorithms.} With the graph partitioned and subgraphs stored across the cluster, query processing takes place on different nodes in parallel. It is thus necessary for the index to assist the identification of relevant subgraphs that may contain pieces of the shortest paths to avoid examination of all subgraphs. Such an index should provide sufficient pruning power, and at the same time guarantee the correctness of the query result. 

\subsection{Overview of DTLP}

In view of the above requirements, we devise an index called DTLP that facilitates distributed KSP query processing over dynamic graphs. Based on a partitioning of $G$, DTLP has a two-level structure: the first level indexes each subgraph by maintaining a list of bounding paths (Section \ref{sec:bound-path}) between any pair of boundary vertices (Section \ref{sec:bv}). This provides the basis for computing the lower bounds of the shortest distances between the boundary vertices. The second level keeps a skeleton graph (Section \ref{sec:skeleton-graph}) with all boundary vertices in all subgraphs, and it is computed based on the bounding paths identified in the first level. Figure~\ref{fig.dtlp} illustrates the structure of the DTLP index.

\captionsetup{font=scriptsize}
\begin{figure}[!h]
\centering
\includegraphics[width=0.35\textwidth]{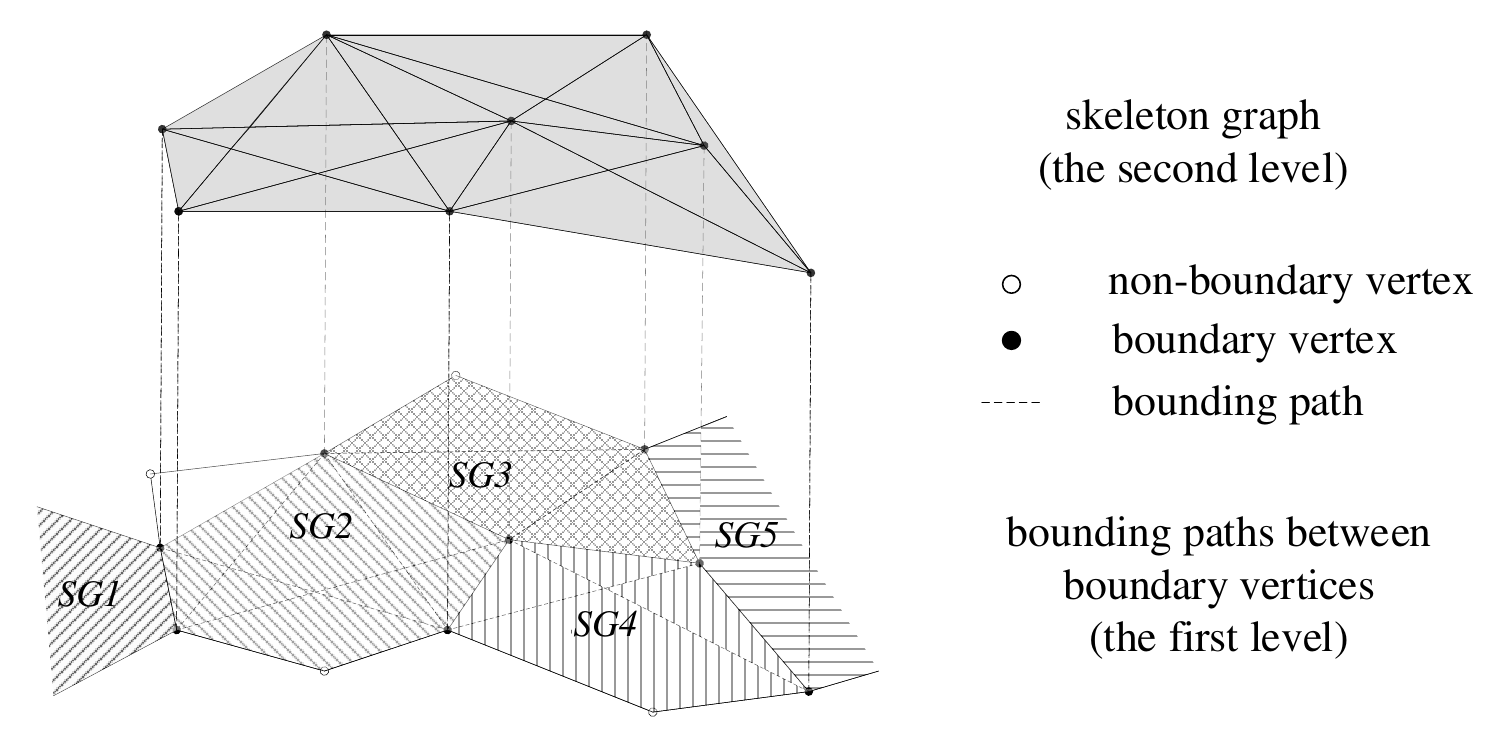}
\caption{Schematic Diagram of DTLP }\label{fig.dtlp}
\vspace{-0.1in}
\end{figure}

\subsection{Graph Partitioning and Boundary Vertices} \label{sec:bv}
The construction of DTLP starts with partitioning graph {\em G} into multiple subgraphs such that each subgraph has at most $z$ vertices. Different subgraphs may share vertices but not edges. For graph partitioning, we start from any vertex and traverse the graph $G$ using a breadth-first strategy to generate the subgraphs. The set of subgraphs is denoted as $\mathcal{S}$=$\lbrace{SG}_{1}$, $\cdots$, $SG_{n}\rbrace$ such that (1)  $\V_1\cup\cdots\cup \V_n=\V$; (2)  ${\E}_1\cup\cdots\cup {\E}_n=\E$; (3) ${\W}_1\cup\cdots\cup {\W}_n=\W$, where $n$ is the number of subgraphs and $SG_i = \lbrace\V_i, \E_i, \W_i\rbrace$ $(i\in[1,n])$.  

The vertices shared by two or more subgraphs are called boundary vertices which are formally defined as follows.
\begin{myDef}[Boundary vertex]\label{def:boundary-vertex} A vertex $v\in\V$ is a boundary vertex if and only if $\exists SG_i, SG_j$ such that $v\in \V_i\cap\V_j$ $(i\neq j, i,j\in[1,n])$.
\end{myDef}

Evidently, any path from a non-boundary vertex in $SG_i$ to a non-boundary vertex in $SG_j$ must pass through one or more boundary vertices, as those boundary vertices are the only "contact vertices" between subgraphs.
\begin{exmp}
 Figure~\ref{global-graph} gives an example of graph $G$ that is partitioned into four subgraphs in Figure~\ref{fig:sub-par}, where the threshold $z$ (maximum number of vertices in a subgraph) is set to 6. The boundary vertices in each subgraph are shaded.
\end{exmp}

\captionsetup{font=scriptsize}
\begin{figure}[htbp]
\centering
\includegraphics[width=0.32\textwidth]{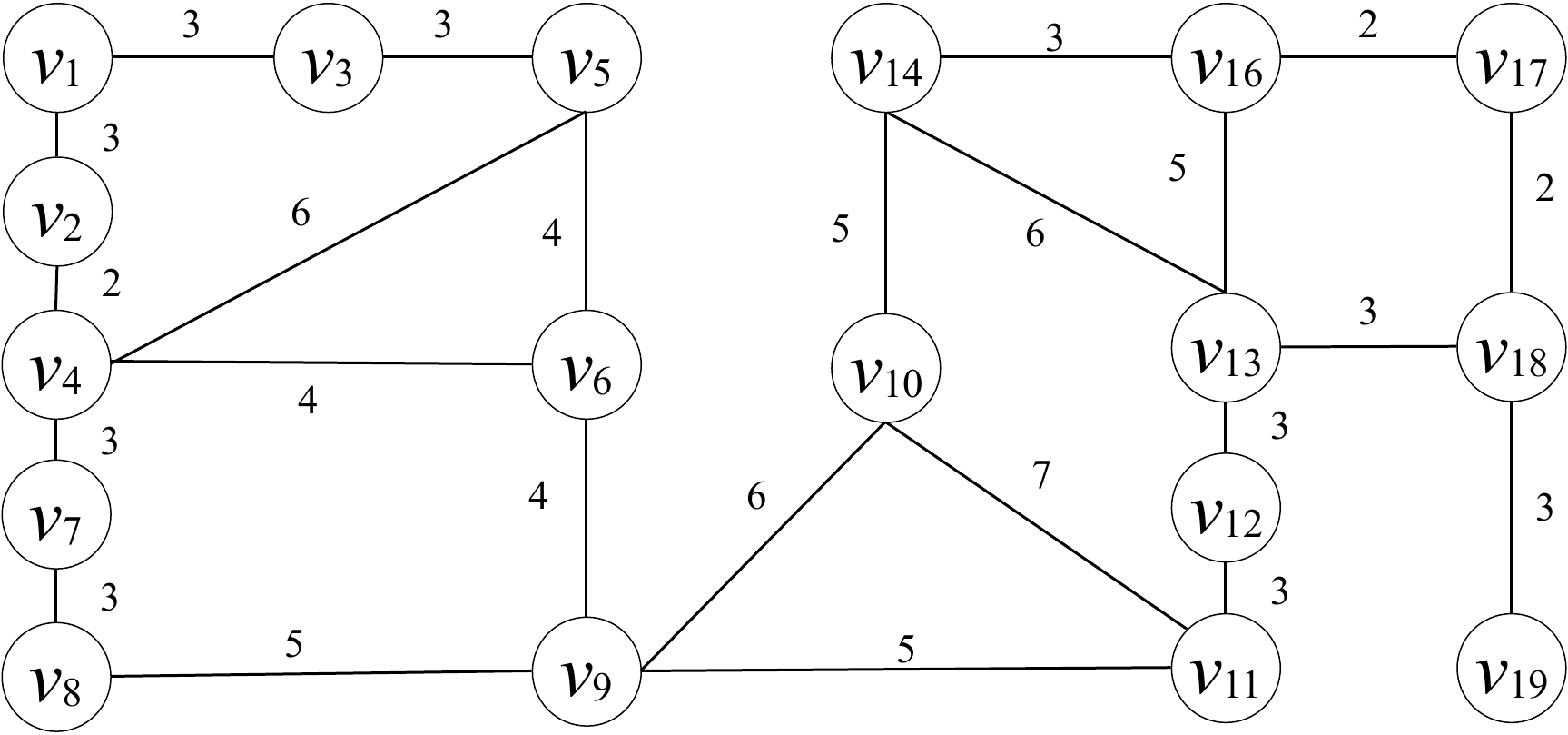}
\caption{Graph $G$}\label{global-graph}
\vspace{-0.3in}
\end{figure}

\begin{figure}[!htpb]
\subfloat[{$SG_1$}]{\label{Fig.sub-par.1}
\includegraphics[width=0.0995\textwidth]{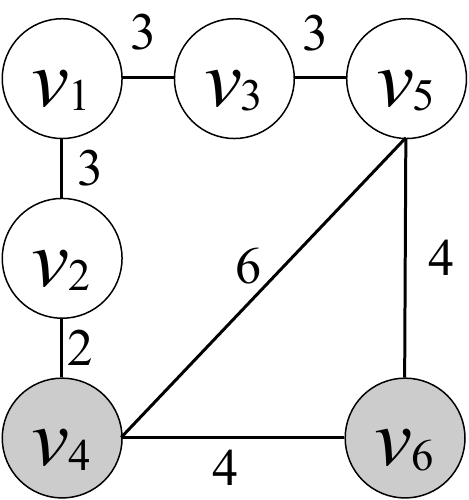}}
\quad
\subfloat[{$SG_2$}]{\label{Fig.sub-par.2}
\includegraphics[width=0.11\textwidth]{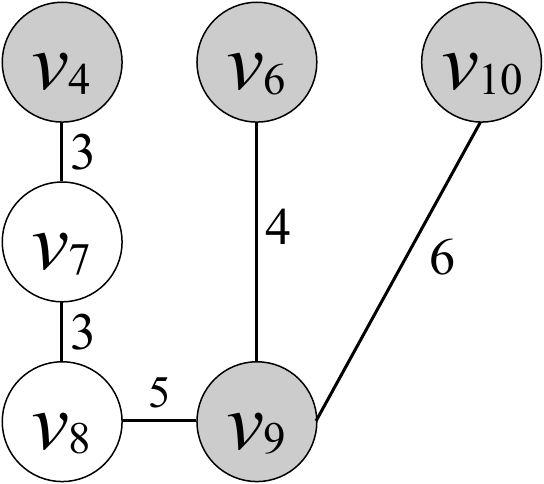}}
\hspace{0.1in}
\subfloat[{$SG_3$}]{\label{Fig.sub-par.3}
\includegraphics[width=0.102\textwidth]{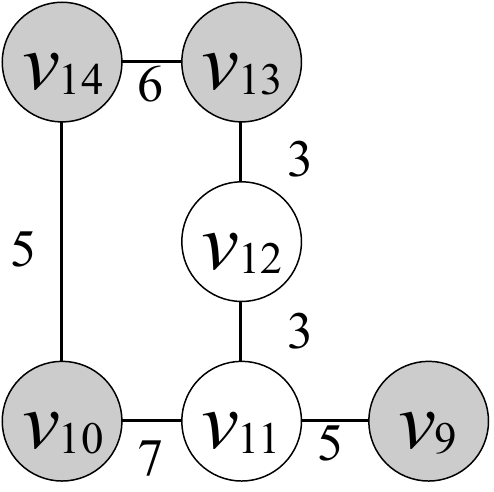}}
\subfloat[{$SG_4$}]{\label{Fig.sub-par.4}
\includegraphics[width=0.1\textwidth]{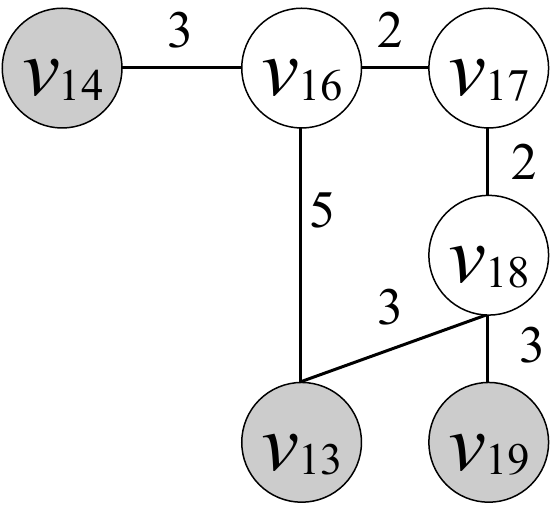}}
\caption{Subgraphs of $G$}\label{fig:sub-par}
\vspace{-0.1in}
\end{figure}

\vspace{-0.3cm}
\subsection{Bounding Paths}\label{sec:bound-path}

Once graph partitions are in place, we move on to identify for each pair of boundary vertices in a subgraph a set of bounding paths.  Bounding paths are specific paths in the subgraph serving as a reference to establish the bound distance, a stable lower bound of the shortest distance between the respective boundary vertices. The bounding paths do not change when the edge weights in the graph change, but the bound distances have to be updated to reflect the new information on edge weights. Those bounding paths will in turn be used to construct the skeleton graph.

As a first attempt, we consider utilizing the path(s) with the fewest edges between two boundary vertices as the bounding path(s). Suppose path $P(i, j)$ with $m$ edges is the path with the least number of edges between two boundary vertices $v_i$ and $v_j$ in subgraph $SG$. We consider it a bounding path, and the corresponding bound distance, denoted by $BD(P(i, j))$, is computed as the sum of the $m$ smallest edge weights in $SG$. It is easy to show that $BD(P(i, j))$ is not greater than the shortest distance between $v_i$ and $v_j$ in $SG$ and thus can be used as a lower bound of the shortest distance.

\begin{exmp}
 In Figure~\ref{Fig.sub.1}, ${P_1}(13, 14)$=$\langle v_{13}$, $v_{16}$, $v_{14}\rangle$ is a bounding path with two edges $(m=2)$ between $v_{13}$ and $v_{14}$ in $SG_4$. As the two smallest edge weights in $SG_4$ are both 2, we have $BD({P_1}(13, 14))=2+2=4$ and $D(P_1(13, 14))=5+3=8$ respectively. If $SG_4$ later changes to $SG'_4$ in Figure~\ref{Fig.sub.2}, the two smallest weights in $SG'_4$ both become 1. Correspondingly, the updated bounding distance is $BD({P_1}(13, 14))=1+1=2$, with the shortest distance between $v_{13}$ and $v_{14}$ being $6$.
 \vspace{-0.5cm}
\end{exmp}

\begin{figure}[!htpb]
\subfloat[{$SG_4$}]{\label{Fig.sub.1}
\includegraphics[width=0.115\textwidth]{Figures/sg4.pdf}}
\quad
\subfloat[{$SG'_4$}]{\label{Fig.sub.2}
\includegraphics[width=0.12\textwidth]{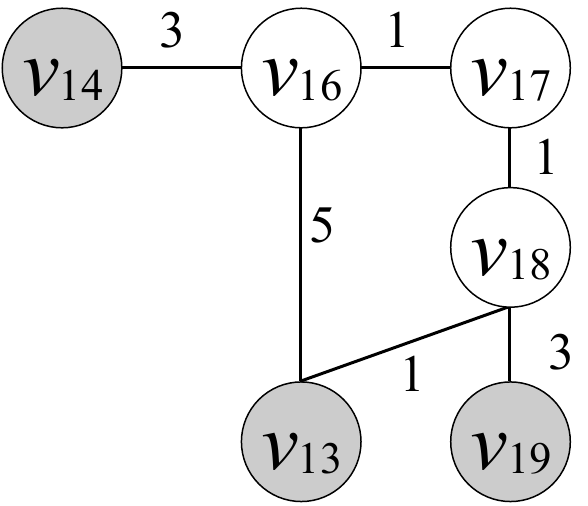}}
\caption{Subgraphs $SG_4$ and $SG'_4$}
\vspace{-0.1in}
\end{figure}

The above example demonstrates that the bound distance may be too loose, i.e., there may be a large discrepancy between the bound and the actual shortest distances. We thus seek to further improve this bound in two ways. 

First, we identify more than one bounding path between a pair of boundary vertices. For example, we can take the path  ${P_1}'(13, 14)$=$\langle v_{13}$, $v_{18}$, $v_{17}$, $v_{16}$, $v_{14}\rangle$, which is the path with the second smallest number of edges (at 4) between $v_{13}$ and $v_{14}$ in $SG_4$, to be another bounding path. Its corresponding bound distance in $SG'_4$ is 6, reducing the gap between the bound distance and the shortest distance.

Second, we decompose the edge weight to finer granularity to provide better resolution when computing the bound distance. 
In particular, we assume that each edge $e_{i,j}$ consists of $w^0_{i,j}$ {\em virtual fragments} ({\em vfrags} for short), where $w^0_{i,j}$ is the initial weight of $e_{i,j}$, i.e., the weight of $e_{i,j}$ at the beginning of the DTLP construction. Such vfrags are additive, i.e., the number of vfrags on a path, $\phi(P(i,j))$, is the sum of the number of vfrags on each edge in the path. 
We call the weight of each vfrag in $e_{i,j}$ the {\em unit weight}, defined as ${w_{i,j}}/{w^0_{i,j}}$, where $w_{i, j}$ is the current weight of edge $e_{i,j}$, which varies over time. Note that for any edge, the number of vfrags always remains the same, but the unit weight will change with varying weight. For example, in
Figure~\ref{Fig.sub.1}, initially the unit weight of all edges is $1$, and the initial weight of $e_{13,18}$ is $w^0_{13,18}=3$. When $SG_4$ changes to $SG'_4$ (Figure~\ref{Fig.sub.2}), $w^0_{13,18}$ is still $3$, but the unit weight on $e_{13,18}$ has changed to $\frac{w_{13,18}}{w^0_{13,18}} = \frac{1}{3}$.

With the two improvements above, we compute the bounding path and the bound distance as follows. 
For two boundary vertices $v_i$ and $v_j$, given a configurable parameter $\xi$, we compute a set  $\mathcal{B}_{i,j}$ consisting of at most $\xi$ bounding paths that contain the least number of vfrags (where bounding paths containing the same number of vfrags are counted as only one path). Formally,  $\forall P(i,j)\notin \mathcal{B}_{i,j}$ and $\forall {P'_l}(i,j)\in \mathcal{ B}_{i,j}$, $\phi(P(i,j))>\phi({P'_l}(i,j))$. 
Here, ${P'_l}(i,j)$ ($l\in [1, |\mathcal{ B}_{i,j}|]$) is called a bounding path, and ${P'_1}(i,j)$  represents the bounding path with the fewest vfrags. 

\vspace{-0.2cm}
\begin{exmp} \label{exmp-bound-path}
For boundary vertices $v_{13}$ and $v_{14}$ in $SG_4$ in Figure~\ref{Fig.sub.1}, ${P'_1}(13,14)$=$\langle v_{13},v_{16},v_{14}\rangle$ and ${P'_2}(13,14)$=$\langle v_{13},v_{18},v_{17}$, $v_{16},v_{14}\rangle $ are the bounding paths if $\xi=2$. If $\xi$ is set to 1, only ${P'_1}(13,14)$ is the bounding path.
\end{exmp}

For a bounding path ${{P_l}'}(i,j)$ from $v_i$ to $v_j$ in subgraph $SG$, we can compute a corresponding bound distance, which is the sum of the $\phi({{P_l}'}(i,j))$ smallest unit weights in $SG$, denoted by $BD({{P_l}'}(i,j))$. 

\vspace{-0.2cm}
\begin{exmp}
For the bounding path ${P'_1}(13,14)$ in {\em Example \ref{exmp-bound-path}}, $\phi({P'_1}(13,14))=8$ and all unit weights in $SG_4$ are 1 initially, so $BD({P'_1}(13,14))=8$. When $SG_4$ change to $SG'_4$, the unit weights in $SG'_4$ are updated to $(1/3, 3)$, $(1/2, 4)$, $(1, 8)$, and $(2, 3)$, meaning there are 3 vfrags of unit weight $1/3$, 4 vfrags of unit weight $1/2$, and so on. Thus, $BD({P'_1}(13,14))$ can be computed using the 8 smallest unit weights, consisting of 3 unit weights of $1/3$, 4 unit weights of $1/2$, and 1 unit weight of $1$, with the result being $3\times \frac{1}{3}+4\times \frac{1}{2}+1\times 1=4$.
\end{exmp}

Introducing vfrags is helpful to tighten the bound on distances. The number of vfrags and the number of edges in a bounding path both remain the same as weights evolve. Moreover, the variation in the unit weights of virtual fragments are in general much smaller than that in the weights of edges. Thus, the bounding distance determined based on the smallest unit weights is usually much tighter than that based on the smallest edge weights.

Bounding paths can be computed in the initial graph offline using any off-the-shelf shortest path algorithm such as Dijkstra's algorithm. They do not have to be recomputed as graphs change since they remain the same regardless of the changing weights.

\subsection{Lower Bound Distances}
After identifying the set of bounding paths between $v_i$ and $v_j$ in subgraph $SG$, we aim to determine the bounding path that would impose as tight a lower bound on the shortest distance as possible.  This translates to locating the maximal bound distance corresponding to this set of bounding paths. It will be referred to as the {\em lower bound distance} between $v_i$ and $v_j$, and the corresponding bounding path is called the {\em lower bounding path}.

To reach a tighter lower bound, recall that for each bounding path ${p_l}'(i, j)$ between $v_i$ and $v_j$ in $SG$, its bound distance computed using vfrags is never greater than its actual distance. Consequently, if this actual distance is not greater than the bound distance of another bounding path ${p_g}'(i, j)$, then the actual distance of ${p_l}'(i, j)$ is not greater than that of ${p_g}'(i, j)$. Exploiting this relationship allows us to identify the shortest path between $v_i$ and $v_j$ among the set of bounding paths in many cases. If this fails, however, we can still utilize the bounding path with the maximal bound distance as the lower bounding path. We  detail this approach below.

\begin{myDef}
[Lower Bounding Path]\label{def:lower-bound-path} Let $\mathcal{B}_{i,j}$ be a set of bounding paths between $v_i$ and $v_j$ in $SG$ and ${P_1}(i,j)$ be the shortest path connecting these two vertices in the original graph $G$. A bounding path ${P'_b}(i, j)$ is the lower bounding path between $v_i$ and $v_j$ if it satisfies one of the following conditions:

(1) $D({P'_b}(i, j))=D({P_1}(i, j))$; or

(2) the bound distance corresponding to ${P'_b}(i, j)$ is maximal among all bounding paths in $\mathcal{B}_{i,j}$. That is, $\forall P'(i, j)\in \lbrace\mathcal{B}_{i,j}\setminus \lbrace {P'_b}(i, j)\rbrace\rbrace$, $BD({P'_b}(i, j))\geq BD(P'(i, j))$.
\end{myDef}

\begin{myDef}[Lower Bound Distance]
Following Definition \ref{def:lower-bound-path}, the lower bound distance between $v_i$ and $v_j$, denoted by $LBD(i,j)$, equals to $D({P_1}(i, j))$ if condition (1) is established, or $BD({P'_b}(i, j))$ if condition (2) is met.
\end{myDef}

Since any pair of boundary vertices $v_i$ and $v_j$ can co-occur in more than one subgraph, there may be multiple lower bound distances associated with them, each for one subgraph. We call the least of these lower bound distances the {\em minimum lower bound distance}, denoted by $MBD(i,j)$.

Next, we introduce the following theorem to help identify the lower bound distance for a pair of boundary vertices. 

\begin{myTheo}\label{theorem-bound-path}
Let $\mathcal{B}_{i,j}$=$\lbrace{P'_l(i,j)}\rbrace$ ($l\in[1,r]$, $r=|\mathcal{B}_{i,j}|$) be the set of bounding paths connecting $v_i$ and $v_j$ in subgraph $SG$, and let $P'_u(i,j)$ ($u\in[1,r]$) be the path whose actual distance in $SG$ is the shortest within $\mathcal{B}_{i,j}$. Assuming that the paths in $\mathcal{B}_{i,j}$ are sorted in ascending order based on their bound distances, we can make the following claims: 
\begin{enumerate}
\item If $BD(P'_l(i,j))\leq D(P'_u(i,j))$ and $BD(P'_{l+1}$ $(i,j))$ $\geq$\ $D(P'_u(i,j))$ ($l+1\in[1,r]$), then $P'_u(i,j)$ is the shortest path between $v_i$ and $v_j$ in $SG$. In this case, $P'_u(i,j)$ is also the lower bounding path. 
\item If $BD(P'_r(i,j))<D(P'_u(i,j))$, then $BD(P'_r(i,j))$ must be less than the shortest distance between $v_i$ and $v_j$ in $SG_i$ and $P'_r(i,j)$ is the lower bounding path. 
\end{enumerate}
\end{myTheo}

\begin{proof}
 To show the first claim holds, suppose the shortest path from $v_i$ to $v_j$ is not $P'_u(i,j)$ but another path denoted by $P'_f(i,j)$ that is not covered by $\mathcal{B}_{i,j}$. Based on this assumption, it can be inferred that $\phi(P'_f(i,j))>\phi(P'_r(i,j))$, so $BD(P'_f(i,j))>BD(P'_r(i,j))$. As the actual distance of a path is no less than its bound distance,  we have $D(P'_f(i,j))>BD(P'_r(i,j))\geq BD(P'_{l+1}(i,j))$. Also, because $BD(P'_{l+1}(i,j))>D(P'_u(i,j))$, we have $D(P'_f(i,j))>D(P'_u(i,j))$ which is a contradiction.

The second claim is true if $P'_u(i,j)$ is the shortest path from $v_i$ to $v_j$ in $SG$; otherwise, we infer that the shortest path from $v_i$ to $v_j$ is not in $\mathcal{B}_{i,j}$. In this case, the bound distance of this shortest path must be greater than $BD(P'_r(i,j))$, so its actual distance is also greater than $BD$ $(P'_r(i,j))$.
\vspace{-0.15in}
\end{proof}

\begin{figure}[htbp]
\centering
\subfloat[]{\label{Fig.the.1}
\includegraphics[width=0.2\textwidth]{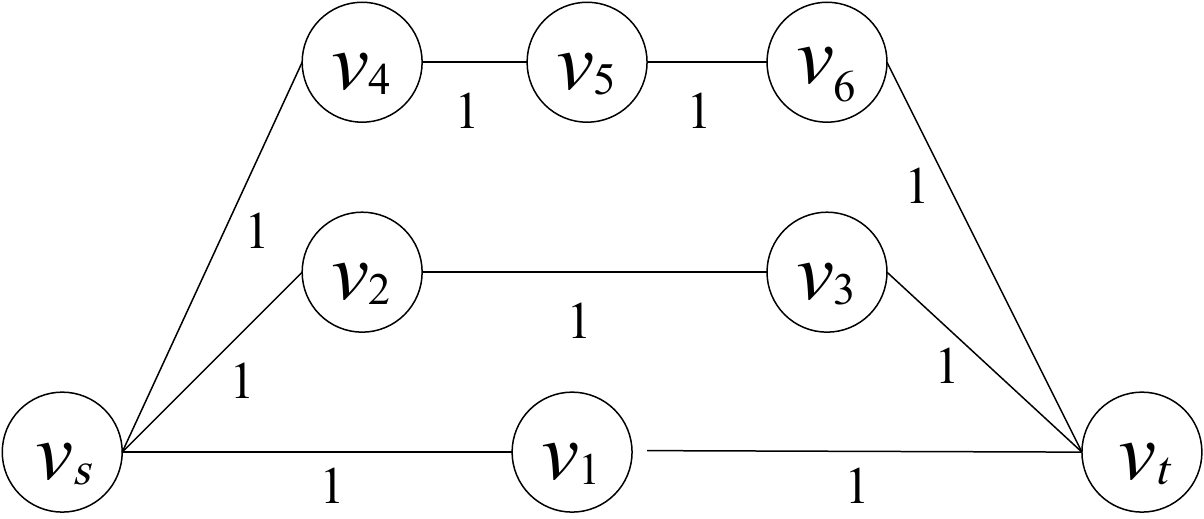}}
 \hspace{0.05in}
\subfloat[]{\label{Fig.the.2}
\includegraphics[width=0.2\textwidth]{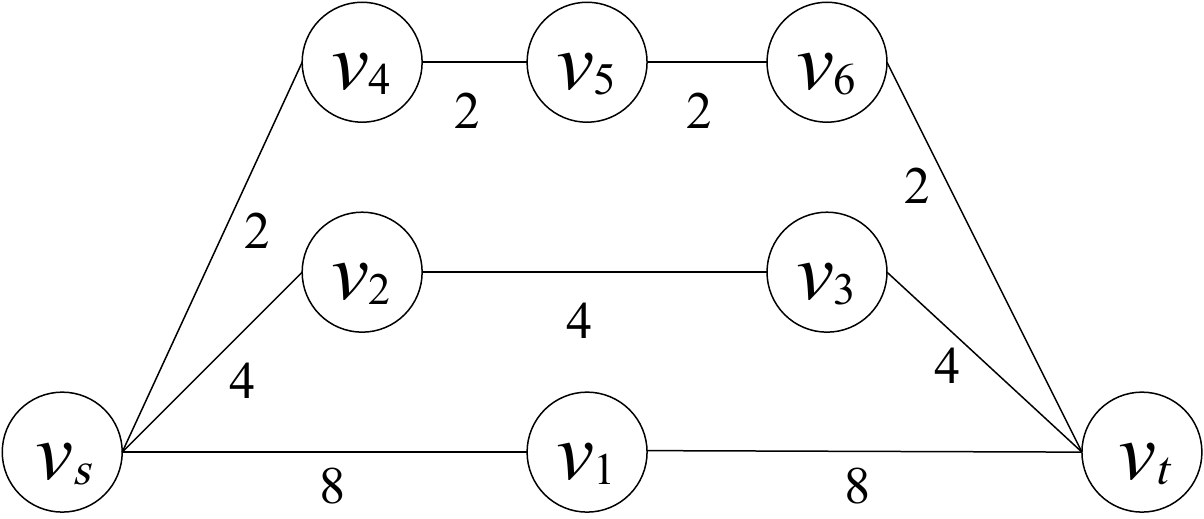}}
\quad
\subfloat[]{\label{Fig.the.3}
\includegraphics[width=0.2\textwidth]{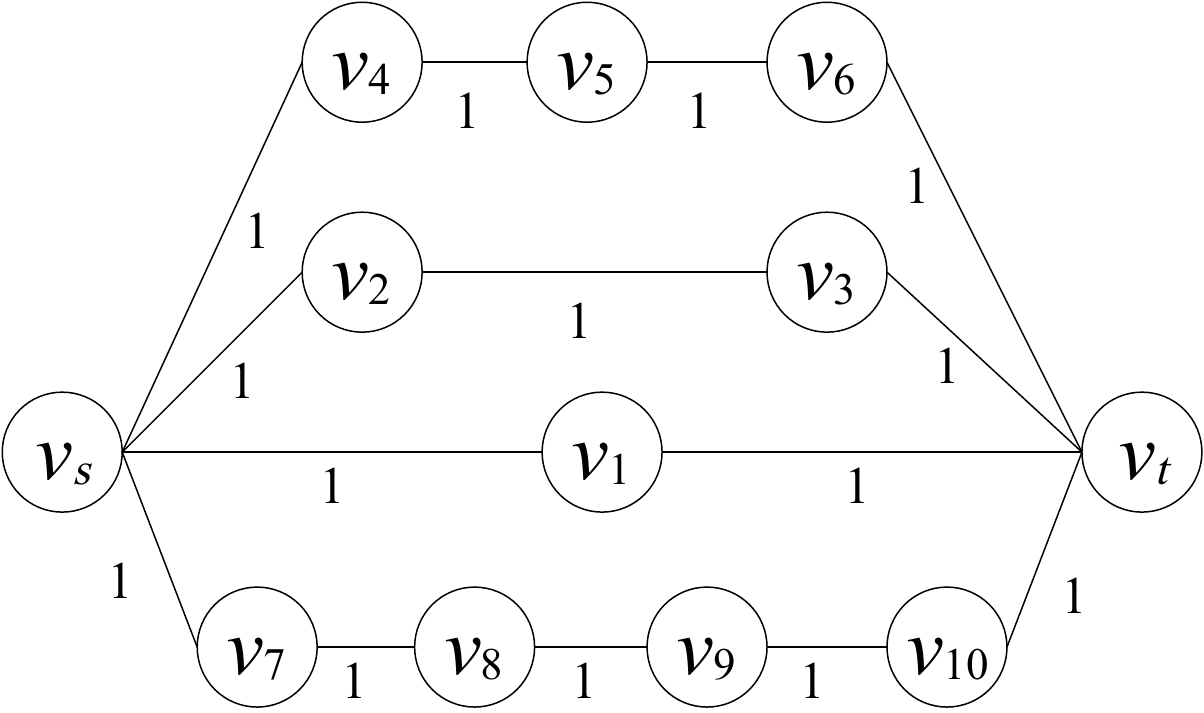}}
 \hspace{0.05in}
\subfloat[]{\label{Fig.the.4}
\includegraphics[width=0.2\textwidth]{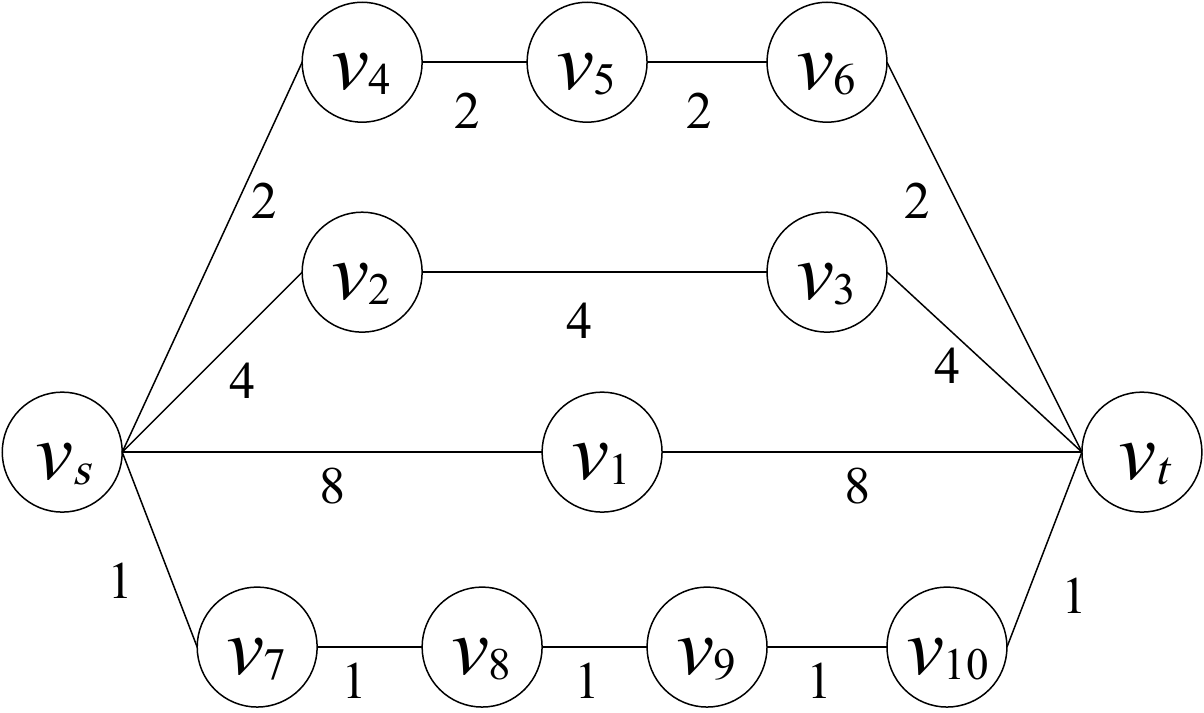}}
\caption{Example for Theorem \ref{theorem-bound-path} \label{fig:theorem-1}}
\vspace{-0.18in}
\end{figure}

\begin{exmp}
  This example illustrates the two cases in Theorem \ref{theorem-bound-path}. For the graph  in Figure~\ref{Fig.the.1}, $P'_1(s,t)$=$\langle v_s, v_1, v_t\rangle$, $P'_2(s,t)$ = $\langle v_s,v_2,v_3,$ $v_t\rangle$, and $P'_3(s,t)$=$\langle v_s,v_4,v_5, v_6, v_t\rangle$ are bounding paths between $v_s$ and $v_t$ when $\xi=3$. If the weights change as shown in Figure~\ref{Fig.the.2}, $P'_3(s,t)$ becomes the bounding path with the shortest distance and the unit weights in this graph are updated to (2, 4), (4, 3), and (8, 2). As $BD(P'_1(s,t))=4$, $BD(P'_2(s,t))=6$, and $BD(P'_3(s,t))=8$, the bound distance of $BD$$(P'_3(s,t))$ equals to its actual distance. Hence, $P'_3(s,t)$ is the shortest path from $v_s$ to $v_t$, which follows the first claim in Theorem \ref{theorem-bound-path}.

  For the scenario shown as Figure~\ref{Fig.the.3}, $P'_1(s,t)$, $P'_2(s,t)$, and $P'_3(s,t)$ are still the bounding paths connecting $v_s$ and $v_t$. When the graph in Figure~\ref{Fig.the.3} changes to the one in Figure~\ref{Fig.the.4}, there are five vfrags with unit weight 1. In this case, $BD(P'_1(s,t))=2$, $BD(P'_2(s,t))=3$, and $BD(P'_3(s,t))=4$. Because $D(P'_3(s,t))>BD(P'_3(s,t))$, it is uncertain whether $P'_3(s,t)$ is the shortest path from $v_s$ to $v_t$, but one can guarantee that the shortest distance between $v_s$ and $v_t$ is greater than $BD(P'_3(s,t))$, conforming to the second claim in Theorem \ref{theorem-bound-path}.
\end{exmp}

Based on Theorem \ref{theorem-bound-path}, we can identify the lower bounding paths and the lower bound distances between any pair of boundary vertices in each subgraph, and then the minimum lower bound distances between these boundary vertices can be determined. 

\begin{figure}[htbp]
\centering
\includegraphics[width=0.2\textwidth]{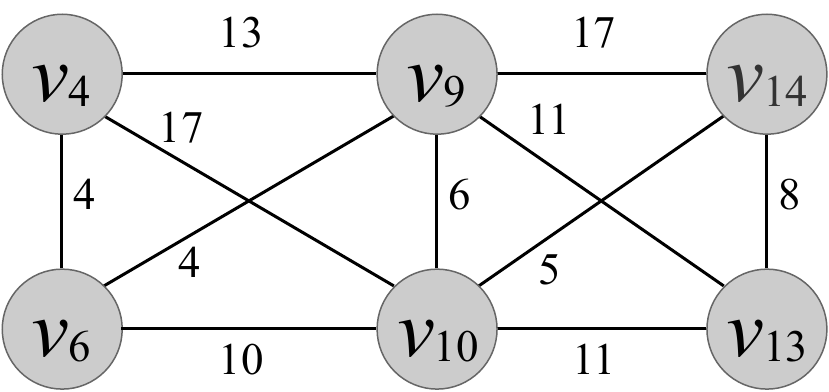}
\caption{Skeleton Graph $G_\lambda$}
\label{fig.skeleton-graph}
\vspace{-0.27in}
\end{figure}
\subsection{Skeleton Graph}\label{sec:skeleton-graph}
The skeleton graph $G_\lambda$ contains all boundary vertices of all subgraphs. Any pair of boundary vertices $v_i$ and $v_j$ within the same subgraph is connected by edge $e'_{i,j}$ with its weight being the minimum lower bound distance between $v_i$ and $v_j$ ($MBD(i,j)$). The rationale behind  introducing the skeleton graph is that KSPs between two vertices in the original graph $G$ possibly pass through the same sequence of boundary vertices as their shortest paths in $G_\lambda$. Thus, $G_\lambda$ can provide an approximate search guideline to identify the different subgraphs in finding KSPs between a pair of vertices in $G$. Figure~\ref{fig.skeleton-graph} shows a skeleton graph $G_\lambda$ corresponding to the graph $G$ in Figure~\ref{global-graph}.

\begin{figure}[htbp]
\centering
\includegraphics[width=0.35\textwidth]{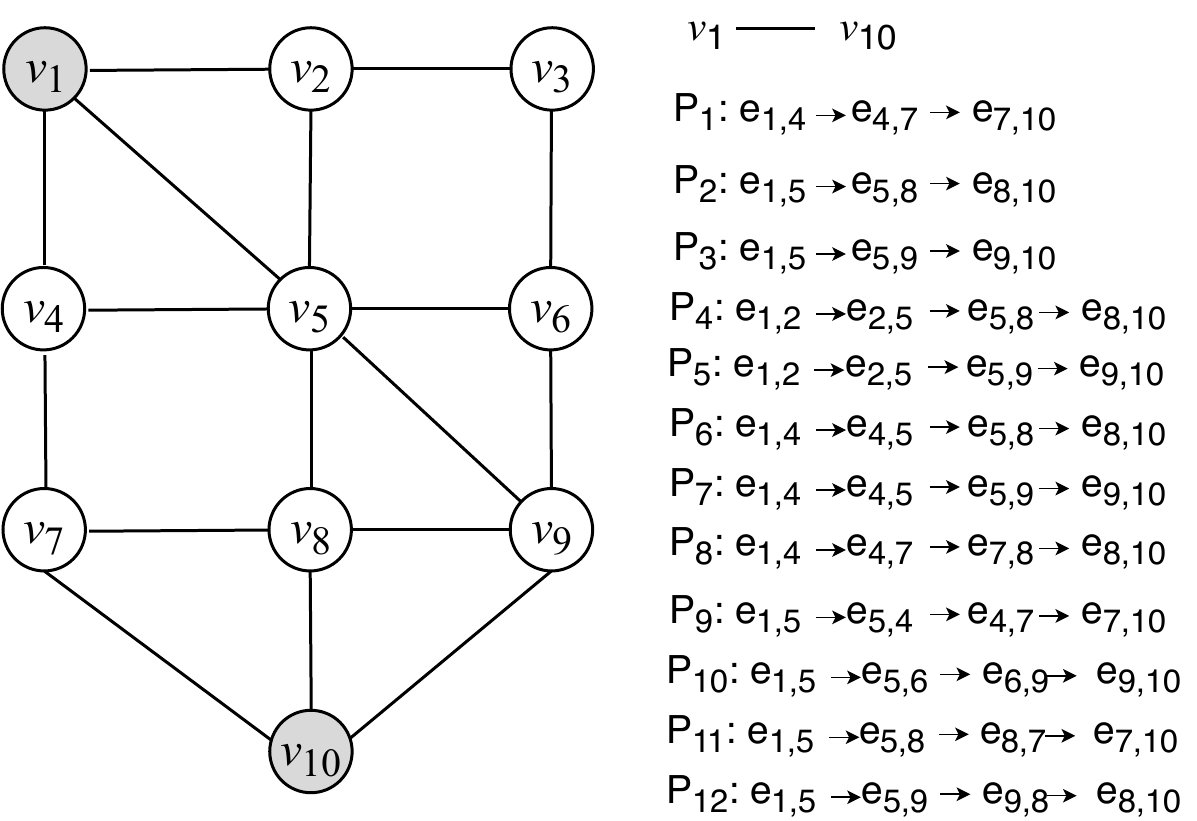}
\caption{Bound paths between boundary vertices}
\label{subgraph-bound-paths}
\end{figure}

\vspace{-1em}
\subsection{Maintenance of Bounding Paths}
In DTLP, the minimum lower bound distances between boundary vertices (i.e., weights of edges in the skeleton graph) computed based on bounding paths are vital to the discovery of KSPs, however they change frequently as edge weights vary. In order to update the minimum lower bound distances efficiently, it is imperative to design a data structure that can effectively manage the large number of bounding paths in each subgraph. For this purpose, we propose an Edge-Path Index (EP-Index for short) that supports maintenance of the bounding paths.

The EP-Index is a map consisting of {\em key}-{\em value} pairs: every key is an edge $e_{i,j}$, and the corresponding value is a list (denoted by $\mathcal{BP}_{i,j}$) where each element consists of a unique bounding path passing through $e_{i,j}$ along with the distance of this path. Fig.\ref{subgraph-bound-paths} gives a subgraph with boundary vertices $v_1$ and $v_{10}$, and the bound paths between $v_1$ and $v_{10}$, where every edge is supposed to have one virtual-fragment. Figure \ref{inverted-EP-index} shows the EP-Index matching the bound paths in Figure \ref{subgraph-bound-paths}.

When the weight of $e_{i,j}$ changes by $\Delta w$, the distance of each bounding path in $\mathcal{BP}_{i,j}$ is updated correspondingly by $\Delta w$. The bound distances of these bounding paths can be easily obtained using the updated unit weights in the corresponding subgraph, after which the lower bound distance between any two boundary vertices in the subgraph can be computed.  Algorithms \ref{al:build-dtlp} and  \ref{al:update-dtlp} show the pseudocode for building and updating the DTLP index respectively.

\begin{algorithm}[htbp]
\begin{small}
\begin{algorithmic}[1]
\REQUIRE
 $\mathcal S$=$\lbrace SG_1, \cdots, SG_n\rbrace$;\\
\ENSURE DTLP (EP-Index and the skeleton graph $G_\lambda$);
\STATE EP-Index=$\phi$, $G_\lambda=\phi$;
\STATE $\mathcal{I}_i = $ the set of boundary vertices of $SG_i$;
\FOR{$SG_i$ : $\mathcal S$}
\FOR{($v_a$,$v_b$): $\mathcal{I}_i$}
\STATE Compute $\mathcal{B}_{a,b}$; //bounding paths between $v_a$ and $v_b$
\STATE Compute the lower bounding distance $LBD(v_a, v_b)$;\\
\STATE Add ($v_a$, $v_b$, $LBD(v_a, v_b)$) into $G_\lambda$;
\STATE Add (($v_a$, $v_b$), $\lbrace PIDs\rbrace$) into $\mathcal{P}_i$; //{\footnotesize $\lbrace PIDs\rbrace$ is IDs of paths in $\mathcal{B}_{a,b}$, $\mathcal{P}_i$ is the set of such tuples associated with $SG_i$}
\STATE Add $\mathcal{B}_{a,b}$ into EP-Index <$e_{i,j}$, $\lbrace P'_l(x,y)\rbrace$>;
\ENDFOR
\ENDFOR
 \caption{Building DTLP}\label{al:build-dtlp}
 \end{algorithmic}
 \end{small}
\end{algorithm}

\begin{algorithm}[htbp]
\begin{small}
\begin{algorithmic}[1]
\REQUIRE
 DTLP, $\Delta w_{i,j}$, $SG_a\in {\mathcal{S}}$; //{\footnotesize $\Delta w_{i,j}$: weight change of {$e_{i,j}$ in $SG_a$}}
\ENSURE Updated DTLP;
\STATE $\mathcal{BP}_{i,j}$=EP-Index.$get(e_{i,j})$; //bounding paths covering $e_{i,j}$;
\FOR{$P'_l(x,y)$ : $\mathcal{BP}_{i,j}$}
\STATE $D(P'_l(x,y))$=$D(P'_l(x,y))$+$\Delta w_{i,j}$;
\STATE Update $BD(P'_l(x,y))$ based on weights in $SG_a$;
\STATE Identify $\mathcal{B}_{x,y}$ based on $\mathcal{P}_i$;
\STATE Update $LBD(v_x, v_y)$ based on $\mathcal{B}_{x,y}$;
\ENDFOR
\STATE Update the skeleton graph $G_\lambda$;
 \caption{Updating DTLP}\label{al:update-dtlp}
 \end{algorithmic}
 \end{small}
\end{algorithm}

{EP-Index} supports the maintenance of bounding paths, however its storage efficiency is less than ideal as the bounding paths are duplicated in the list $\mathcal{BP}_{i,j}$ for different edges. Suppose $N_b$ is the number of boundary vertices in subgraph $SG$, and every bounding path on average has $n_e$ edges. Then the number of elements in the EP-Index is $\frac{N_b\times(N_b-1)}{2}\cdot \xi\cdot n_e$ (where $\xi$ is the number of bounding paths between any pair of boundary vertices in $SG$, as defined in Section~\ref{sec:bound-path}), which implies that the EP-Index is usually much larger than the size of its corresponding subgraph.

\begin{figure}[htbp]
\centering
\includegraphics[width=0.48\textwidth]{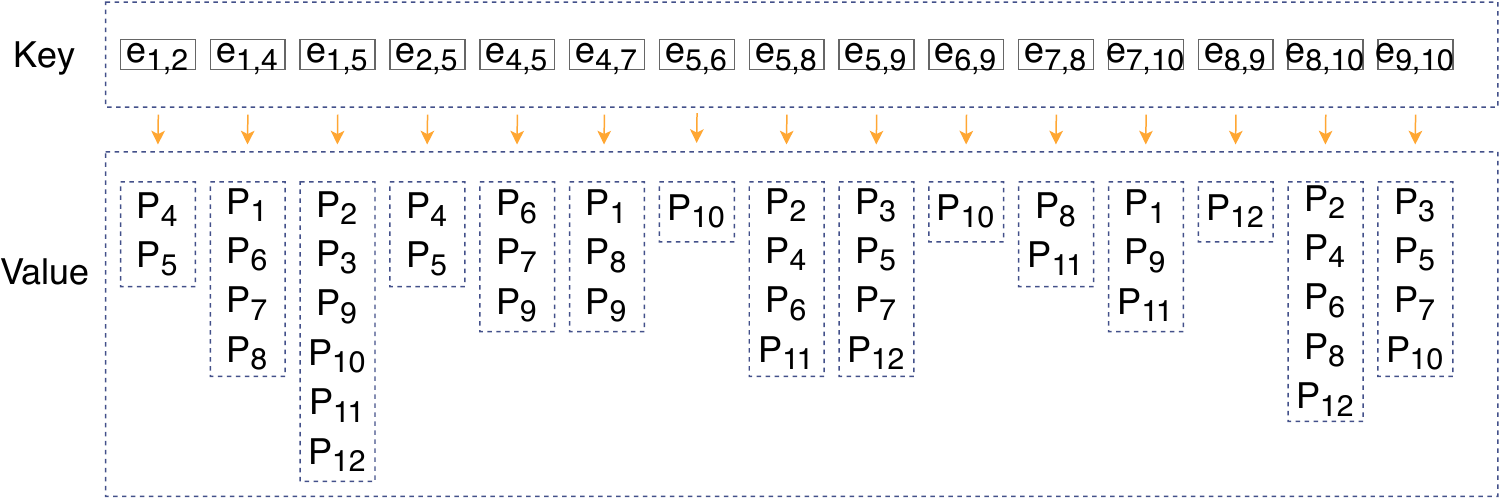}
\caption{EP-Index}
\label{inverted-EP-index}
\end{figure}

\vspace{-1em}
\section{MFP-tree}\label{sec:maintain-dtlp}
Since EP-Index is probably vary large, we devise a compression approach that utilizes locality-sensitive hashing (LSH) to partition the EP-Index into different groups, as well as a modified version of the FP-tree\cite{han2000mining} (called the MFP-tree) to compact bounding paths within each group. This optimization can help compress a large proportion of the duplicate bounding paths in EP-Index. 

\subsection{Partitioning { EP-Index}}
Suppose $\mathcal{P}_i$ and $\mathcal{P}_j$ are two sets of paths corresponding to two different edges. Then their ideal compressing ratio is $\frac{|\mathcal{P}_i\cap \mathcal{P}_j|}{|\mathcal{P}_i\cup \mathcal{P}_j|}$, which is also their Jaccard Similarity. Since the edges in the same group are required to have more common paths, their path sets are expected to have a higher Jaccard similarity. We thus employ LSH to hash the edges with path sets having high Jaccard similarity into the same group.

\begin{itemize}
\item Transferring {EP-Index} into a new matrix called {PE-Matrix}, where each path becomes a row and edges in {EP-Index} are viewed as columns. Figure~\ref{PE-matrix} gives the {PE-Matrix} of {EP-Index} in Fig.\ref{inverted-EP-index}. In {PE-Matrix}, if path $p_i$ is in the path set of $e_{i,j}$, the corresponding position is set to 1 and 0 otherwise.
\item Generating the signature matrix of {PE-Matrix} based on the MinHash strategy. Given two sets $A$ and $B$, MinHash is to estimate their Jaccard similarity quickly, without explicitly computing the intersection and union. In particular, MinHash employs $h$ hash functions to compute the signature for each column (edge) of {PE-Matrix}, which produces a signature matrix denoted by {Sig-Matrix}. This matrix is of size $m\times h$, where $m$ is the number of edges and $h$ is the number of hash functions.

\item Employing LSH to hash columns of {Sig-Matrix} into different groups. First, the rows of the signature matrix is partitioned into $b$ bands. In every band, the sequence of $\frac{h}{b}$ integers in every column is used as a hash number to hash the columns into a separate group array, where $\frac{h}{b}$ is set an integer. The two columns (edges) hashed into the same group are identical in at least one band and are very likely to share more common paths than edges in different groups. 
\end{itemize}

\begin{exmp}
{For the PE-Matrix in Figure~\ref{PE-matrix}, we suppose its signature matrix can be calculated by two hash functions $h_1$ and $h_2$ as shown in Figure~\ref{sig-matrix}, where the edges with the same color can be partitioned into the same group.} 
\end{exmp}

\begin{figure}[htbp]
\centering
\includegraphics[width=0.47\textwidth]{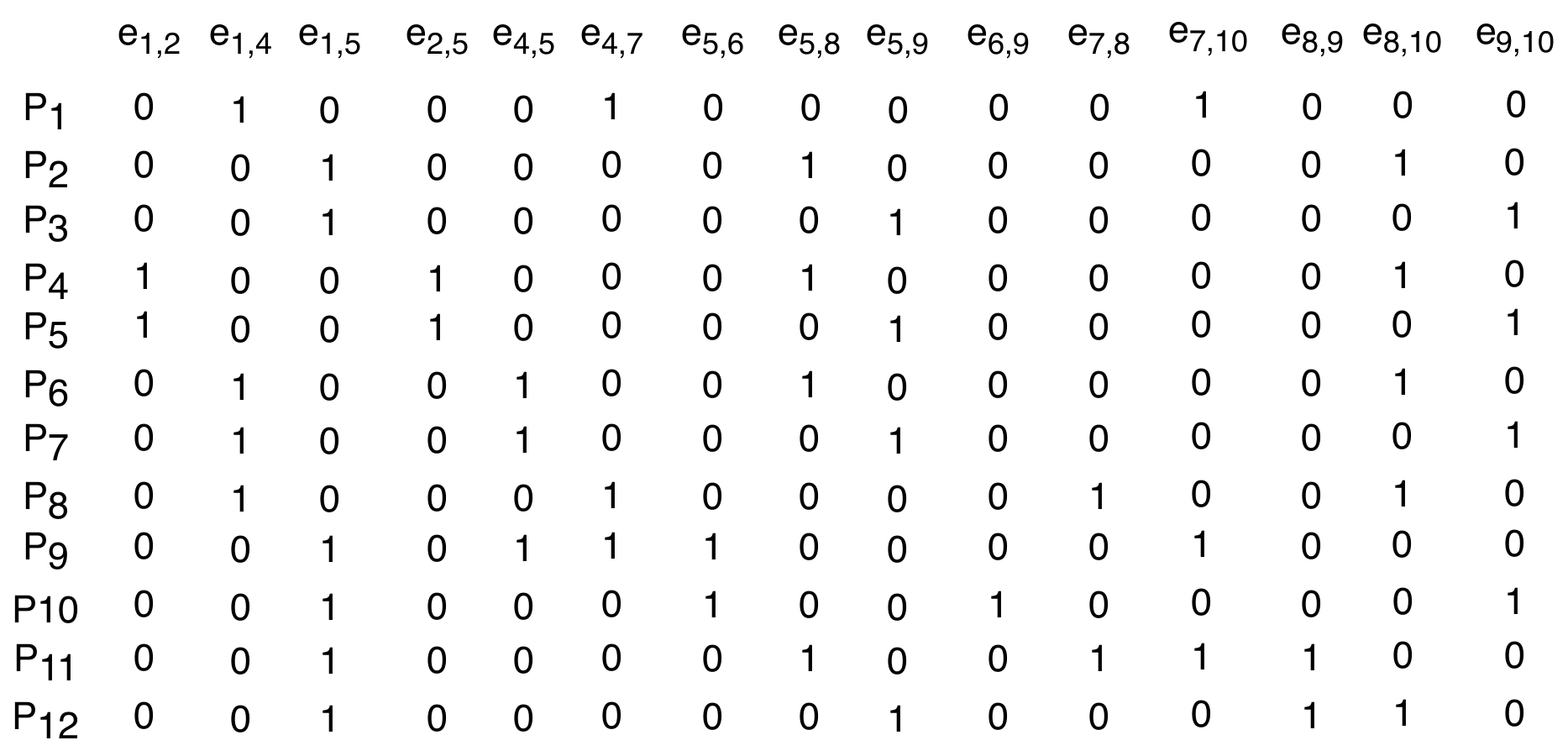}
\caption{{PE-Matrix}}
\label{PE-matrix}
\end{figure}

\vspace{-0.8em}
\begin{figure}[htbp]
\centering
\includegraphics[width=0.47\textwidth]{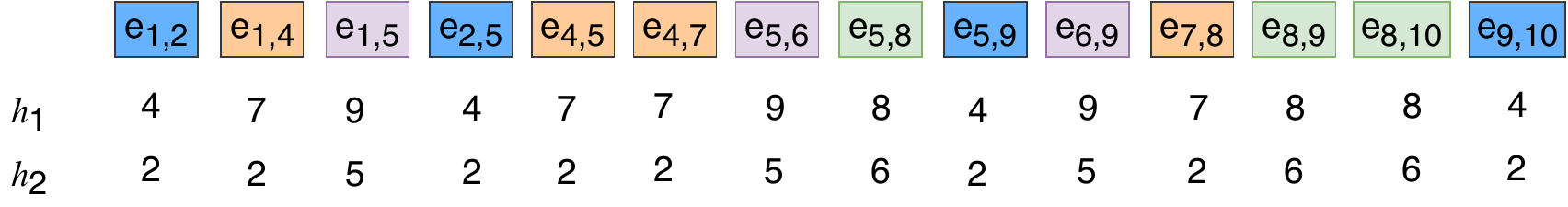}
\caption{Sig-Matrix}
\label{sig-matrix}
\end{figure}

\vspace{-0.5em}
\subsection{Compressing paths with MFP-tree}
Once the path sets are partitioned into different groups, MFP-tree is designed to compress the paths within each group. For any edge $e_{i,j}$ in a group $\mathcal{G}_i$, $\mathcal{P}_{i,j}$ represents the set of paths covering $e_{i,j}$. The paths in $\mathcal{P}_{i,j}$ are sorted based on their number of occurrences in all path sets in descending order.

MFP-tree is initialized to contain only one root node, and then every edge and its path set are consecutively added into the MFP-tree. Suppose the path set $\mathcal{P}_{i,j}$ of edge $e_{i,j}$ has been ranked as $\lbrace p_0, \cdots, p_l\rbrace$, then $\mathcal{P}_{i,j}$ and $e_{i,j}$ will be combined as a sequence of nodes $\mathcal{S}=\lbrace p_0, \cdots, p_l, e_{i,j}\rbrace$, where $p_i$ ($i\in [0, l]$) and $e_{i,j}$ are respectively called the normal node and the tail node. The insertion procedure is as follows.
\begin{itemize}
  \item The longest matching prefix ($pre$) of $\mathcal{S}$ in MFP-tree is first identified. Please note that $pre$ does not need to start from the root but possibly from any node, which is different from the FP-tree.
  \item If $pre$ exists, the remaining part of $\mathcal{S}$ will be directly appended to $pre$. Because $pre$ probably exists in different branches of the MFP-tree, the first being found will be picked. Otherwise, $\mathcal{S}$ will be inserted at the root.
  \item After $\mathcal{S}$ has been added into the MFP-tree, we let the tail node $e_{i,j}$ keeps $|\mathcal{P}_{i,j}|$ (the size of $\mathcal{P}_{i,j}$) that can help us identify the set $\mathcal{P}_{i,j}$ in MFP-tree.
\end{itemize}

\begin{figure}[htbp]
\centering
\subfloat[MFP-tree 1]{\label{Fig.subtree.1}
\includegraphics[width=0.2\textwidth]{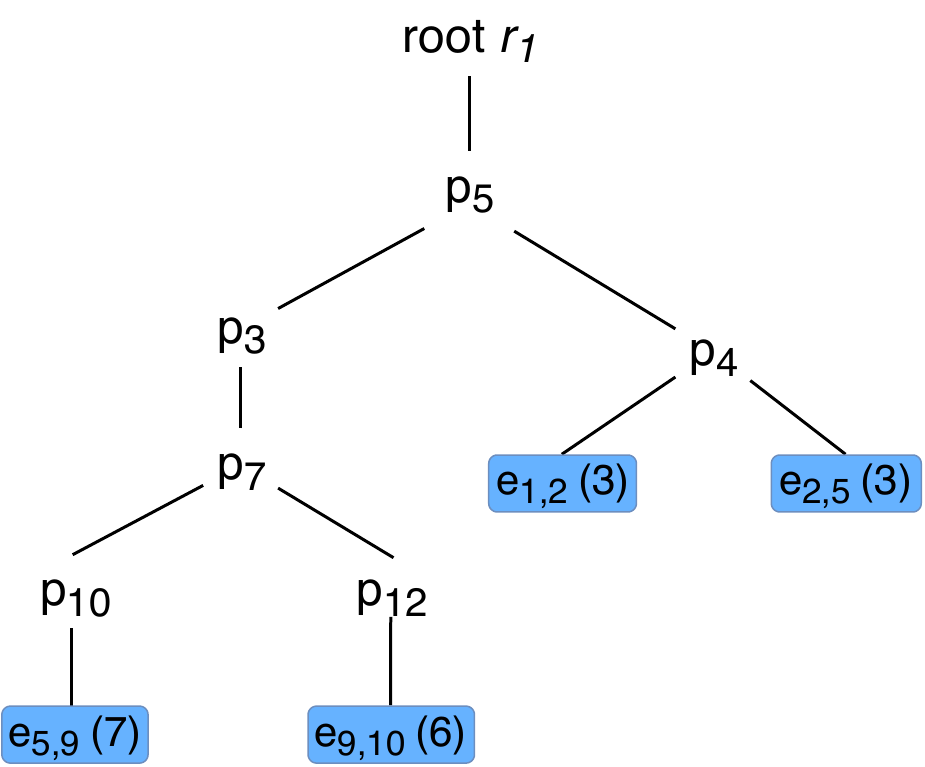}}
\hspace{0.5em}
\subfloat[MFP-tree 2]{\label{Fig.subtree.2}
\includegraphics[width=0.2\textwidth]{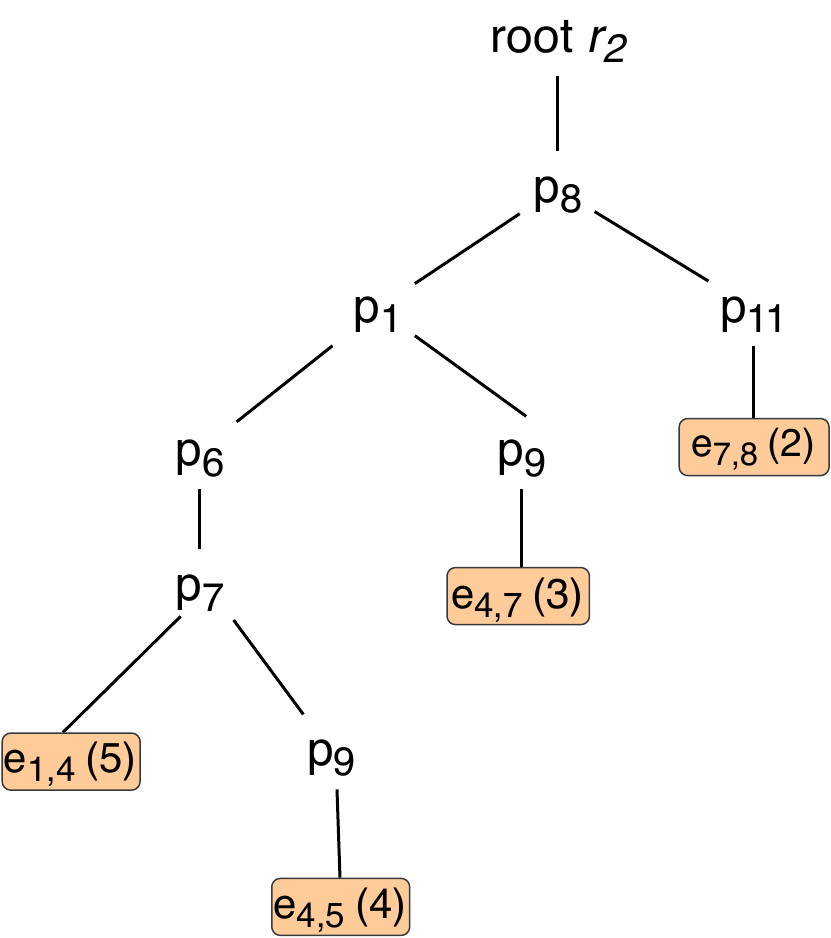}}
\caption{MFP-trees}
\label{Fig-MFP-trees}
\end{figure}

\vspace{-0.5em}
\begin{figure}[htbp]
\centering
\includegraphics[width=0.4\textwidth]{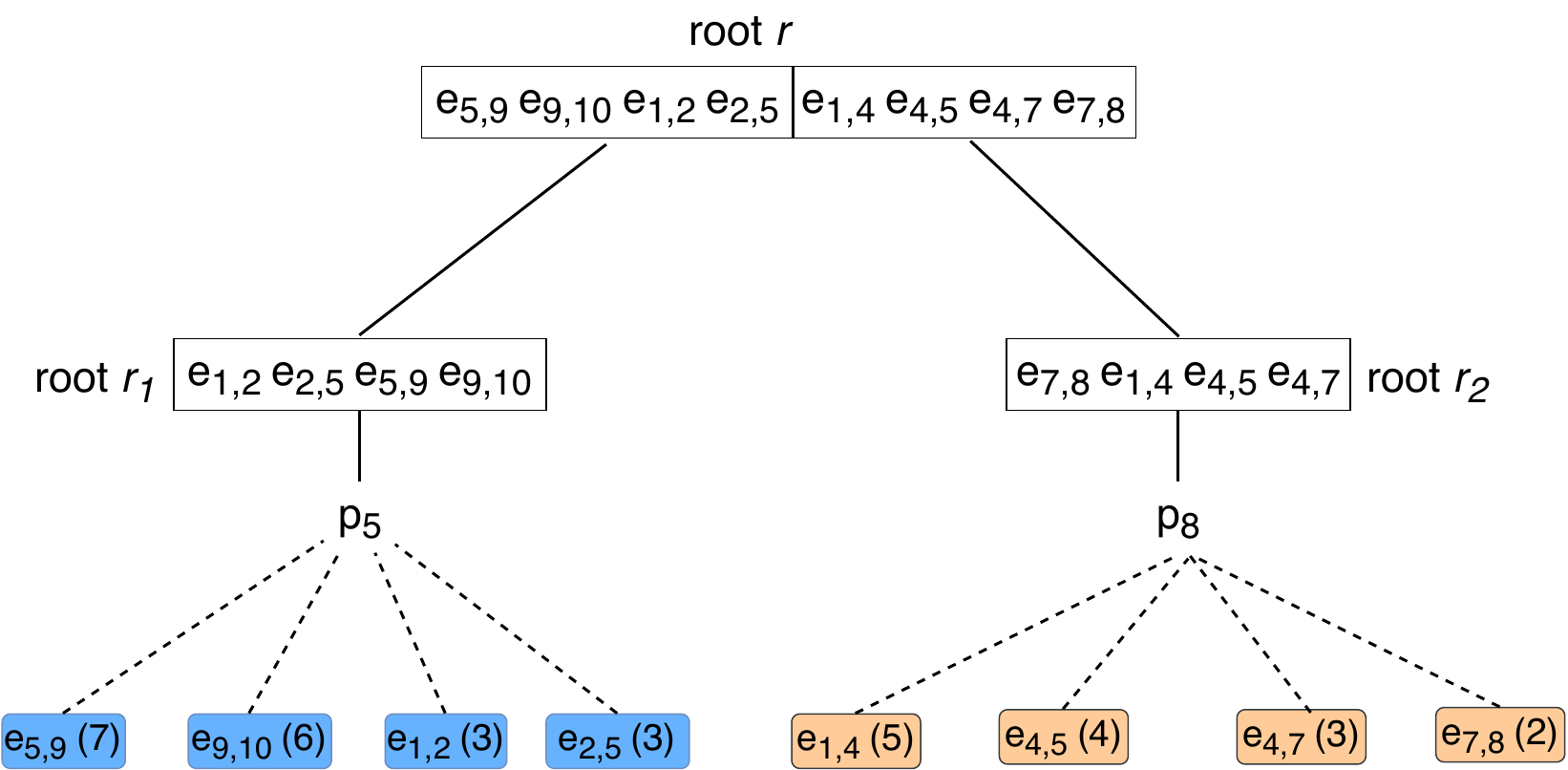}
\caption{Merged MFP-tree $T_e$}
\label{Merged-MFP-tree}
\end{figure}

After all path sets of edges in a group are processed, a matching MFP-tree will be built. In this tree, the root is empty and every leaf is an edge in this group. As a  subgraph probably corresponds to multiple MFP-trees, these MFP-trees will be merged together, denoted by $T_e$ with an empty root $r$. The children of the root $r$ are the roots of merged MFP-trees, which record the edges in the corresponding merged trees.

\begin{exmp}
{ Figure \ref{Fig-MFP-trees} gives two MFP-trees that correspond to the blue and orange groups of edges in Figure \ref{sig-matrix} respectively. Taking the blue group as an example, $e_{5,9}$ should be first inserted into the empty root of $T_1$. Since the path set of $e_{9,10}$ has a prefix ({$P_{33}$, $P_{44}$}) in $T_1$, the remaining part is directly appended to $P_{44}$. Other branches can also be processed in the same fashion. These two MFP-trees are further merged into a bigger MFP-tree shown in Figure~\ref{Merged-MFP-tree}.} 
\end{exmp}

The distances of bound paths can be easily updated in MFP-tree. Suppose the weight of $e_{i,j}$ changes with $\Delta w$, the subtree covering $e_{i,j}$ and the leaf matching $e_{i,j}$ in a MFP-tree can be identified quickly. Since $e_{i,j}$ records the size of $\mathcal{P}_{i,j}$ ($|\mathcal{P}_{i,j}|$), the nodes visited by traversing up $|\mathcal{P}_{i,j}|$ steps are all bound paths containing $e_{i,j}$. For example, the number 4 in $e_{9,10}$ in Figure \ref{Fig-MFP-trees} means four elements containing $e_{9,10}$ can be found by traversing up 4 steps from $e_{9,10}$. Next, we update the real distances of these bound paths with $\Delta w$, and then recompute their lower bounding distances with a constant time cost. 

\section{KSP-DG algorithm} \label{sec:KSP-DG}
In this section, we first discuss the theoretical basis of KSP-DG (Section~\ref{subsec:foundation}), and then present the detailed algorithm (Section~\ref{subsec:KSP-DG}). Finally, we prove its correctness (Section~\ref{subsec:correctness}) and analyze its time complexity (Section~\ref{subsec:complexity}). For clarity, we use ${P_i}(s,t)$ and ${P^{\lambda}_j}(s,t)$ to denote the $i^{th}$ and the $j^{th}$ shortest path in the original graph $G$ and the skeleton graph $G_\lambda$ respectively. For a given query $q(v_s, v_t)$, we assume $v_s$ and $ v_t$ are both boundary vertices in $G$ for ease of presentation, which means both of them are in $G_{\lambda}$. Cases where $v_s$ and $ v_t$ are not boundary vertices in $G$ are discussed in Section~\ref{subsec:nonboundary}. 

\subsection{Underpinnings of {KSP-DG}}\label{subsec:foundation}

For a given query $q(v_s,v_t)$, the basic idea of KSP-DG is to use the shortest paths between $v_s$ and $v_t$ in $G_\lambda$ one by one (in increasing order of distance) as a {\em reference path} to identify the corresponding shortest paths in $G$ that have the same sequence of boundary vertices. This iterative process continues until the KSPs for the query are found. For this to succeed, one key observation is that the path between $v_s$ and $v_t$ in $G_\lambda$ is not longer than the path linking $v_s$ and $v_t$ in $G$ with the same sequence of boundary vertices. which is formally presented as follows.\begin{myLemma}\label{lem:path-border-vertex}
 Given two boundary vertices $v_i$ and $v_j$ in  $G_\lambda$, if the shortest path  between $v_i$ and $v_j$ in $G_\lambda$, ${P^{\lambda}_1}(i,j)$, contains only $v_i$ and $v_j$, then $D({P^{\lambda}_1}(i,j)) \le D({P_1}(i,j))$, where $D({P^{\lambda}_1}(i,j))$ and $D({P_1}(i,j))$ are the shortest distances between $v_i$ to $v_j$ in $G_\lambda$ and $G$ respectively.
\end{myLemma}

\begin{proof}
Assume for the sake of contradiction that $D({P^{\lambda}_1}(i,j))$ $>$ $D({P_1}(i,j))$. By definition, the weight of edge ($v_i$, $v_j$) in $G_\lambda$ is the minimum lower bound distance between $v_i$ and $v_j$, which is not greater than $D({P_1}(i,j))$. Therefore, if $D({P^{\lambda}_1}(i,j))$ $>$ $D({P_1}(i,j))$, then there must exist one or more vertices between $v_i$ and $v_j$ in ${P^{\lambda}_1}(i,j)$, which contradicts the initial assumption.
\end{proof}

\begin{myTheo}\label{the:shotest-path-compare}
 $\forall v_s, v_t \in G_\lambda$, $D({P^{\lambda}_1}(s,t)) \leq D(P_1(s,t))$.
\end{myTheo}
\begin{proof}
 For the shortest path $P_1(s,t)$ between $v_s$ and $v_t$ in $G$,  there must be a corresponding path in graph $G_\lambda$ with the same sequence of boundary vertices as those present in $P_1(s,t)$; let us denote this path by ${P^{\lambda}_f}(s,t)$. For any pair of adjacent vertices $v_i$ and $v_j$ in this sequence of boundary vertices (i.e., there exists an edge $(v_i, v_j)$ in ${P^{\lambda}_f}(s,t)$), it follows from Lemma~\ref{lem:path-border-vertex} that the distance between $v_i$ and $v_j$ in ${P^{\lambda}_f}(i,j)$ cannot be greater than the distance connecting $v_i$ and $v_{j}$ in $P_1(i,j)$. Hence,  $D({P^{\lambda}_f}(s,t))\leq D(P_1(s,t))$. If ${P^{\lambda}_f}(s,t)$ is the shortest path between $v_s$ and $v_t$ in $G^{\lambda}$, then we immediately have $D({P^{\lambda}_1}(s,t)) \leq D(P_1(s,t))$; otherwise, there must exist a shortest path ${P^{\lambda}_1}(s,t)$ such that $D({P^{\lambda}_1}(s,t))\leq D({P^{\lambda}_f}(s,t))$. As $D({P^{\lambda}_f}(s,t))\leq D(P_1(s,t))$, we have $D({P^{\lambda}_1}(s,t))$ $\leq$ $ D(P_1(s,t))$.
\end{proof}

\subsection{KSP-DG Algorithm in Detail}\label{subsec:KSP-DG}
KSP-DG is designed to run in a distributed cluster with the {\em master}-{\em worker} model. The master node maintains the original graph $G$, receives the weight updates for edges in $G$, and serves as the entry point of KSP queries. DTLP, as the fundamental index structure, is distributed across workers. In particular, the subgraphs resulting from partitioning $G$ are allocated to different workers on a many-to-one basis based on their load. Each worker maintains the assigned subgraphs as well as the bounding paths between boundary vertices (i.e., the first level of the DTLP index) within each subgraph. Moreover, a copy of the skeleton graph $G_\lambda$ (the second level of the DTLP index) is also kept on each worker. 

Using the DLTP index, KSP-DG adopts a filter-and-refine strategy to iteratively find KSPs for the query $q(v_s,v_t)$. Each iteration of KSP-DG consists of two steps: \textit{filter} and \textit{refine}. Without loss of generality, we describe the two steps for the $i^{th}$ iteration.

 In the filter step, we compute the $i^{th}$ shortest path between $v_s$ and $v_t$ in the skeleton graph $G_\lambda$, and use it as the reference path. This path corresponds to a sequence of boundary vertices, denoted by $\mathcal{R}$, between $v_s$ and $v_t$ in $G$. This step can be executed on one of the workers responsible for processing this query.

In the refine step, we aim to compute the corresponding {\em k} shortest paths connecting $v_s$ and $v_t$ in $G$ that traverse the same sequence of boundary vertices as those present in the reference path. Since any two adjacent boundary vertices in $\mathcal{R}$ must belong to the same subgraph, we can identify partial {\em k} shortest paths for each pair of adjacent boundary vertices in $\mathcal{R}$ from the corresponding subgraphs individually. This operation can be carried out by the workers maintaining those respective subgraphs in parallel. Next, the generated partial {\em k} shortest paths are reported back by these workers to the worker responsible for this query, which will then merge the partial {\em k} shortest paths received to form $k$ complete shortest paths, which all share the same sequence of boundary vertices as the reference path.  

In each iteration, the generated {\em k} shortest paths corresponding to the reference path, called the {\em candidate KSPs}, are used to update a list $\mathcal{L}$  of the shortest paths that have been obtained so far. $\mathcal{L}$ keeps only {\em k} shortest paths found so far in ascending order of distance. After using the generated candidate KSPs to update $\mathcal{L}$ in the $i^{th}$ iteration, if the distance of the $k^{th}$ path in $\mathcal{L}$ is not greater than that of the reference path generated in the $(i+1)^{th}$ iteration, the algorithm terminates, and the paths in $\mathcal{L}$ are the final answer, i.e., the KSPs between $v_s$ and $v_t$ in $G$. Otherwise, the iteration continues.

Algorithm \ref{al:KSP-DG} gives the pseudo-code of {KSP-DG}. It first initializes the parameters in Line 1, and then executes the filter and refine steps in each iteration in Lines 2-12, where $P^{\lambda}_i(s,t)$ denotes the $i^{th}$ reference path. The function {\em candidateKSP} is to identify the candidate KSPs for a given reference path, and its pseudo-code is separately shown in Algorithm~\ref{al:canKSP}. Line 6 in Algorithm~\ref{al:canKSP} uses Yen's Algorithm \cite{yen1971finding} to compute the $k$-shortest paths in subgraph SG between the $j^{th}$ and $(j+1)^{th}$ vertices of the reference path. The cost of running {\em candidateKSP} dominates the total cost of each iteration in KSP-DG, and it can be further optimized. In particular, two neighboring reference paths $P^{\lambda}_i(s,t)$ and $P^{\lambda}_{i+1}(s,t)$ usually share many common pairs of boundary vertices. Once the partial {\em k} shortest paths between these pairs of boundary vertices are identified while dealing with $P^{\lambda}_i(s,t)$, we do not need to compute them again for $P^{\lambda}_{i+1}(s,t)$, which can often assist in reducing the cost of {\em candidateKSP} significantly.

\begin{algorithm}[htbp]
\begin{algorithmic}[1]
\REQUIRE
 $G_\lambda$, $q (v_s, v_t)$, $\mathcal{S}$=$\lbrace{SG}_{1}$, $\cdots$, $SG_{n}\rbrace$;\\
\ENSURE
KSPs from $v_s$ to $v_t$ in $G$;
\STATE $\mathcal{L}=\phi$; $Dist=\infty$; $i=1$;
 \WHILE{$\mathcal{L}=\phi$ $||$ $Dist\leq D({P^{\lambda}_{i+1}}(s,t))$}
  \STATE ${\mathcal{C}}= ${\em candidateKSP} ($\mathcal{S}$, ${P^{\lambda}_i}(s,t)$); 
  \STATE Add ${\mathcal{C}}$ into ${\mathcal{L}}$;\\
  \IF{$|{\mathcal{L}}|>k$}
   \STATE Keep the {\em k} shortest paths in $\mathcal{L}$ and remove others;
   \ENDIF
  \STATE {\em Dist} = the distance of the $k^{th}$ path in $\mathcal{L}$;
  \STATE $i++$;
 \ENDWHILE
 \STATE return $\mathcal{L}$;
 \caption{KSP-DG}\label{al:KSP-DG}
 \end{algorithmic}
\end{algorithm}

\begin{algorithm}[htbp]
\begin{algorithmic}[1]
\REQUIRE
 $\mathcal{S}$=$\lbrace{SG}_{1}$, $\cdots$, $SG_{n}\rbrace$, $P^{\lambda}_i(s,t)$;\\
\ENSURE
Candidate KSPs from $v_s$ to $v_t$ in $G$;
\STATE Set $\mathcal{C}=\phi$; Set $\mathcal{Y}=\phi$; $j$=1;
 \WHILE{$j< |P^{\lambda}_i(s,t)|$} 
  \STATE Identify $\mathcal{U}$, the set of subgraphs containing the $j^{th}$ and $(j+1)^{th}$ vertices in $P^{\lambda}_i(s,t)$, $\texttt{v}_j$ and $\texttt{v}_{j+1}$; 
  \STATE $\mathcal{Y}=\phi$;
  \FORALL{subgraph $SG$ $\in$ $\mathcal{U}$}
  \STATE $\mathcal{Y}$=$\mathcal{Y}$ $\cup$ {\em Yen}($\texttt{v}_j$,$\texttt{v}_{j+1}$,$SG$);
  \ENDFOR
  \STATE Keep only $k$ shortest paths in $\mathcal{Y}$;
  \STATE $\mathcal{C}$=$\mathcal{C}$ $\Join$ $\mathcal{Y}$;
  \STATE Keep only $k$ shortest paths in $\mathcal{C}$;
  \STATE $j++$;
 \ENDWHILE
 \STATE Return $\mathcal{C}$;
 \caption{\em candidateKSP}\label{al:canKSP}
\end{algorithmic}
\end{algorithm}

\vspace{-0.5em}
\begin{exmp}
Suppose $v_4$ and $v_{13}$ in Figure~\ref{global-graph} are the source and destination vertices respectively, and $k=2$. In the first iteration, KSP-DG identifies the first reference path from $v_4$ to $v_{13}$ in $G_\lambda$ (shown as Figure~\ref{fig.skeleton-graph}) as ${P^{\lambda}_1(4,13)}=\langle v_4, v_6, v_9, v_{13}\rangle$ with distance 19. Next, KSP-DG computes $k=2$ shortest paths between any two adjacent boundary vertices as shown in the following table, where the third column shows the subgraphs involved. Then, {KSP-DG} joins the partial shortest paths to generate $k=2$ candidate shortest paths from $v_4$ to $v_{13}$, denoted by ${P_{1}}= \langle v_4, v_6, v_9, v_{11}, v_{12}, v_{13}\rangle$ with distance 19 and ${P_{2}}=\langle v_4, v_5, v_6, v_9, v_{11}, v_{12}, v_{13}\rangle$ with distance 25. So, $\mathcal{L}=\lbrace {P_{1}}, {P_{2}}\rbrace$.

\vspace{-0.05in}
\smallskip\noindent
\resizebox{\linewidth}{!}{
\begin{tabular}{|c|c|c|}
\hline
adjacent boundary vertices & partial shortest paths & involved subgraphs\\
\hline
$(v_4$, $v_6)$& $\langle v_4, v_5, v_6\rangle$, $\langle v_4, v_6\rangle$ & $SG_1$, $SG_2$\\
\hline
$(v_6$, $v_9)$& $\langle v_6, v_9\rangle$ & $SG_2$ \\
\hline
$(v_9$, $v_{13})$& \makecell[cl]{$\langle v_9, v_{11},  v_{12}, v_{13}\rangle$\\ $\langle v_9,  v_{11}, v_{10}, v_{14}, v_{13}\rangle$}& {$SG_3$}\\
\hline
\end{tabular}
}
\smallskip

Since the second reference path is ${P^{\lambda}_2}(4,13)$=$\langle v_4, v_9, v_{13}\rangle$ with distance 22, and $D(P_2)>D({P^{\lambda}_2}(4, 13))$, KSP-DG continues onto the second iteration, where it calculates candidate KSPs w.r.t. ${P^{\lambda}_2}(4, 13)$, denoted by $P_{3}$=$\langle v_4, v_7, v_8, v_9, v_{11}, v_{12}, v_{13}\rangle$ with distance 22, and $P_{4}$=$\langle v_4, v_7, v_8, v_9, v_{11},$ $v_{10}, v_{14}, v_{13}\rangle$ with distance 34. Then $\mathcal{L}$ is updated to $\lbrace P_{1}, P_{3}\rbrace$. After identifying the third reference path $P^{\lambda}_3(4, 13)$=$\langle v_4, v_6, v_{10}, v_{13}\rangle$ with distance 25, it is safe for KSP-DG to conclude that $P_{1}$ and $P_{3}$ are the two shortest paths from $v_4$ to $v_{13}$, as $D(P_{3})<Dist(P^{\lambda}_{3}(4,13))$.
\end{exmp}

{\subsection{Optimization on computing partial KSPs}
In the refine step of KSP-DG, a basic operation extensively conducted is to find {KSP}s for a pair of boundary vertices on a scratch of a subgraph. Yen's algorithm is designed to find {KSP}s, but its low efficiency seriously restrict the performance of KSP-DG. We thus propose a Progressive Yen's (PYen) algorithm to speed up the identification of partial KSPs between boundary vertices within each subgraph.

\subsubsection{Yen's algorithm}
Yen's algorithm is inefficiency due to its centralized search paradigm. In particular, for given source and terminal vertices $v_s$ and $v_t$, it iteratively identifies the $(i+1)^{th}$ shortst path $P_{i+1}(s,t)$  based on the $i^{th}$ shortst path $P_{i}(s,t)$  ($1\leq i<k$). In $i^{th}$ iteration, it sequentially computes the deviation path corresponding to each vertex $v_l$ (except for $v_t$) in $P_{i}(s,t)$, where the deviation path corresponding to $v_l$ consists of two components: a base partial path following $P_{i}(s,t)$ from $v_s$ to $v_l$ and a spur partial path from $v_l$ to $v_t$ that does not contain $e_{l,l+1}$, the edge connecting $v_l$ and $v_{l+1}$ that are two adjacent vertices in ${P_i}(s,t)$. The spur partial path can be computed by searching the shortest path from $v_l$ to $v_t$ in $G'$ on the condition that the weight of the edge $e_{l,l+1}$ is set infinity. Once the spur path is determined, the weight of $e_{l,l+1}$ is reverted to the original value. When all deviation paths of ${P_i}(s,t)$ have been obtained, these deviation paths combining with the previously identified deviation paths of other shortest paths form a candidate deviation path set. From the set, Yen's algorithm picks up the deviation path with minimum weight as the next shortest path $P_{i+1}(s,t)$. Since each iteration only produces one new shortest path, Yen's algorithm needs to conduct $k$ iterations to find KSPs.

\subsubsection{PYen algorithm}
PYen algorithm optimizes Yen's algorithm by parallelizing the deviation paths for the shortest path, early pruning the unnecessary deviation paths, and reusing the identified spur partial paths.


{\bf Identifying deviation paths in parallel.} To speed up the identification of deviation paths of the $i^{th}$ shortest path $P_{i}(s,t)$ ($1\leq i< k$), PYen algorithm initializes an individual instance to compute the deviation path corresponding to each deviation vertex in $P_{i}(s,t)$. Let $n_p=|P_{i}(s,t)|$ denotes the number of vertices in $P_{i}(s,t)$, then there exist $n_p-1$ deviation paths to be generated for $P_{i}(s,t)$. Compared with $\sum_{l=1}^{|P_{i}(s,t)|}T_l$, the cost of computing deviation paths by Yen's algorithm , the parallization of PYen algorithm decreases the cost to $Max\lbrace T_l\rbrace$ ($l\in\lbrace 1,\cdots, |P_{i}(s,t)|-1\rbrace$), where $T_l$ represents the cost of computing the deviation path corresponding to a deviation vertex $v_l$. 

{\bf Early pruning unnecessary deviation paths.} In the above process, PYen algorithm introduces a pruning strategy to avoid generating the redundant unnecessary deviation paths. Suppose $P_i(s,t)$ ($1\leq i< k$) shortest path has been identified, there exist remaining $(k-i)$ paths to be identified. At this moment, not all deviation paths of $P_i(s,t)$ but only the $(k-i)$ shortest ones can be the potential $P_{i+1}(s,t)$ and need to be identified for computing $P_{i+1}(s,t)$. For this objective, PYen uses a priority queue $\mathcal{Q}$ with size $(k-i)$ to keep the current shortest $(k-i)$ deviation paths of $P_i(s,t)$ identified so far. For any deviation vertex $v_l$ whose corresponding deviation path has not been identified, the corresponding search instance will detect whether the {\em stop condition} is satisfied to early terminate exploring the spur partial path from $v_l$ to $v_t$ when traversing each vertex, where the {\em stop condition} refers to the sum of lengths of the base partial path from $v_s$ to $v_l$ and the current spur partial path starting from $v_l$ being explored is not smaller than the length of the $(k-i)^{th}$ deviation path in $\mathcal{Q}$. According to the stop condition, plenty of deviation paths can be early pruned before completely identified.

{\bf Reusing identified spur partial paths.} To avoid repetitive generation of the spur partial paths for the same deviation vertices in the different shortest paths, PYen algorithm indexes the identified spur partial paths from different deviation vertices to the terminal vertex. It encapsulates a deviation vertex and its corresponding spur partial path as a pair of $\langle key, value\rangle$, where the keys (i.e., deviation vertices) are ranked in ascending order with respect to the IDs. For a deviation vertex to be processed, PYen algorithm can first efficiently detects whether the deviation vertex existed in the index, and then decides to if initialize a search instance to compute the corresponding spur partial path. In this way, PYen algorithm can avoid repetitively producing the same spur partial paths when computing KSPs between two boundary vertices. Note that the indexed spur partial paths are not permanently kept but discarded once the KSPs between the boundary vertices are identified, because these spur partial paths will become invalid as edge weights evolve and are no longer valuable for the future generation of deviation paths in later iterations of KSP-DG algorithm.   
}

{\subsection{Optimization on computing partial KSPs}
In the refine step of KSP-DG, a basic operation extensively conducted is to find {KSP}s for a pair of boundary vertices on a scratch of a subgraph. Yen's algorithm is designed to find {KSP}s, but its low efficiency cannot satisfy the requirement of  KSP-DG on the high performance. We thus propose a Progressive Yen's (PYen) algorithm to speed up the identification of partial KSPs between boundary vertices within each subgraph.

\subsubsection{Yen's algorithm}
Yen's algorithm is inefficiency due to its centralized search paradigm. In particular, for given source and terminal vertices $v_s$ and $v_t$, it iteratively identifies the $(i+1)^{th}$ shortst path $P_{i+1}(s,t)$  based on the $i^{th}$ shortst path $P_{i}(s,t)$  ($1\leq i<k$). In $i^{th}$ iteration, it sequentially computes the deviation path corresponding to each vertex $v_l$ (except for $v_t$) in $P_{i}(s,t)$, where the deviation path corresponding to $v_l$ consists of two components: a base partial path following $P_{i}(s,t)$ from $v_s$ to $v_l$ and a spur partial path from $v_l$ to $v_t$ that does not contain $e_{l,l+1}$, the edge connecting $v_l$ and $v_{l+1}$ that are two adjacent vertices in ${P_i}(s,t)$. The spur partial path can be computed by searching the shortest path from $v_l$ to $v_t$ in $G'$ on the condition that the weight of the edge $e_{l,l+1}$ is set infinity. Once the spur path is determined, the weight of $e_{l,l+1}$ is reverted to the original value. When all deviation paths of ${P_i}(s,t)$ have been obtained, these deviation paths combining with the previously identified deviation paths of other shortest paths form a candidate deviation path set. From the set, Yen's algorithm picks up the deviation path with minimum weight as the next shortest path $P_{i+1}(s,t)$. Since each iteration only produces one new shortest path, Yen's algorithm needs to conduct $k$ iterations to find KSPs.

\subsubsection{PYen algorithm}
PYen algorithm optimizes Yen's algorithm from the following aspects.


{\bf Identifying deviation paths in parallel.} To speed up the identification of deviation paths of the $i^{th}$ shortest path $P_{i}(s,t)$ ($1\leq i< k$), PYen algorithm initializes an individual instance to compute the deviation path corresponding to each deviation vertex in $P_{i}(s,t)$. Let $n_p=|P_{i}(s,t)|$ denotes the number of vertices in $P_{i}(s,t)$, then there exist $n_p-1$ deviation paths to be generated for $P_{i}(s,t)$. Compared with $\sum_{l=1}^{|P_{i}(s,t)|}T_l$, the cost of computing deviation paths by Yen's algorithm , the parallization of PYen algorithm decreases the cost to $Max\lbrace T_l\rbrace$ ($l\in\lbrace 1,\cdots, |P_{i}(s,t)|-1\rbrace$), where $T_l$ represents the cost of computing the deviation path corresponding to a deviation vertex $v_l$. 

In the above process, PYen algorithm introduces a pruning strategy to avoid generating the redundant unnecessary deviation paths. Suppose $P_i(s,t)$ ($1\leq i< k$) shortest path has been identified, there exist remaining $(k-i)$ paths to be identified. At this moment, not all deviation paths of $P_i(s,t)$ but only the $(k-i)$ shortest ones can be the potential $P_{i+1}(s,t)$ and need to be identified for computing $P_{i+1}(s,t)$. For this objective, PYen uses a priority queue $\mathcal{Q}$ with size $(k-i)$ to keep the current shortest $(k-i)$ deviation paths of $P_i(s,t)$ identified so far. For any deviation vertex $v_l$ whose corresponding deviation path has not been identified, the corresponding search instance will detect whether the {\em stop condition} is satisfied to early terminate exploring the spur partial path from $v_l$ to $v_t$ when traversing each vertex, where the {\em stop condition} refers to the sum of lengths of the base partial path from $v_s$ to $v_l$ and the current spur partial path starting from $v_l$ being explored is not smaller than the length of the $(k-i)^{th}$ deviation path in $\mathcal{Q}$. According to the stop condition, plenty of deviation paths can be early pruned before completely identified.

Moreover, to avoid repetitive generation of the spur partial paths for the same deviation vertices in the different shortest paths, PYen algorithm indexes the identified spur partial paths from different deviation vertices to the terminal vertex. It encapsulates a deviation vertex and its corresponding spur partial path as a pair of $\langle key, value\rangle$, where the keys (i.e., deviation vertices) are ranked in ascending order with respect to the IDs. For a deviation vertex to be processed, PYen algorithm can first efficiently detects whether the deviation vertex existed in the index, and then decides to if initialize a search instance to compute the corresponding spur partial path. In this way, PYen algorithm can avoid repetitively producing the same spur partial paths when computing KSPs between two boundary vertices. Note that the indexed spur partial paths are not permanently kept but discarded once the KSPs between the boundary vertices are identified, because these spur partial paths will become invalid as edge weights evolve and are no longer valuable for the future generation of deviation paths in later iterations of KSP-DG algorithm.   
}

\subsection{Discussions}\label{subsec:nonboundary}
{\bf Non-boundary Vertices as Source or Destination.} In the preceding discussions, we have assumed that both $v_s$ and $v_t$ in the given query $q(v_s, v_t)$ are boundary vertices. We now show how {KSP-DG} works when $v_s$ and/or $v_t$ are not boundary vertices. Without loss of generality, we consider the case where $v_s$ in subgraph $SG_x$ and $v_t$ in subgraph $SG_y$ are not boundary vertices. To tackle this case, we first make a tactical maneuver by treating $v_s$ and $v_t$ as boundary vertices and adding them into the skeleton graph $G_\lambda$. In particular, we connect $v_s$ with an edge to every boundary vertex $v_i$ in $SG_x$, and the weight of the edge is set to be the minimum lower bound distance between $v_s$ and $v_i$; $v_t$ is processed in the same fashion. Next, we add $v_s$ and $v_t$ to $G_\lambda$, followed by the process of KSP-DG shown in Algorithm~\ref{al:KSP-DG}, just as if $v_s$ and $v_t$ were boundary vertices.

{\bf Finding KSPs in directed graphs.} 
With some minor modification to the DTLP index, KSP-DG can be used to support KSP queries in directed graphs as well. In particular, instead of computing one lower bounding path for a pair of boundary vertices $v_i$ and $v_j$, we maintain two bounding paths (and their corresponding distances), one for each direction (from $v_i$ to $v_j$ or the opposite). Accordingly, the skeleton graph is now a directed graph with two edges between each pair of adjacent vertices, one for each direction. KSP-DG can then be run based on this modified DTLP index to find KSPs.

\subsection{Correctness of KSP-DG}\label{subsec:correctness}
The {KSP-DG} algorithm is provably correct, as shown below.

\begin{myLemma}\label{lem-oriented}
Let ${P^{\lambda}_i}(s,t)$ be the $i^{th}$ reference path from $v_s$ to $v_t$ and $\mathcal{C} _i$ be the set of candidate KSPs w.r.t. ${P^{\lambda}_i}(s,t)$. We have $\forall P_i(s,t)\in \mathcal{C} _i$, $D(P^{\lambda}_i(s,t))$ $\leq$ $D(P_i(s,t))$.
\end{myLemma}

\begin{proof}
Let $\mathcal{S}_i$ denote the sequence of boundary vertices on ${P^{\lambda}_i}(s,t)$ in graph $G$, where $v_s$ and $v_t$ are viewed as boundary vertices. Since $P_i(s,t)\in \mathcal{C}_i$, the sequence of boundary vertices on $P_i(s,t)$ is the same as $\mathcal{S}_i$ in graph $G$. For any two adjacent boundary vertices $v_j$ and $v_{j+1}$ in $\mathcal{S}_i$, based on Lemma\ref{lem:path-border-vertex}, we infer that the shortest distance between $v_j$ and $v_{j+1}$ in skeleton graph $G_\lambda$ is not greater than their shortest distance in graph $G$. Accumulating the shortest distances of all pairs of adjacent vertices in $\mathcal{S}_i$ on $G_\lambda$ and $G$,  we have $D(P^{\lambda}_i(s,t))$ $\leq$ $D(P_i(s,t))$.
\end{proof}

\begin{myTheo}\label{correctness-KSP-DG-I}
In the $i^{th}$ ($i\geq 1$) iteration of {KSP-DG}, if $D(P_k)\leq D({P^\lambda_{i+1}}(s,t))$, where $P_{k}$ is the $k^{th}$ path in $\mathcal{L}$, then the paths in $\mathcal{L}$ are KSPs from $v_s$ to $v_t$ in $G$, .
\end{myTheo}
\begin{proof}
Suppose for the sake of contradiction that there exists a path ${P_f}(s, t)$ with a smaller distance than $P_k$ and not in $\mathcal{L}$. ${P_f}(s, t)$ is thus not a candidate shortest path generated based on any reference path ${P^{\lambda}_j}(s, t)$ ($j\in[1, i]$). Therefore, there must be another reference path ${P^{\lambda}_f}(s, t)$ that matches ${P_f}(s, t)$ and $D({P^{\lambda}_f}(s, t))\geq D({P^{\lambda}_{i+1}}(s, t))$. Based on the given condition $D(P_k)\leq D({P^{\lambda}_{i+1}}(s, t))$, we can infer that $D(P_k)\leq D({P^{\lambda}_f}(s,t))$. Because $D({P^{\lambda}_f}(s, t))\leq D({P_f}(s, t))$ follows from Lemma~\ref{lem-oriented}, we have $D(P_k)\leq D({P_f}(s, t))$, which contradicts the initial assumption. 
\end{proof}

\subsection{Number of Iterations in KSP-DG}
At the very beginning when the DTLP is constructed, all weights are identical to the initial states and the lower bound distance of any two boundary vertices equals to their shortest distance within every subgraph. In this case, the first reference path ${P^{\lambda}_1}$ in $G_\lambda$ and the first shortest path ${P_1}$ in $G$ have the same sequence of boundary vertices. Suppose the other $(k-1)$ shortest paths also have the same sequence of boundary vertices as ${P^{\lambda}_1}$, then KSPs can be found in only one iteration, which is the best case. In the worst case, each of {\em k} shortest paths has a different sequence of boundary vertices, KSP-DG executes at most $k$ iterations to identify KSPs.

When the weights of edges dynamically change, it is hard to obtain an exact number of iterations based on a snapshot of $G$ and we just estimate an number of iterations. For this purpose, we assume the weights of only a portion of the edges change, and that the magnitude of the change is usually not large in road networks, which is in accordance with previous studies \cite{Bernhard2004Time}. 
Moreover, under the assumption that all roads have a similar varying trend in travel times,  the changes in the  weights of the edges will have the similar trend; namely most of them will increase or decrease similarly. When the edge weights change under the above assumptions, we find the lower bound distance between every pair of boundary vertices is still very close to their shortest distance in each subgraph. Like the initial snapshot of $G$ with unchanged weights, KSP-DG can determine KSPs with at most {\em k} iterations in most cases.  
\vspace{-1em}
\subsection{Cost Analysis of KSP-DG}\label{subsec:complexity}

\subsubsection{Communication Cost}
Assume that the number of iterations required to complete KSP-DG is $N_r$.
We consider a vertex as a unit amount of data transmitted across different nodes to measure communication cost. At each iteration of KSP-DG, the worker $u_i$ responsible for the query transmits a reference path $P^{\lambda}$ to all other workers, which incurs a communication cost of $O(|P^{\lambda}|)$. Then each worker which has received $P^{\lambda}$ computes partial {\em k} shortest paths between each pair of vertices is maintaining, and then sends those shortest paths to the worker $u_i$. This operation is expected to be evenly shared by all workers, and its communication cost is $O(|P^{\lambda}| k\cdot z/N_s)$ as each of the partial {\em k} shortest paths includes at most $z$ vertices, where $z$ is the threshold on the number of vertices in a subgraph, and ${N_s}$ denotes the number of workers in the cluster. Therefore, the communication cost for $N_r$ iterations in KSP-DG is $O(N_r|P^{\lambda}|k\cdot z/N_s)$.

\vspace{-0.5em}
\subsubsection{Computation Cost}
KSP-DG involves two major operations: (1) identifying reference paths and (2) computing the corresponding candidate KSPs for each reference path. For Operation (1), $N_r$ reference paths have to be computed in total using Yen's algorithm by the worker responsible for the query, so the time complexity of this operation is that of Yen's Algorithm,  $O(N_r{|\mathcal{V}_\lambda|}(|\mathcal{E}_\lambda|+|\mathcal{V}_\lambda|log{|\mathcal{V}_\lambda|}))$ \cite{gao2010fast}, where $|\mathcal{V}_\lambda|$ and $|\mathcal{E}_\lambda|$ are the numbers of vertices and edges in the skeleton graph $G_\lambda$ respectively. In each iteration in KSP-DG, Operation (2) is further decomposed into computing partial {\em k} shortest paths between each pair of adjacent boundary vertices in $P_\lambda$ with Yen's Algorithm, and there are ${|P_\lambda|}-1$ pairs of adjacent boundary vertices to be processed. Hence, the computation cost of this operation in each iteration is $O$($k\cdot z(N_e+{z\log_z}$)${|P_\lambda|}$), 
where $N_e$ is the maximum number of edges in a subgraph. Since this operation is almost evenly shared by $N_s$ workers and executed $N_r$ times, its computation cost is  $O$($k\cdot z\cdot(N_e+z\log_z$)$|P_{\lambda}| N_r/{N_s}$).
Therefore, the total computation cost for KSP-DG is $O(N_r(|\mathcal{V}_\lambda|(|\mathcal{E}_\lambda|+|\mathcal{V}_\lambda|log{|\mathcal{V}_\lambda|})+k\cdot z(N_e+z\log_z) |P_{\lambda}|/N_s))$.

\section{Experiments}\label{sec:exp}
\subsection{Implementation of KSP-DG on Storm}
 We implement KSP-DG on Apache Storm\cite{Storm}, a popular distributed stream processing framework, to evaluate its performance. Following the Storm paradigm, KSP-DG is designed as a "topology" of a directed acyclic graph with Spouts and Bolts as nodes of this graph. In Storm, spouts and bolts, as logical processors, are connected with streams, where each stream is an unbounded sequence of tuples  containing the data to be processed. A spout acts as a source of streams in a topology, and every bolt processes the tuples according to a user-specified logic.  
 
 In the KSP-DG topology, only one EntranceSpout is deployed on the master, which receives edge weight updates and new KSP queries. Moreover, the EntranceSpout just maintains the partitioning of subgraphs, which assists Spout to send edge weight updates to the SubgraphBolt maintaining the subgraph covering this edge. One or more SubgraphBolts are deployed on each worker to take care of subgraphs organized as adjacency lists, where each SubgraphBolt can be in charge of more than one subgraph. One or more QueryBolts are also created on each worker to process KSP queries, and each query is processed by a unique QueryBolt. 
 
\begin{figure}[htbp]
\centering
\includegraphics[width=0.45\textwidth]{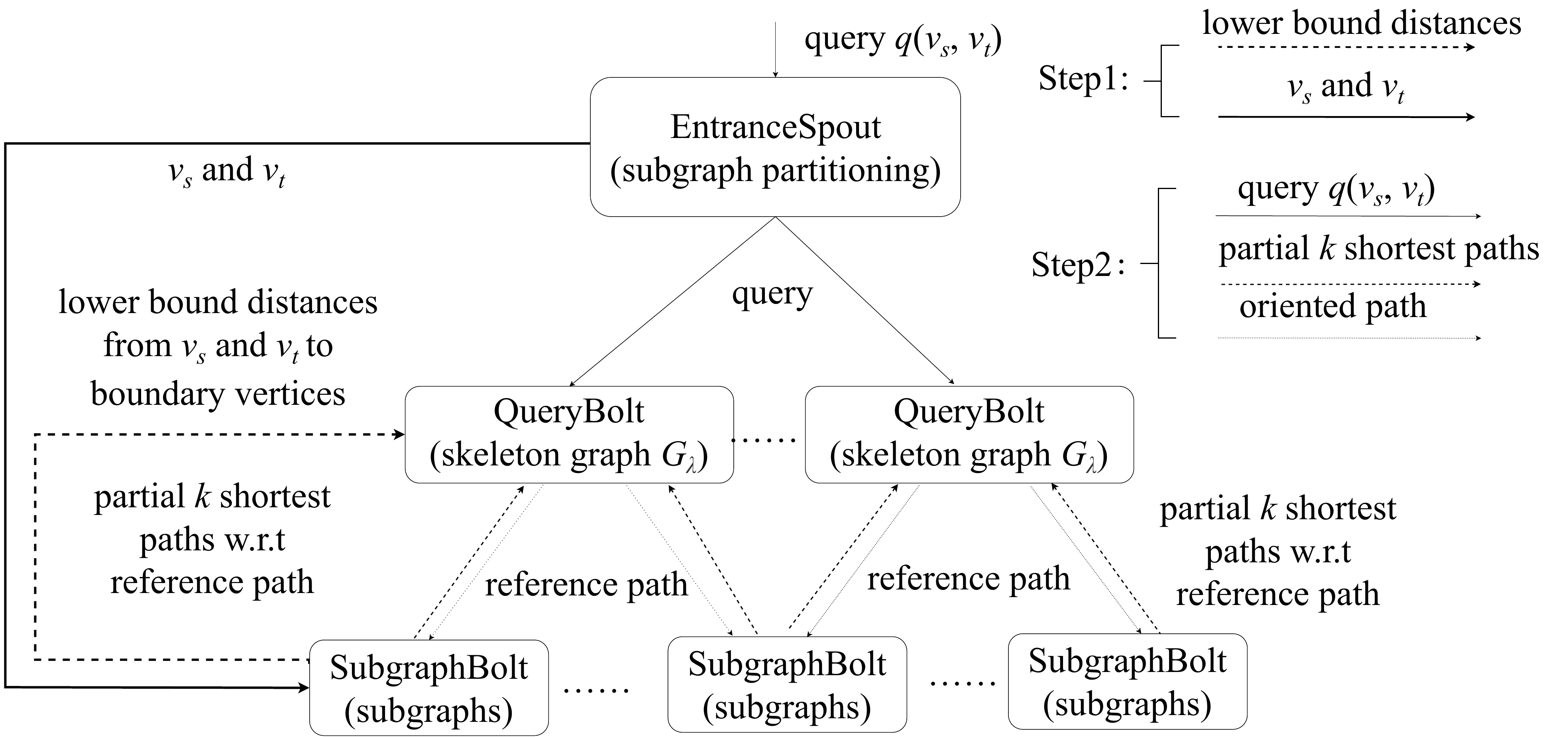}
\caption{Framework for Deploying {KSP-DG} on Storm}\label{fig:KSP-DG-Storm}
\vspace{-0.15in}
\end{figure}
 
 Figure~\ref{fig:KSP-DG-Storm} presents the procedure of processing $q(v_s,v_t)$ by KSP-DG on Storm. When the query $q$ arrives at the EntranceSpout, we first detect whether $v_s$ and $v_t$ are boundary vertices. If not, they will be processed by {\em Step 1} first; otherwise, $q$ will be directly processed by {\em Step 2}.
 
 {\em Step 1} : The EntranceSpout sends $v_s$ and $v_t$ together as a tuple to the SubgraphBolt(s) responsible for the subgraph(s) covering $v_s$ and $v_t$ respectively. Then each of the relevant SubgraphBolt(s) computes the lower bound distances from $v_s$ and $v_t$ to each boundary vertex in the subgraph(s) covering these two vertices respectively. Next, those identified lower bound distances as well as $v_s$ and $v_t$ are packaged into tuples that are delivered to every QueryBolt, which adds $v_s$ and $v_t$ into the skeleton graph $G_\lambda$ using the principle discussed in Section~~\ref{subsec:nonboundary}.
 
 {\em Step 2} : The EntranceSpout assigns $q$ to a QueryBolt $QB_i$, which computes a reference path $P^{\lambda}$ for this query. Next, the tuple consisting of  $P^{\lambda}$ and the ID of $q$ is broadcast to all SubgraphBolts. After receiving this tuple, every SubgraphBolt identifies the subgraphs within those it maintains that cover any pair of two adjacent vertices in $P^{\lambda}$, and then generates partial {\em k} shortest paths for each pair of vertices from the corresponding subgraphs. A SubgraphBolt uses a {\em map} to keep the partial {\em k} shortest paths, where a key is a pair of adjacent vertices in $P^{\lambda}$ and its value is a priority queue storing the partial {\em k} shortest paths for this pair ordered according to their distances. Afterwards, this map and query $q$ form a tuple that is returned to QueryBolt ${QB}_i$ based on the ID of $q$. Since a pair of adjacent vertices in $P^{\lambda}$ probably belongs to more than one subgraph, ${QB}_i$ has to choose {\em k} shortest ones from all received partial {\em k} shortest paths for each pair, and then joins the selected partial paths to generate candidate KSPs corresponding to $P^{\lambda}$. These candidate KSPs are used to update the list of paths obtained so far. When the identified KSPs satisfy the terminating condition, QueryBolt ${QB}_i$ outputs the final result; otherwise, it generates the next reference path and continues onto the next iteration. 

\subsection{Experiment Setup and Datasets}\label{subsec:setup}

The system is deployed on a cluster of 10 servers  from a public cloud service provider, and each server has a quadcore CPU of 2.5GHz and 32GB memory. These servers are connected via Ethernet with a bandwidth of 1Gbps. {We use as the datasets four real road networks with travel times from New York, Colorado, and Florida, and Central USA \cite{DIMACS}, which are directed graphs, denoted by {NY}, {COL}, {FLA}, and CUSA respectively.} 
The number of vertices and edges in these graphs are given in Table~\ref{expri-graph}, along with the number of subgraphs and skeleton graphs that can be obtained when $z$ (the maximum subgraph size) takes on their typical values. The numbers in parentheses are the numbers of subgraphs with more than five boundary vertices ($n_b>$5), and the last column shows the number of vertices in the skeleton graph. Notice that the skeleton graph is much smaller than the original graph.
\vspace{-0.1in}
\begin{table}[h]
\centering
\caption{Statistics on the Road Network Datasets}\label{expri-graph}
\vspace{-0.15in}
\resizebox{\linewidth}{!}{
\begin{tabular}{llllll}
\hline
road network& \#vertices & \#edges & {\em z}& \#subgraphs ($n_b>$5) & $G_\lambda$\\
\hline
NY & 264,346 & 733,846 & 200 & 4,173 (1,586)&  24461 \\
\hline
COL & 435,666 & 1,057,066 & 200 & 8,001 (2,004)& 27,665\\
\hline
FLA & 1,070,376 & 2,712,798 & 500 & 13,701 (3,682) & 52,640\\
\hline
CUSA & 14,081,816 & 34,292,496 & 1000 & 121,725 (18,251) & 514,618\\ 
\hline
\end{tabular}}
\vspace{-0.1in}
\end{table}

We use the (normalized) travel time on each road in the road networks as the edge weights in the graphs. Since only one snapshot of the travel times in each  road network is given in the dataset, we adopt a well-established model \cite{Bernhard2004Time} to dynamically vary the travel time in each road to simulate real-world traffic conditions. We use $\alpha$ to represent the percentage of edges that change weights at each snapshot, and $[-\tau, \tau]$  to denote the range of weight variation. In the following experiments, $\alpha=35\%$, and $\tau=30\%$ unless otherwise specified. {We apply identical changes to the weights of the two edges in the opposite direction between a pair of vertices to simulate varying undirected graphs; for CUSA, we also experiment with the case where the weights of the opposite edges change independently to simulate varying directed graphs.} All results shown are the average of 20 runs on the cluster of 10 servers unless otherwise specified. 

\subsection{Evaluation of DTLP}
We evaluate the influence of parameters $z$, $\xi$, $\alpha$, and $\tau$ on the performance of DTLP, {which are summarized in Table~\ref{expri-parameters}.  Due to space limitations, only the results on NY and CUSA are shown; the trends on the other two datasets are similar; those results are available in the full technical report \cite{Anonymous-technique-report}.}

\vspace{-0.5em}
\begin{table}[h]
\centering
\caption{Summary of Parameters Used in Evaluation}\label{expri-parameters}
\vspace{-0.15in}
\smallskip\noindent
\resizebox{\linewidth}{!}{
\begin{tabular}{|c|c|}
\hline
Parameters & Meaning\\
\hline
$z$ & size of subgraph\\
\hline
$\xi$ & number of bounding paths between a pair of boundary vertices\\
\hline
$\alpha$ & percentage of edges changing weights at each snapshot\\
\hline
$\tau$ & range of edge weight variation\\
\hline
\end{tabular}
}
\vspace{-0.1in}
\end{table}

\begin{figure*}[htbp]
\begin{minipage}[t]{0.235\textwidth}
\centering
\includegraphics[width=\textwidth]{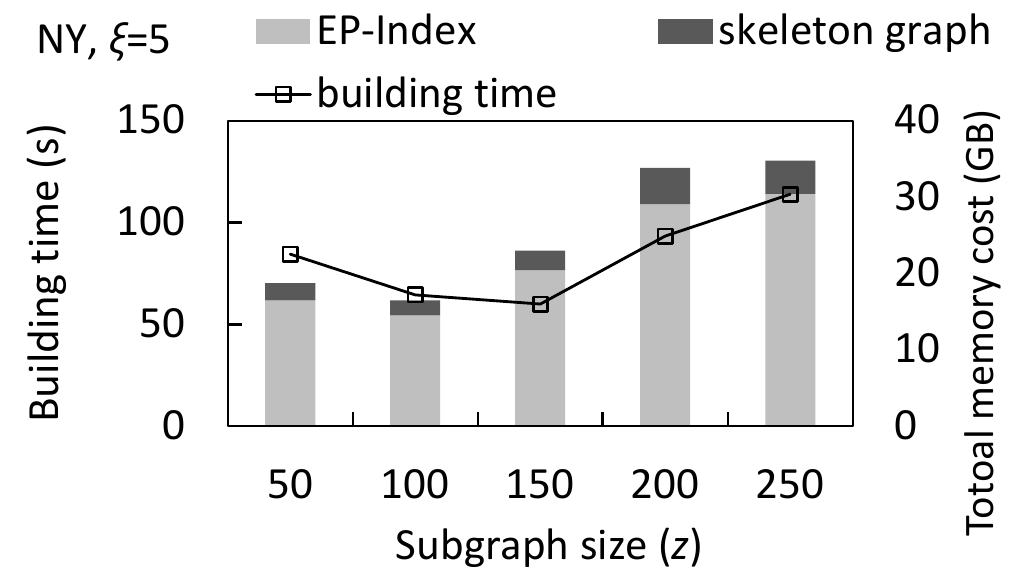}
\vspace{-0.8cm}
\caption{\scriptsize{Construction Cost (NY)}}\label{skeleton-time-NY}
\end{minipage}
\begin{minipage}[t]{0.235\textwidth}
\centering
\includegraphics[width=\textwidth]{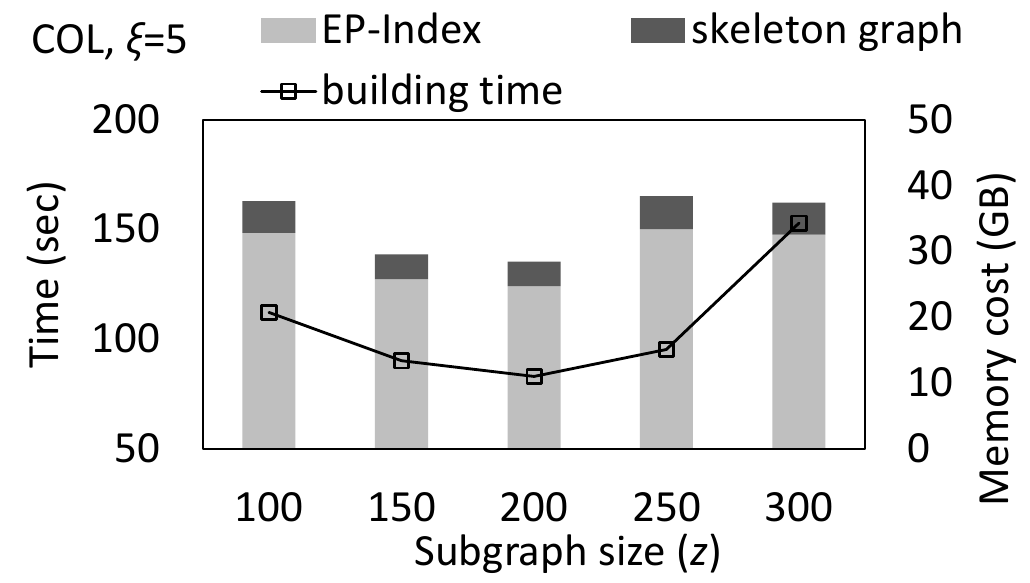}
\vspace{-0.8cm}
\caption{\scriptsize{Construction Cost (COL)}}\label{skeleton-time-COL}
\end{minipage}
\begin{minipage}[t]{0.235\textwidth}
\centering
\includegraphics[width=\textwidth]{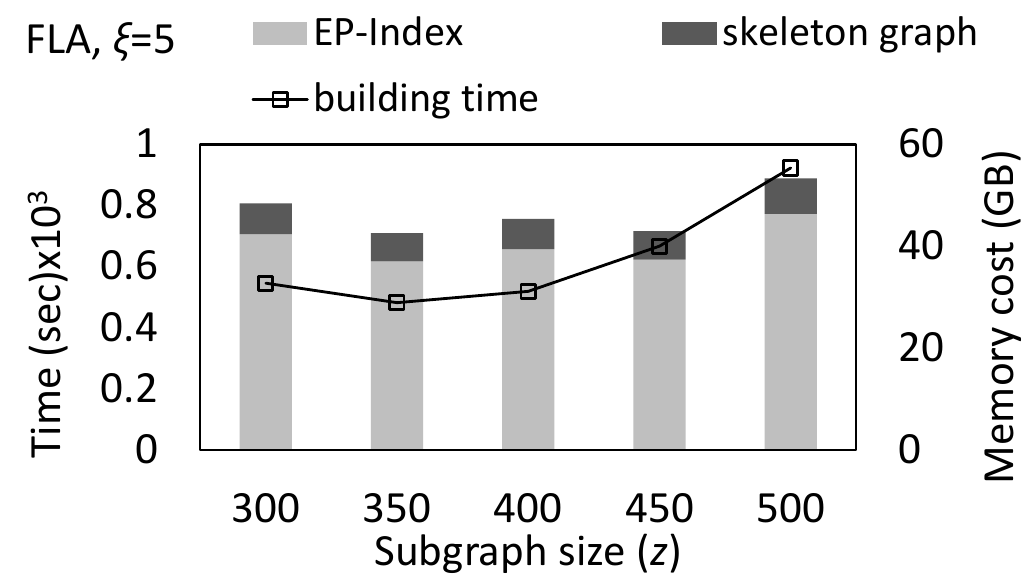}
\vspace{-0.8cm}
\caption{\scriptsize{Construction Cost (FLA)}}\label{skeleton-time-FLA}
\end{minipage}
\begin{minipage}[t]{0.235\textwidth}
\centering
\includegraphics[width=\textwidth]{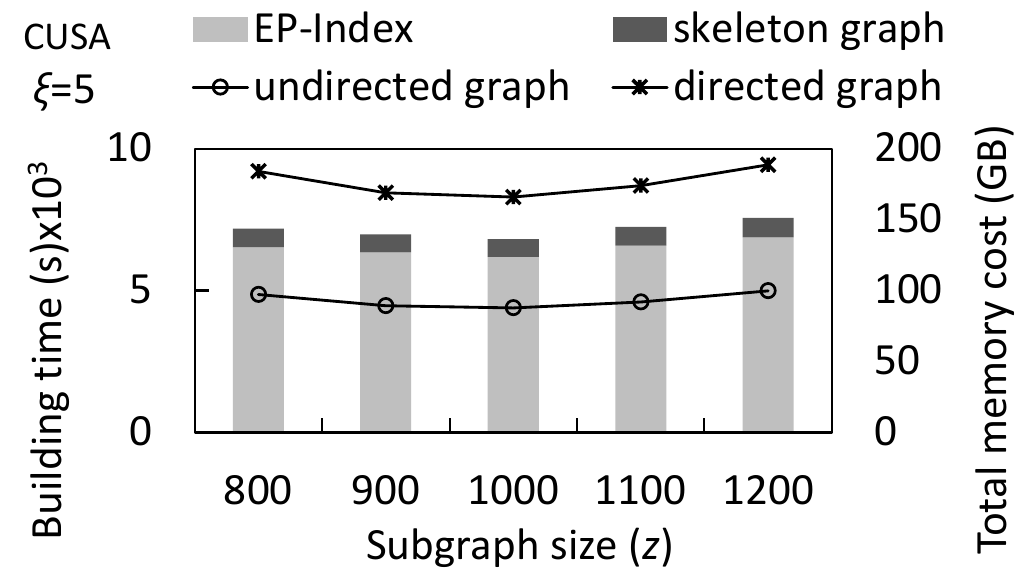}
\vspace{-0.8cm}
\caption{\scriptsize{Construction Cost (CUSA)}}\label{skeleton-time-CUSA}
\end{minipage}
\vspace{-0.15in}
\end{figure*}

\begin{figure*}[htbp]
\begin{minipage}[t]{0.235\textwidth}
\centering
\includegraphics[width=\textwidth]{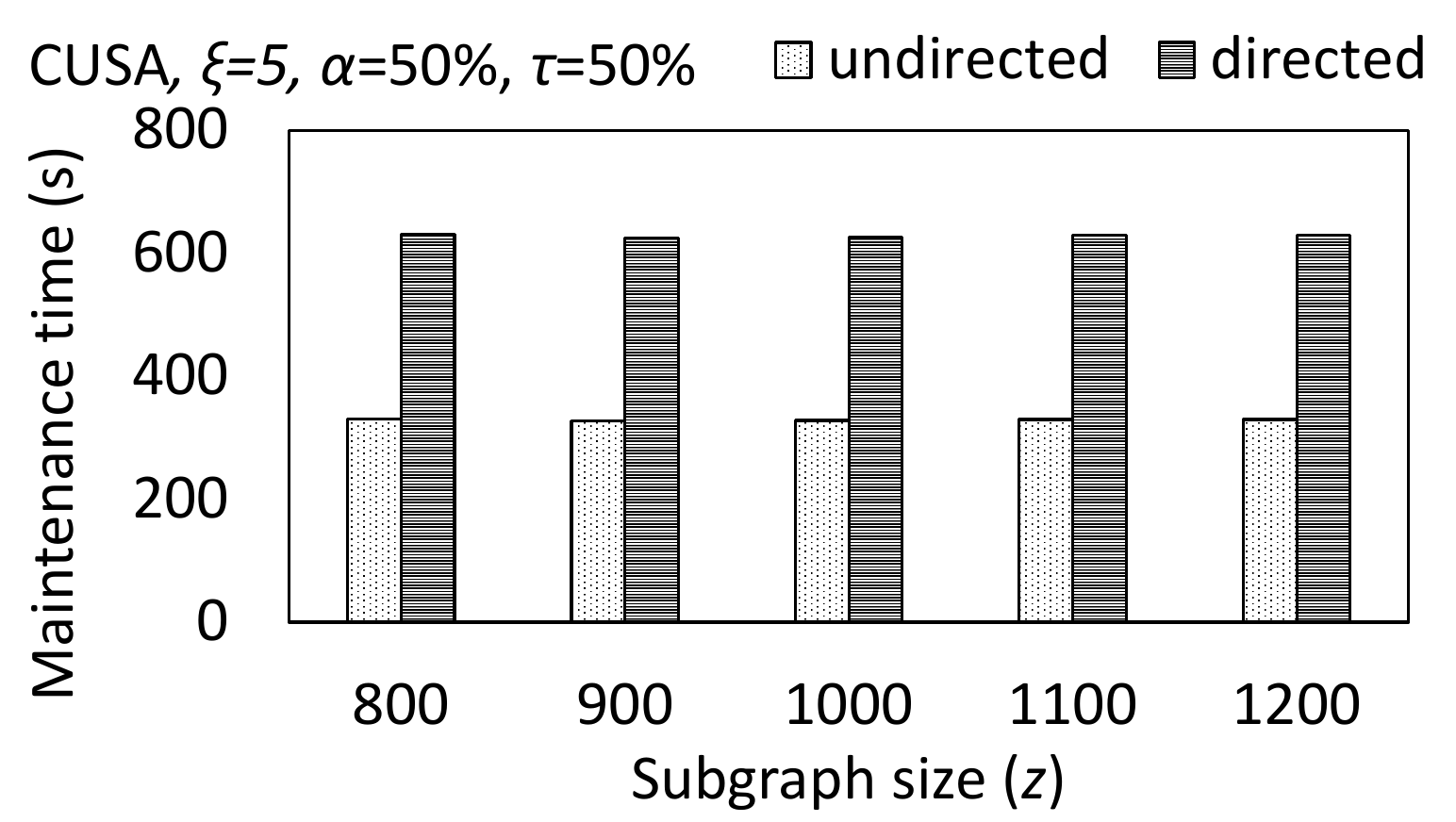}
\vspace{-0.8cm}
\caption{{Maintenance Cost (CUSA)}}\label{update-time-directed}
\end{minipage}
\begin{minipage}[t]{0.235\textwidth}
\centering
\includegraphics[width=\textwidth]{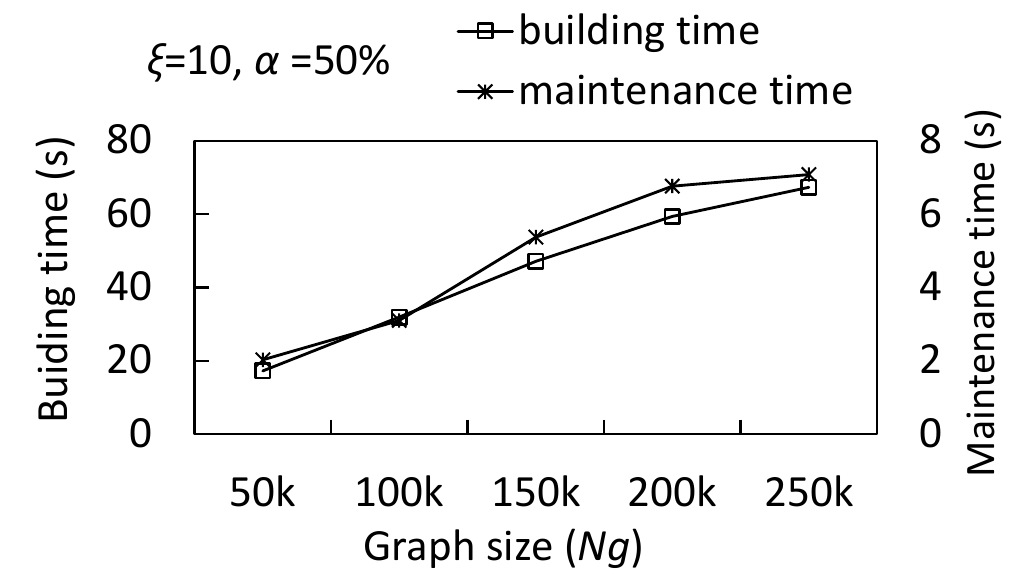}
\vspace{-0.8cm}
\caption{{Time Cost w.r.t. $N_g$}}\label{update-time-size}
\end{minipage}
\begin{minipage}[t]{0.235\textwidth}
\centering
\includegraphics[width=\textwidth]{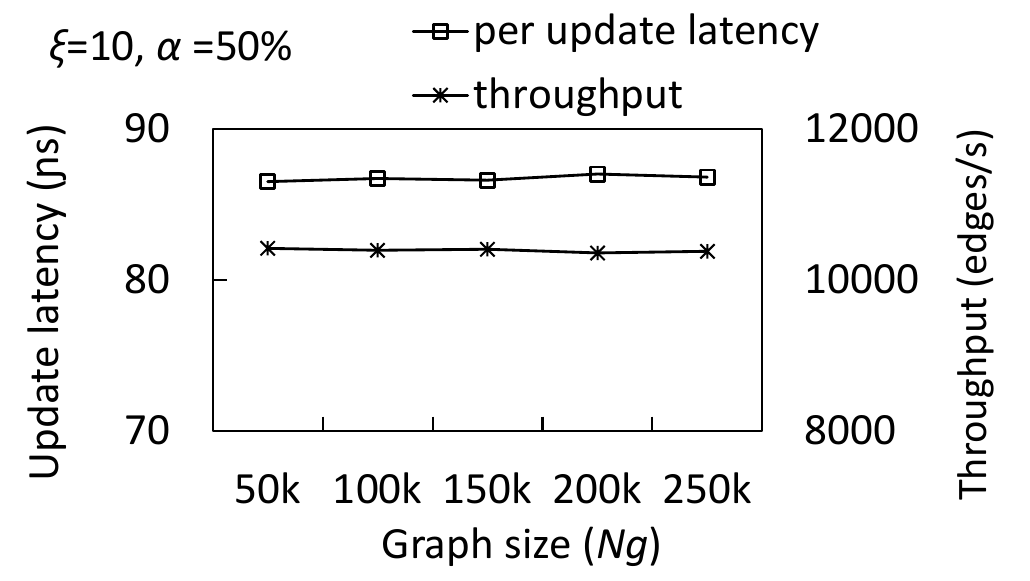}
\vspace{-0.8cm}
\caption{\scriptsize{Throughput w.r.t. $N_g$}}\label{throughput}
\end{minipage}
\begin{minipage}[t]{0.235\textwidth}
\centering
\includegraphics[width=\textwidth]{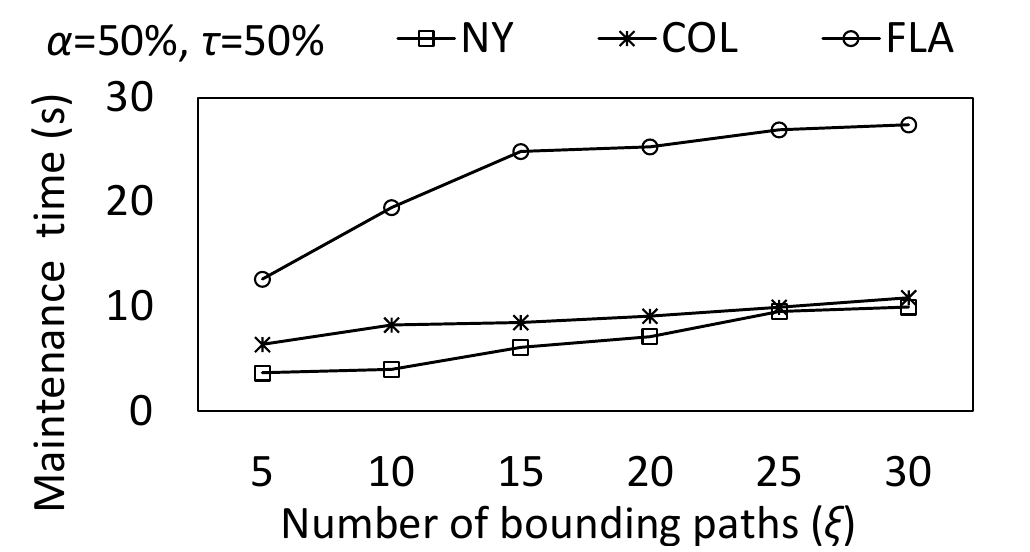}
\vspace{-0.8cm}
\caption{\scriptsize{Maintenance Cost w.r.t $\xi$}}\label{update-time-x}
\end{minipage}
\vspace{-1em}
\end{figure*}

\vspace{-0.05in}
\subsubsection{Construction Cost}
Figures~\ref{skeleton-time-NY}-\ref{skeleton-time-CUSA} depict the time and memory usage to build DTLP for NY, COL, FLA, and CUSA varying the values of $z$.  Building time first decreases and then increases as $z$ grows. This is because the number of subgraphs is reduced when the road network is partitioned into larger subgraphs, resulting in less subgraphs being assigned to each server, so the total building time declines. When $z$ grows beyond a certain threshold (e.g., $z$=100 for NY), however, this decrease is outweighed by the increase in the number, as well as, the average length of bounding paths in each subgraph. For the same reason, memory consumption caused by EP-Index and the skeleton graph shows a trend similar to that of building time. Moreover, we compare the time for building DTLP in the directed and undirected graphs using CUSA. Figure~\ref{skeleton-time-CUSA} shows that the building time for the directed graph is double that of the undirected graph, as two sets of bounding paths (in opposite directions) for each pair of boundary vertices have to be computed in directed graphs. For the same reason, the maintenance cost of DTLP for directed graphs is also almost doubled, as shown in Figure \ref{update-time-directed}.

We further evaluate the influence of the size of a graph on the building time of DTLP. For this purpose, we choose five subgraphs from COL with 50K, 100K, 150K, 200K, and 250K vertices respectively, and use $N_g$ to denote their sizes. We measure the time for building DTLP for these selected graphs, and the results are shown in Figure~\ref{update-time-size} (left vertical axis). Apparently, the construction cost of DTLP increases almost linearly with the size of the graph, as the cost of computing bounding paths is roughly proportional to $N_g$.

\vspace{-0.2cm}
\subsubsection{Maintenance Cost}
We study the trend of the maintenance cost of DTLP on graphs of varying sizes. We change the weights of half of the edges in each graph. Therefore, the number of varying weights is directly proportional to $N_g$, the size of the graph examined.  As shown in Figure~\ref{update-time-size} (right vertical axis), there is an approximately linear ascending trend in the maintenance time with the size of a graph, as the maintenance time is heavily affected by the number of the varying weights which is proportional to  $N_g$. {Additionally, the maximum throughput and per update latency in these graphs are reported in Figure \ref{throughput}. In this evaluation, we continuously apply 1000 rounds of changes to the weights of half of the edges, and measure the maximum throughput and the average per update latency. Apparently, the size of the graph does not have a significant impact on the maximum throughput and per update latency.} 

We further evaluate the time required to update DTLP with different values of $\alpha$ and $\xi$. In Figure~\ref{update-time-x}, we set $\alpha=50\%, \tau=50\%$ and vary $\xi$; in Figure~\ref{update-time-varying-range}, we let $\xi=10, \tau=50\%$, and vary $\alpha$. In these two groups of experiments, all weight updates are fed into the system as a batch, and the maintenance time of DTLP is measured as the time between receiving the new weights and finishing updating the DTLP index. Figure~\ref{update-time-x} shows an ascending trend of the maintenance cost w.r.t. $\xi$, as larger values of $\xi$ cause more bounding paths with bound distances to be updated, resulting in higher maintenance costs. However, the rate of growth slows down when $\xi$ exceeds a certain value (e.g., $\xi$=15 in FLA) because the number of bounding paths in some subgraphs stops increasing when $\xi$ is large enough. 
Figure~\ref{update-time-varying-range} shows an ascending trend with increasing $\alpha$, as more weights need to be processed. 

\vspace{-0.5em}
\subsection{Evaluation of KSP-DG}\label{subsec:ksp-test}
Our next task is to study the impact of different parameters ($z$, $\alpha$, $\xi$, and $k$) on the performance of KSP-DG. The default value of $k$ is 2, but we vary it as needed for our evaluation. 

 \begin{figure*}[!h]
 \begin{minipage}[t]{0.223\textwidth}
\centering
\includegraphics[width=\textwidth]{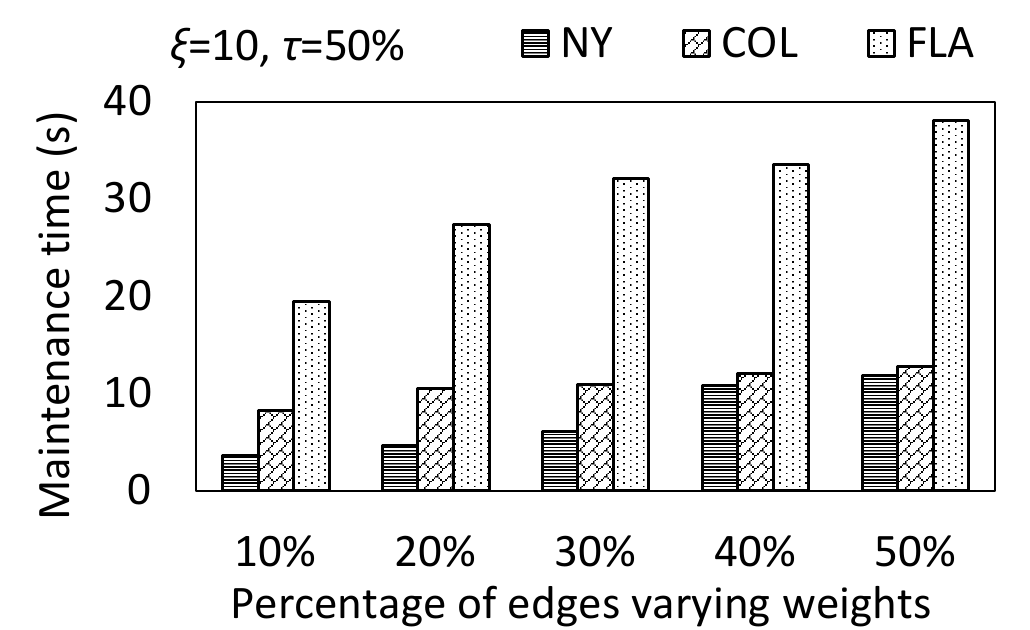}
\vspace{-0.8cm}
\caption{Maintenance Cost w.r.t $\alpha$}\label{update-time-varying-range}
\end{minipage}
\begin{minipage}[t]{0.235\textwidth}
\centering
\includegraphics[width=\textwidth]{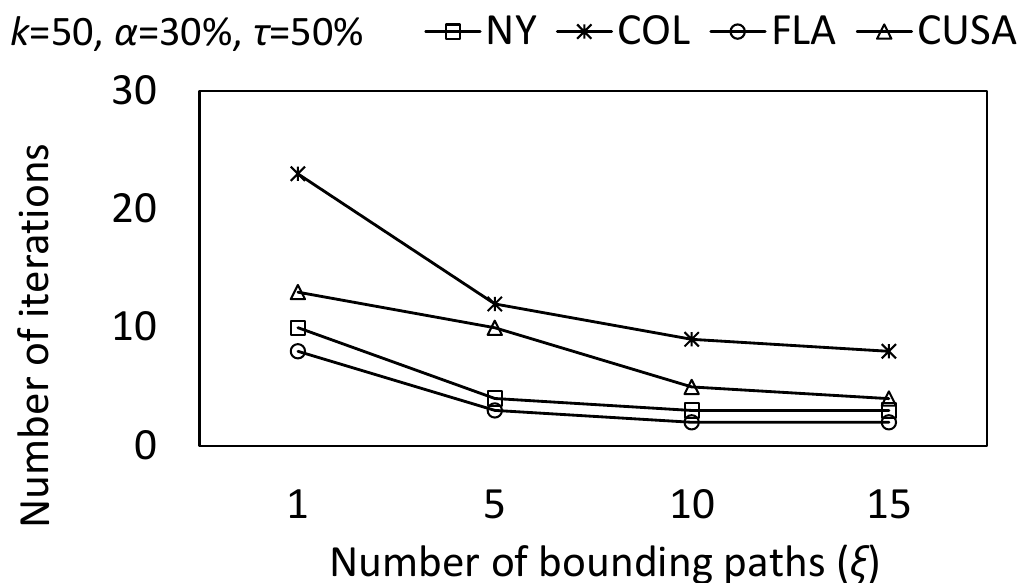}
\vspace{-0.8cm}
\caption{\#Iterations w.r.t. $\xi$}\label{iterations-xi}
\end{minipage}
\begin{minipage}[t]{0.235\textwidth}
\centering
\includegraphics[width=\textwidth]{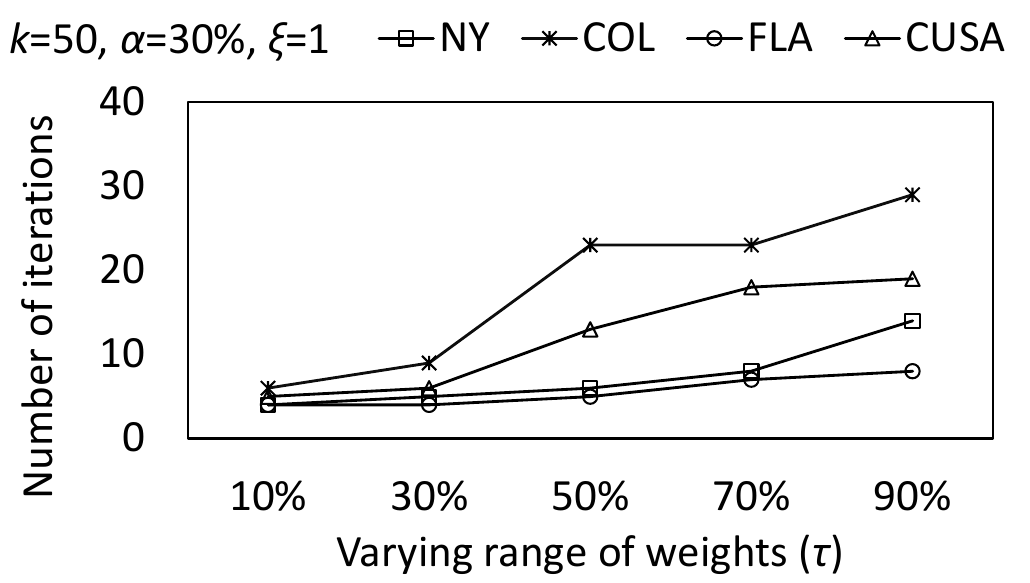}
\vspace{-0.8cm}
\caption{\#Iterations  w.r.t. $\tau$}\label{iterations-varying-rate}
\end{minipage}
\begin{minipage}[t]{0.235\textwidth}
\centering
\includegraphics[width=\textwidth]{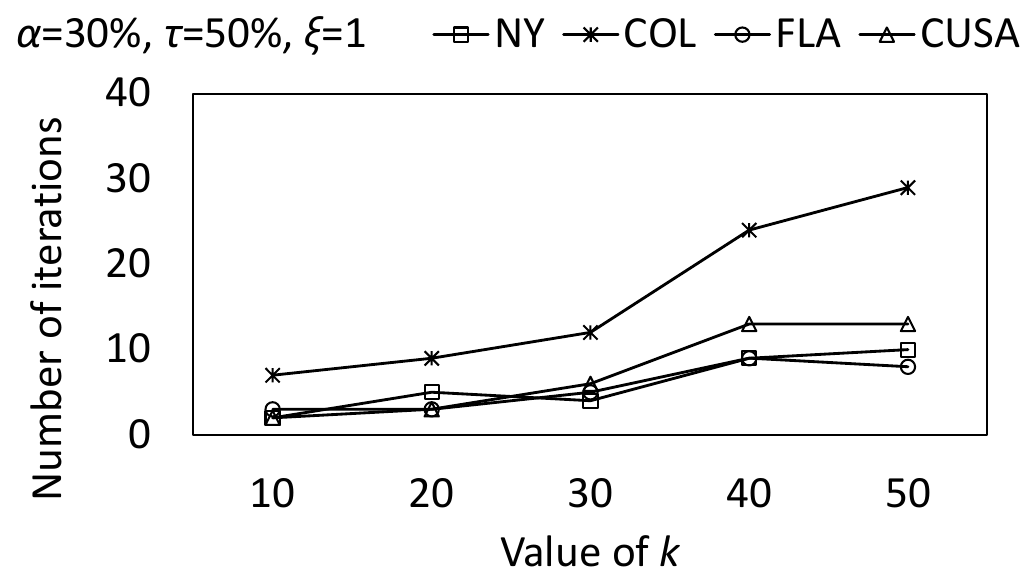}
\vspace{-0.8cm}
\caption{{\#Iterations w.r.t. $k$}}\label{iterations-k}
\end{minipage}
\vspace{-0.8em}
\end{figure*}

\begin{figure*}[htbp]
\begin{minipage}[t]{0.235\textwidth}
\centering
\includegraphics[width=\textwidth]{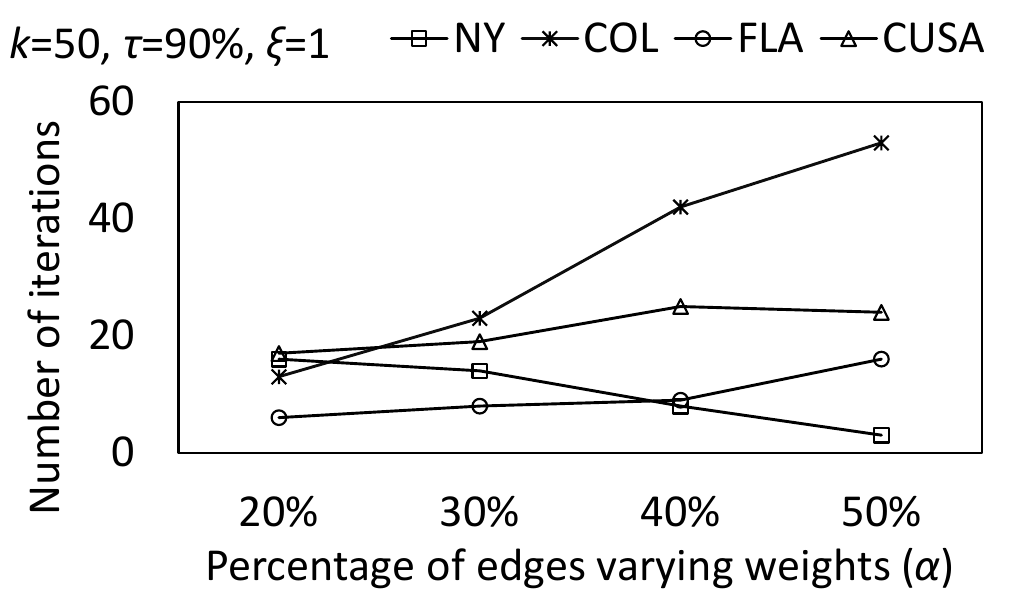}
\vspace{-0.8cm}
\caption{{\#Iterations w.r.t. $\alpha$}}\label{iterations-varying-range}
\end{minipage}
\begin{minipage}[t]{0.235\textwidth}
\centering
\includegraphics[width=\textwidth]{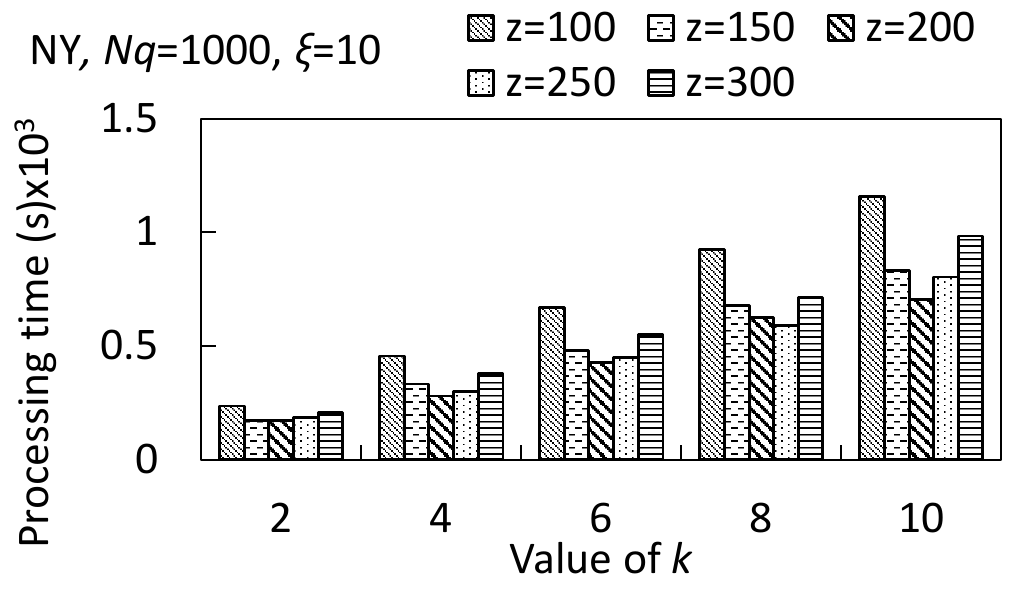}
\vspace{-0.8cm}
\caption{\scriptsize{Processing Time (NY)}}\label{NY-queryTime-k}
\end{minipage}
\begin{minipage}[t]{0.235\textwidth}
\centering
\includegraphics[width=\textwidth]{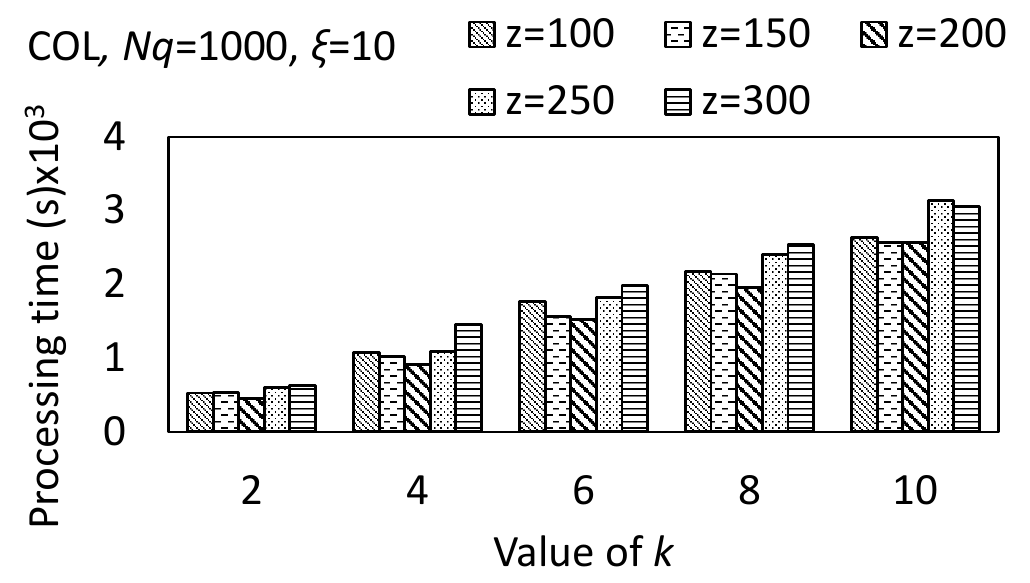}
\vspace{-0.8cm}
\caption{\scriptsize{Processing Time (COL)}}\label{COL-queryTime-k}
\end{minipage}
\begin{minipage}[t]{0.235\textwidth}
\centering
\includegraphics[width=\textwidth]{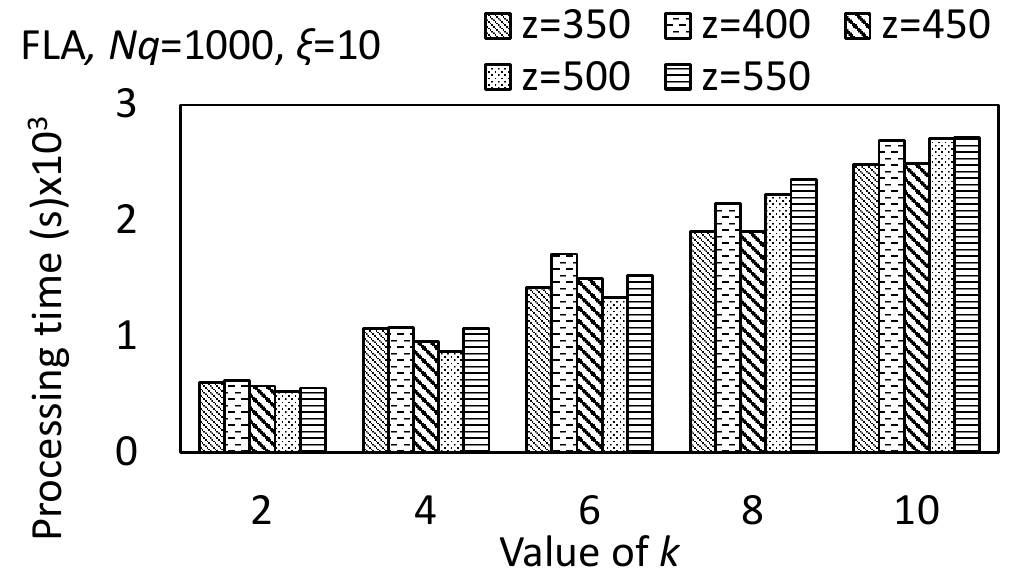}
\vspace{-0.8cm}
\caption{\scriptsize{Processing Time (FLA)}}\label{FLA-queryTime-k}
\end{minipage}

\vspace{-0.1in}
\end{figure*}

\subsubsection{Number of Iterations}
The number of iterations required by KSP-DG with varying values of $\xi$, $\alpha$, $k$, and $\tau$ are shown in Figures~\ref{iterations-xi}-\ref{iterations-varying-range}. {$k$ is set to 50 for effective measuring the the number of iterations as the variation is more apparent when $k$ takes larger values}. Figure~\ref{iterations-xi} shows that the number of iterations significantly decreases with increasing $\xi$, which is as expected, because when more bounding paths are indexed, it narrows the gap between the lower bound distance and the actual shortest distance for the given pair of vertices and thus reduces the number of iterations needed. However, as a higher construction and maintenance cost of DTLP is associated with a larger $\xi$, the value of $\xi$ has to be chosen in a way that balances the processing time of KSP-DG and the construction/maintenance time of DTLP. 

As shown in Figure~\ref{iterations-varying-rate}, the number of iterations increases with growing $\tau$ (the varying range of weights). The reason is that, greater variation in the weights would loosen the lower bound distance and thus weaken the pruning power of the skeleton graph $G_\lambda$. Moreover, greater values of {\em k} would incur more iterations in KSP-DG, as shown in Figure \ref{iterations-k}. But the good news is that the rate of increase is very small when {\em k} is not large (i.e., $k<30$), which should be sufficient for most applications.  The influence of $\alpha$, as shown in Figure \ref{iterations-varying-range}, differs from one dataset to another, implying that its effect may depend on the particular distribution of edges with varying weights. However, it is apparent that the numbers of iterations for all the datasets are small when the weights of the graph are not changing dramatically (i.e., $\alpha<30\%$).

\begin{table}[h]
\tiny
\centering
\caption{Number of Vertices in Skeleton Graph $G_\lambda$ with Varying  $z$}\label{scale-skeleton-graph}
\vspace{-0.15in}
\resizebox{.8\linewidth}{!}{
\begin{tabular}{llllll}
\hline
\makecell[cl]{NY,COL\\FLA\\CUSA}&  \makecell[cl]{{\em z}=100\\{\em z}=350\\{\em z}=800 } &  \makecell[cl]{{\em z}=150\\{\em z}=400\\{\em z}=900}  & \makecell[cl]{{\em z}=200\\{\em z}=450\\{\em z}=1000} & \makecell[cl]{{\em z}=250\\{\em z}=500\\{\em z}=1100}  & \makecell[cl]{{\em z}=300\\{\em z}=550\\{\em z}=1200 }\\
\hline
$G_\lambda$ (NY) & 32,534 & 27,668 & 24,461 & 22,604&  20,775 \\
\hline
$G_\lambda$ (COL) & 36,831 & 30,886 & 27,655 & 25,329 & 23,271\\
\hline
 $G_\lambda$ (FLA) & 60,125 & 57,085 & 54,695 & 52,640 & 50,411\\
 \hline
$G_\lambda$ (CUSA) & 60,125 & 561,085 & 514,618 & 495,606 & 480,801\\
\hline
\end{tabular}}
\end{table}

\subsubsection{Query Processing Time w.r.t. $z$ and $k$}
Figures~\ref{NY-queryTime-k}-\ref{CUSA-queryTime-k} depict the influence of parameters $z$ and $k$ on the query processing time using KSP-DG. In this group of experiments, we randomly generate 1,000 queries ($N_q=1,000$), feed them into the system simultaneously, and measure the total processing time of all the queries. As can be observed from the plots, the processing time first decreases and then increases as $z$ grows.
This is because as $z$ increases, the number of subgraphs decreases, and so does the number of boundary vertices, which in turn leads to a smaller skeleton graph $G_\lambda$ (shown in Table~\ref{scale-skeleton-graph}). Roughly speaking, a smaller $G_\lambda$ means fewer vertices in a reference path and fewer subgraphs to be explored in each iteration. Moreover, the cost of generating partial {\em k} shortest paths within each subgraph grows very slowly when the size of the subgraph is small. Therefore, the processing time decreases as $z$ increases but is still small. However, when $z$ grows beyond a certain threshold (e.g., ~200 for $k=2$ on {NY}), the cost of computing partial {\em k} shortest paths in subgraphs increases significantly and dominates the overall cost. Consequently, the processing time starts to grow as $z$ becomes sufficiently large. In the following discussion, {we set the value of $z$ to 200, 200, 500, and 1000 in NY, COL, FLA, and CUSA respectively, unless otherwise specified.} 

From each figure we also observe that the processing time of KSP-DG increases linearly with the value of $k$, which is because larger values of $k$ lead to more candidate {\em k} shortest paths being generated in each iteration. Moreover, the number of iterations also increases with {\em k}, as shown in Figure~\ref{iterations-k}. 

\vspace{-0.7em}
\subsubsection{Query Processing Time w.r.t. $N_q$}
The scalability of KSP-DG w.r.t. the number of concurrent queries is evaluated, and the results are shown in Figure~\ref{Processing-time-query-number}. We generate multiple batches of queries with different batch sizes (number of queries), and feed each batch into the system to measure the total processing time. From the curves in Figure~\ref{Processing-time-query-number}, we observe that running time of KSP-DG increases approximately linearly w.r.t. the number of queries with a low rate of growth, benefiting from its distributed paradigm. 

\begin{figure*}[!h]
\begin{minipage}[t]{0.235\textwidth}
\centering
\includegraphics[width=\textwidth]{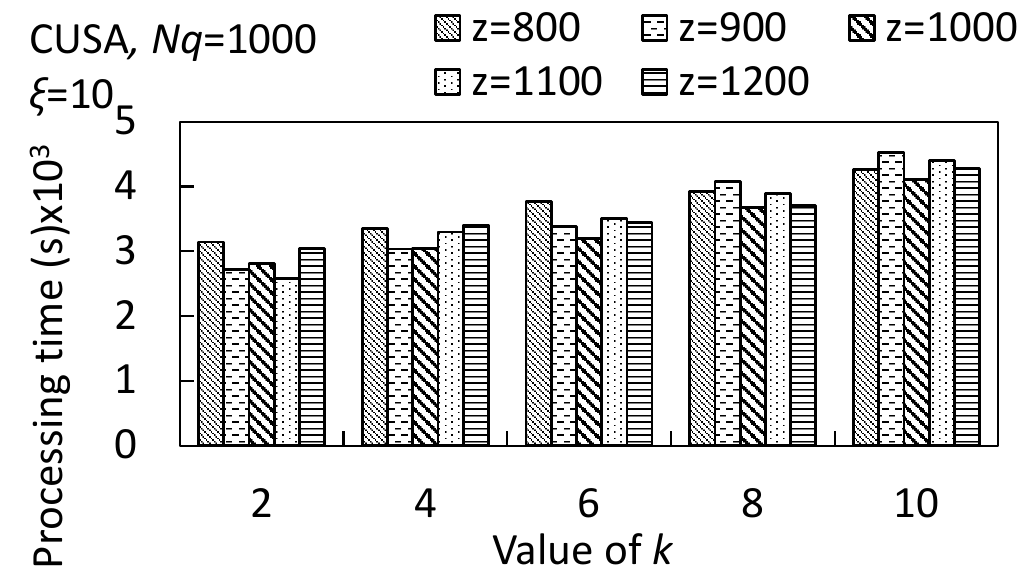}
\vspace{-0.8cm}
\caption{Processing Time (CUSA)}\label{CUSA-queryTime-k}
\end{minipage}
\begin{minipage}[t]{0.235\textwidth}
\centering
\includegraphics[width=\textwidth]{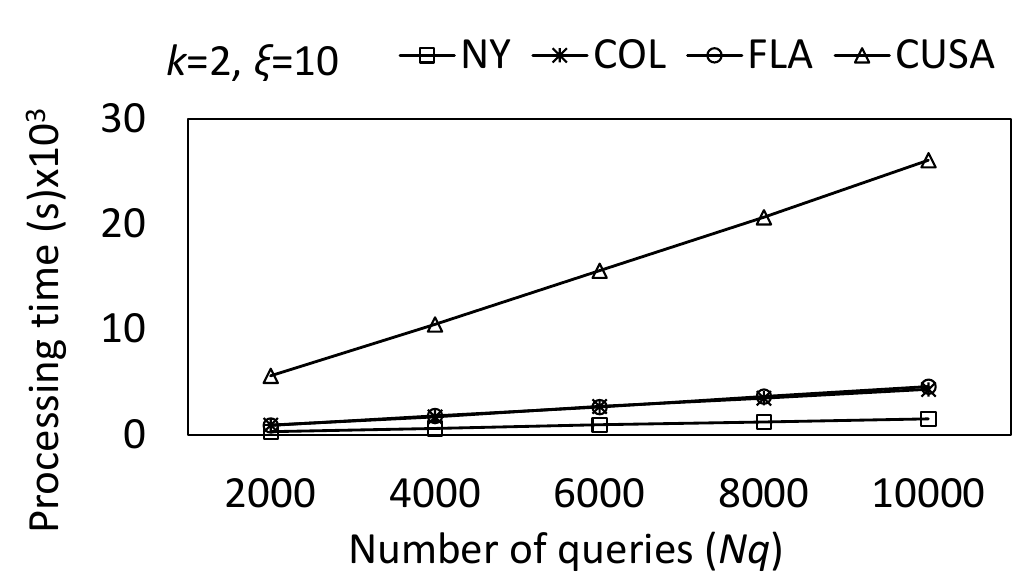}
\vspace{-0.8cm}
\caption{Processing Time w.r.t. $N_q$}\label{Processing-time-query-number}
\end{minipage}
\begin{minipage}[t]{0.235\textwidth}
\centering
\includegraphics[width=\textwidth]{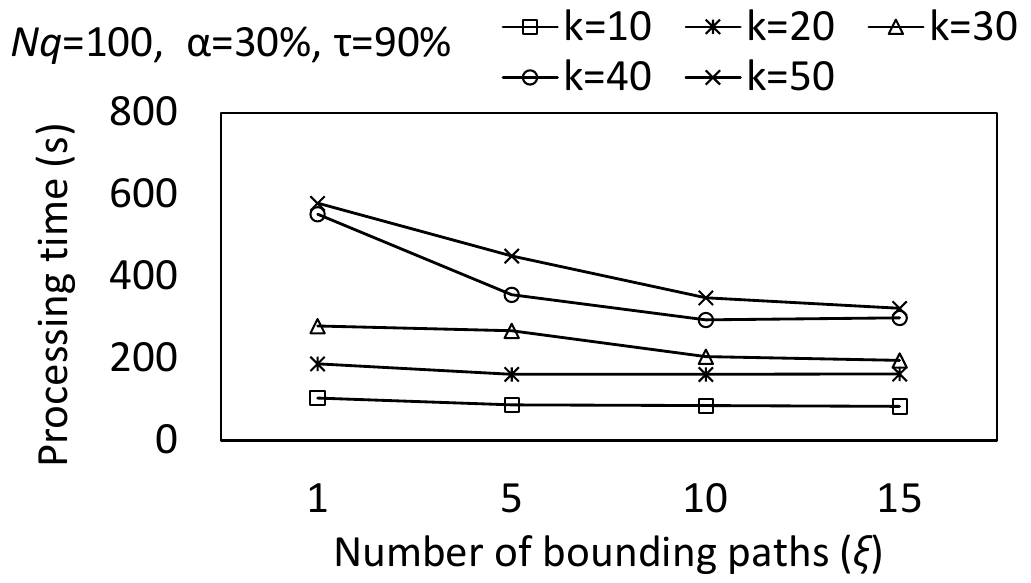}
\vspace{-0.8cm}
\caption{{Processing Time w.r.t. $\xi$}}\label{ProcessingTime-xi}
\end{minipage}
\begin{minipage}[t]{0.235\textwidth}
\centering
\includegraphics[width=\textwidth]{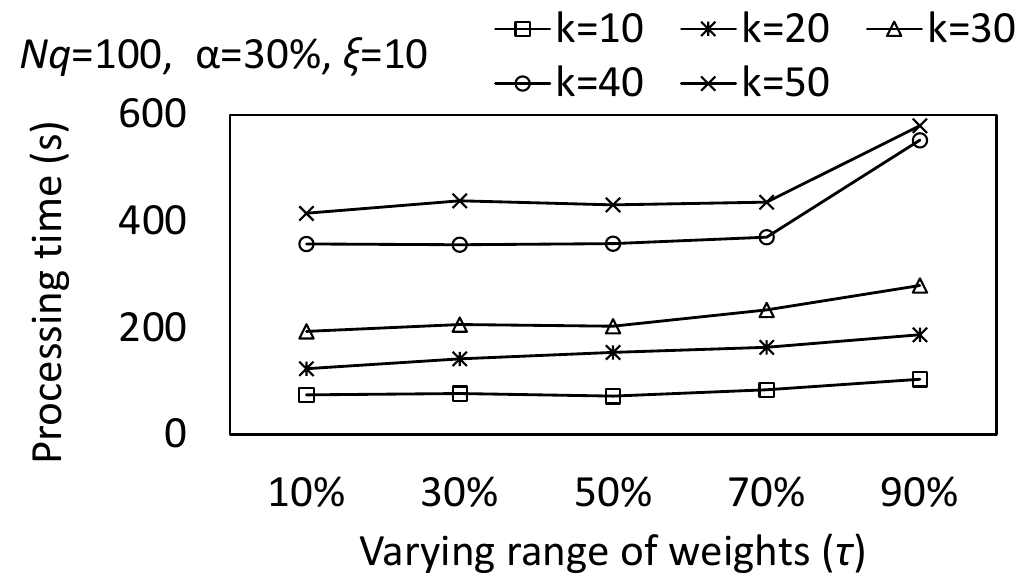}
\vspace{-0.8cm}
\caption{{Processing Time w.r.t. $\tau$}}\label{ProcessingTime-tau}
\end{minipage}
\end{figure*}

\begin{figure*}
    \begin{minipage}[t]{0.235\textwidth}
\centering
\includegraphics[width=\textwidth]{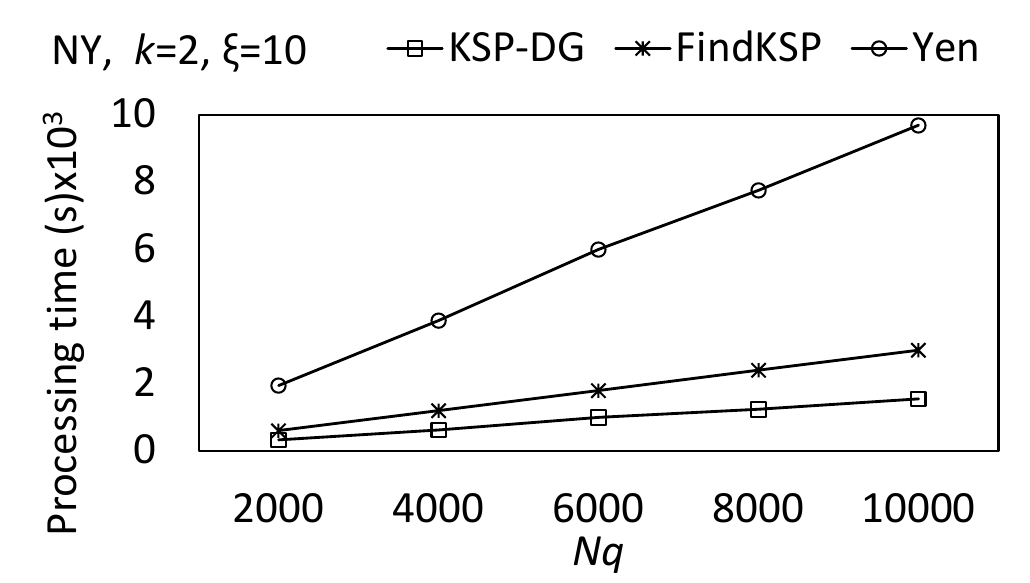}
\vspace{-0.8cm}
\caption{\scriptsize{Comparison in NY}}\label{NY-queryTime-compare}
\end{minipage}
\begin{minipage}[t]{0.235\textwidth}
\centering
\includegraphics[width=\textwidth]{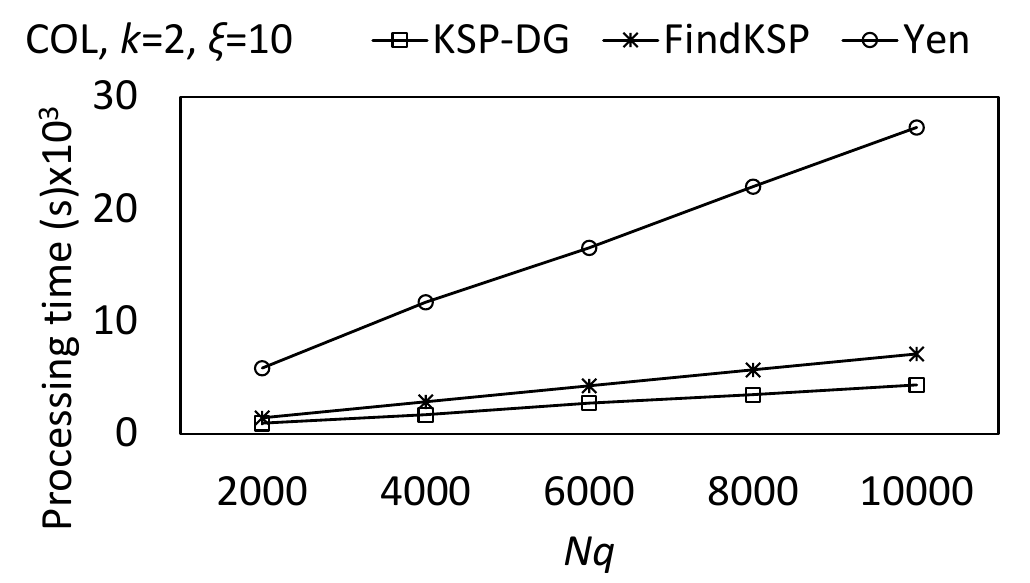}
\vspace{-0.8cm}
\caption{\scriptsize{Comparison in COL}}\label{COL-queryTime-compare}
\end{minipage}
\begin{minipage}[t]{0.235\textwidth}
\centering
\includegraphics[width=\textwidth]{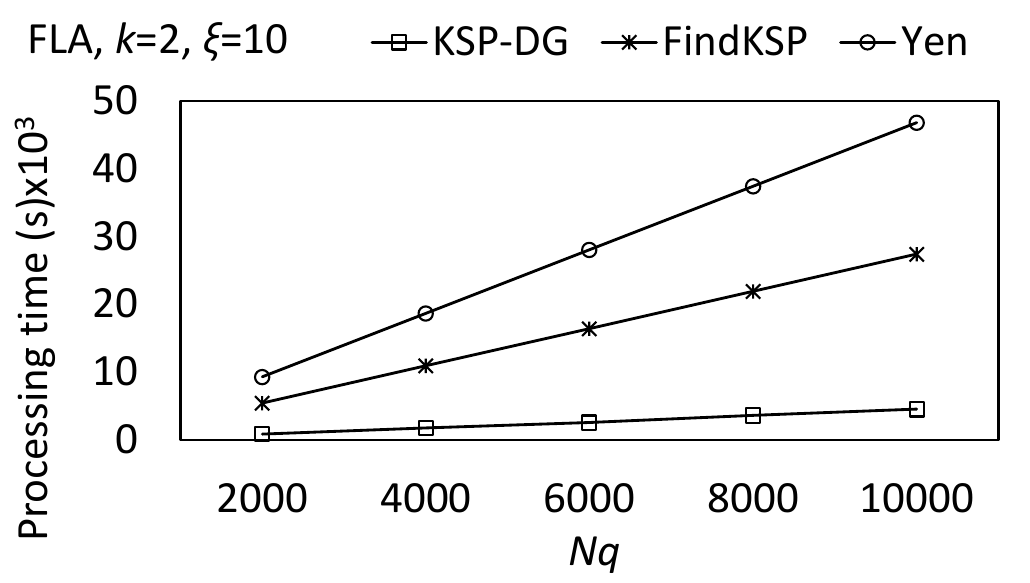}
\vspace{-0.8cm}
\caption{\scriptsize{Comparison in FLA}}\label{FLA-queryTime-compare}
\end{minipage}
\begin{minipage}[t]{0.235\textwidth}
\centering
\includegraphics[width=\textwidth]{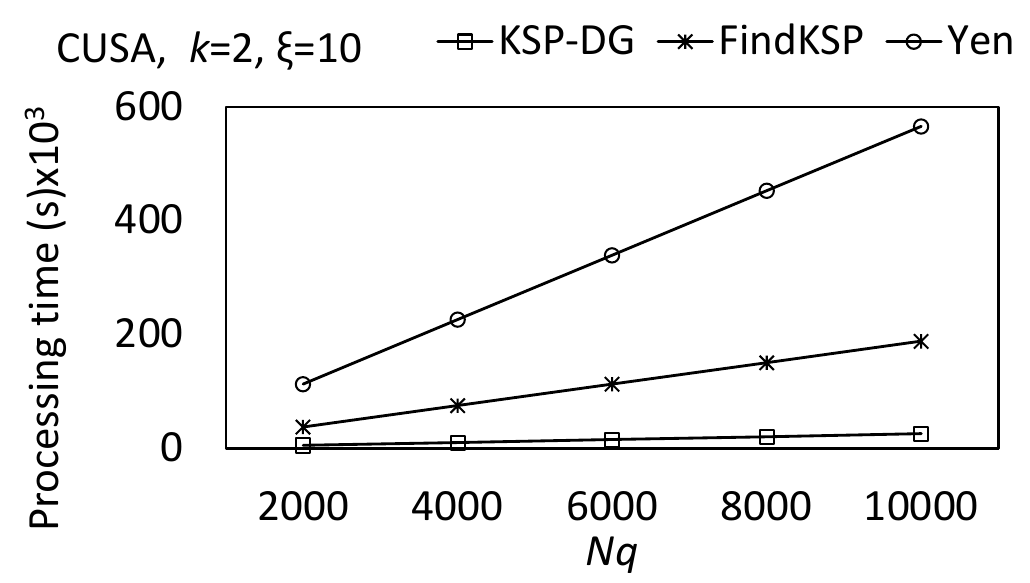}
\vspace{-0.8cm}
\caption{\scriptsize{Comparison in CUSA}}\label{CUSA-queryTime-compare}
\end{minipage}

\end{figure*}

\vspace{-0.5em}
\subsubsection{Query Processing Time w.r.t. $\xi$ and $\tau$}
Figures~\ref{ProcessingTime-xi}-\ref{ProcessingTime-tau} present the impact of $\xi$ and $\tau$ on the running time of KSP-DG. We only display the performance on {NY}; similar trends are observed  on the other datasets. Figure~\ref{ProcessingTime-xi} shows that the running time decreases with increasing $\xi$. This is because a larger $\xi$ leads to a smaller number of iterations, as demonstrated in  Figure~\ref{iterations-xi}. This trend is more apparent when {\em k} takes larger values as more iterations are required by KSP-DG when $k$ is large (as shown in Figure~\ref{iterations-k}). The relationship of $\tau$ and the processing time is evaluated in Figure~\ref{ProcessingTime-tau}, where the processing time slowly increases when $\tau$ grows, as a larger $\tau$ leads to a greater number of iterations (Figure~\ref{iterations-varying-range}).

\subsection{Comparison with Baseline Algorithms}
{We compare KSP-DG to FindKSP \cite{liu2018finding} and Yen's algorithm \cite{yen1971finding} on scalability with respect to the number of queries and $k$, as well as with CANDS \cite{yang2014cands} for the case of $k$=1.}   

Figures~\ref{NY-queryTime-compare}-\ref{CUSA-queryTime-compare} show the scalability comparison of the three algorithms when processing an increasing number of queries in each graph. As is clear from the figures,  KSP-DG significantly outperforms the other two algorithms with a lower rate of growth in processing time. The reason is that KSP-DG allows decomposing the KSP search problem into (much) smaller procedures to be carried out in parallel over the cluster, which is infeasible for FindKSP or Yen's algorithm that has to process the queries in sequence. The performance improvement of KSP-DG over the other two algorithms is more remarkable in large graphs (e.g., CUSA), which demonstrates that KSP-DG is highly suitable in dealing with large graphs and large volumes of concurrent queries.   

\begin{figure*}[!h]
\vspace{-0.05in}
\begin{minipage}[t]{0.235\textwidth}
\centering
\includegraphics[width=\textwidth]{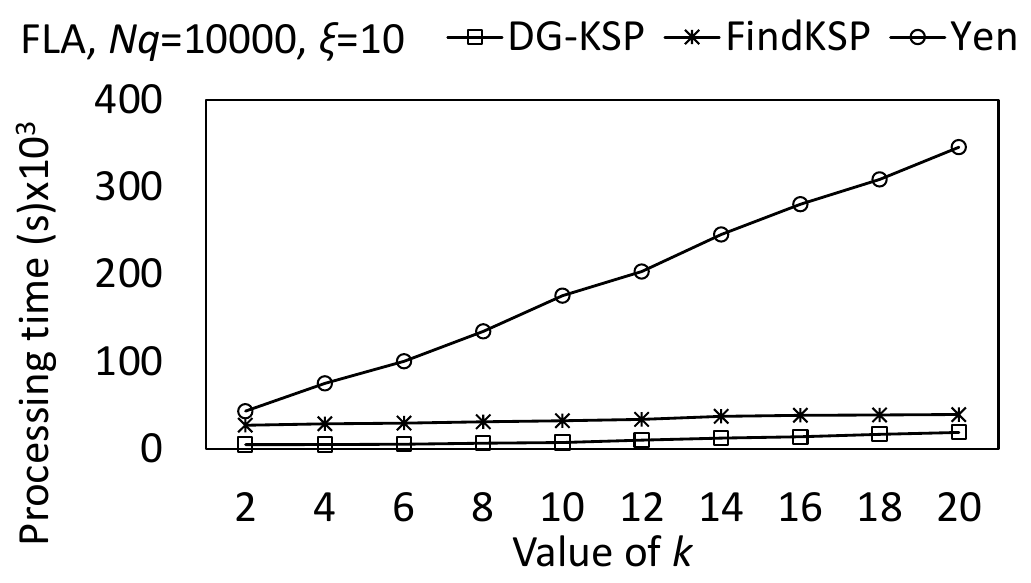}
\vspace{-0.8cm}
\caption{\scriptsize{Comparison w.r.t. {\em k}}}\label{Processing-time-comparison-k}
\end{minipage}
\begin{minipage}[t]{0.235\textwidth}
\centering
\includegraphics[width=\textwidth]{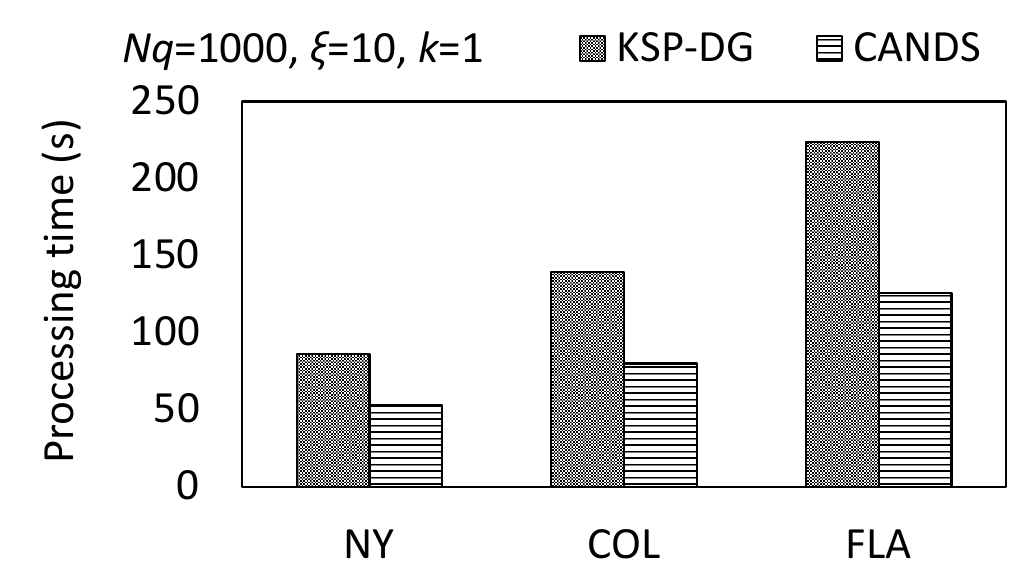}
\vspace{-0.8cm}
\caption{\scriptsize{Comparison with CANDS}}\label{processing time with CANDS}
\end{minipage}
\begin{minipage}[t]{0.23\textwidth}
\centering
\includegraphics[width=\textwidth]{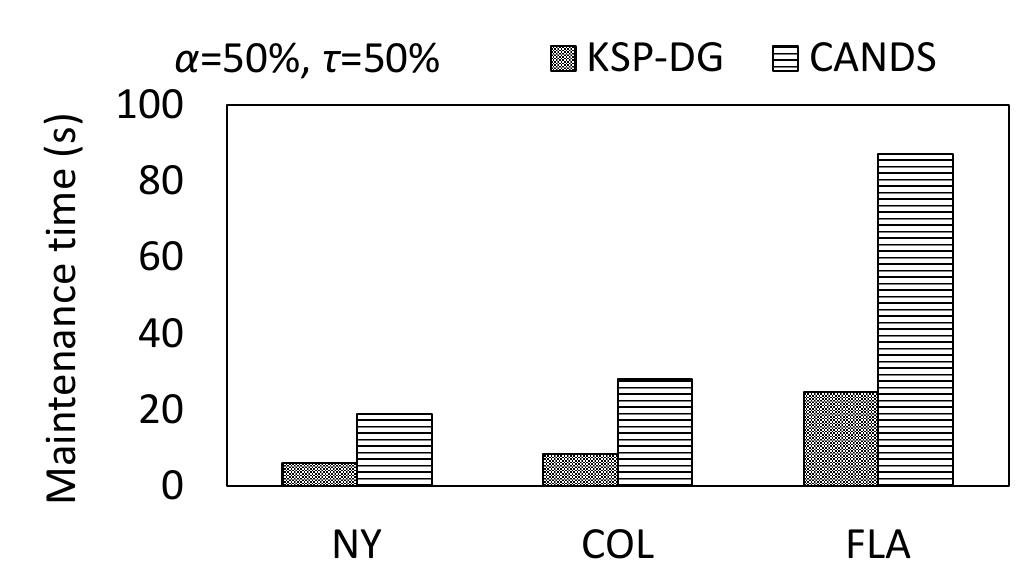}
\vspace{-0.8cm}
\caption{Maintenance Cost}\label{Maintenance-cands}
\end{minipage}
\begin{minipage}[t]{0.24\textwidth}
\centering
\includegraphics[width=\textwidth]{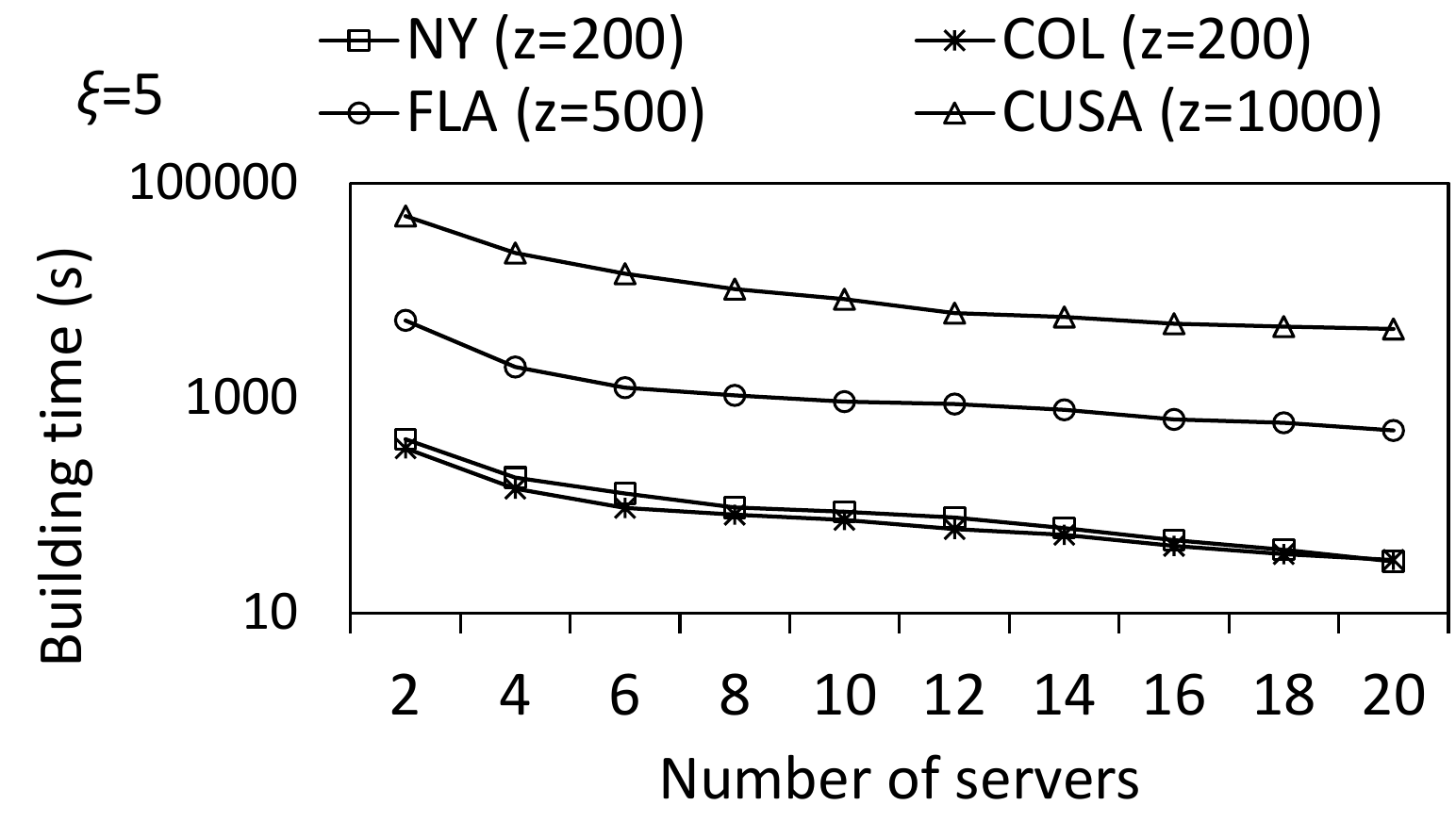}
\vspace{-0.8cm}
\caption{Building Time w.r.t. $N_s$}\label{DTLP-scalability}
\end{minipage}
\vspace{-0.5em}
\end{figure*}

\begin{figure*}[!h]
\vspace{-0.05in}
\begin{minipage}[t]{0.235\textwidth}
\centering
\includegraphics[width=\textwidth]{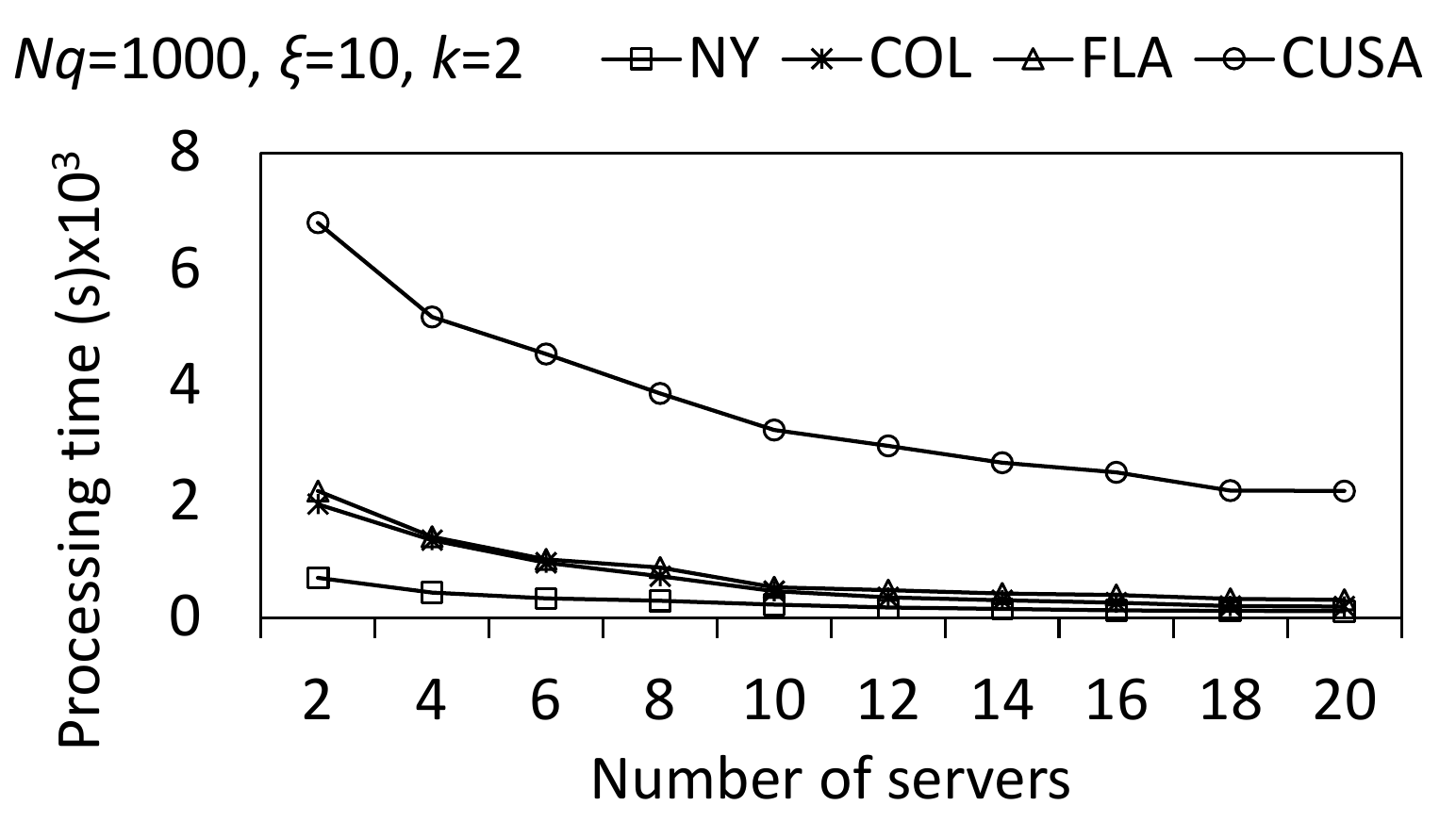}
\vspace{-0.8cm}
\caption{{Processing Time w.r.t. $N_s$}}\label{KSP-DG-Scalability}
\end{minipage}
\begin{minipage}[t]{0.235\textwidth}
\centering
\includegraphics[width=\textwidth]{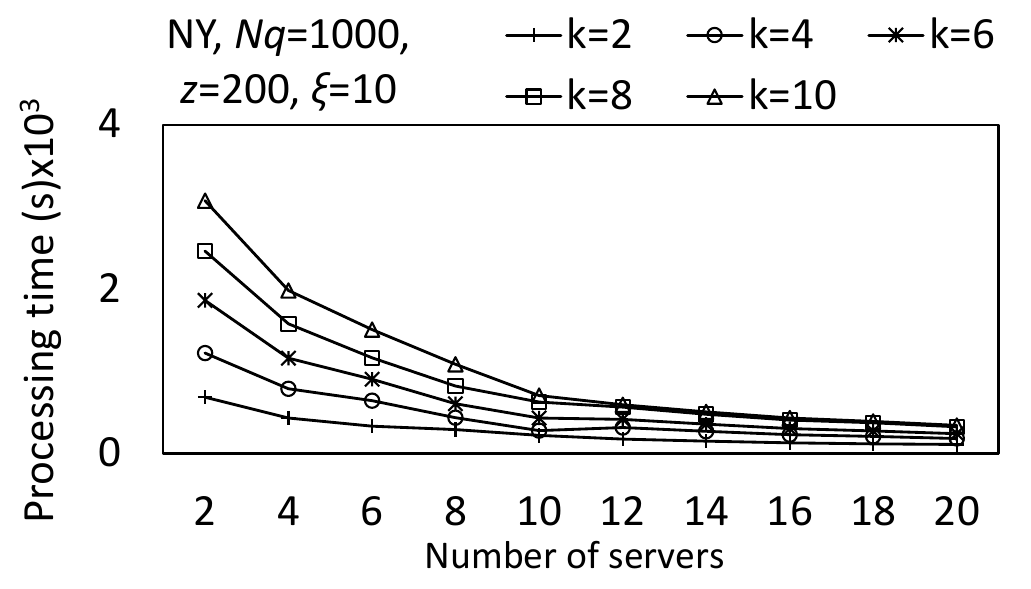}
\vspace{-0.8cm}
\caption{Processing Time w.r.t. k}\label{KSP-DG-Scalability-k}
\end{minipage}
\vspace{-0.05in}
\begin{minipage}[t]{0.235\textwidth}
\centering
\includegraphics[width=\textwidth]{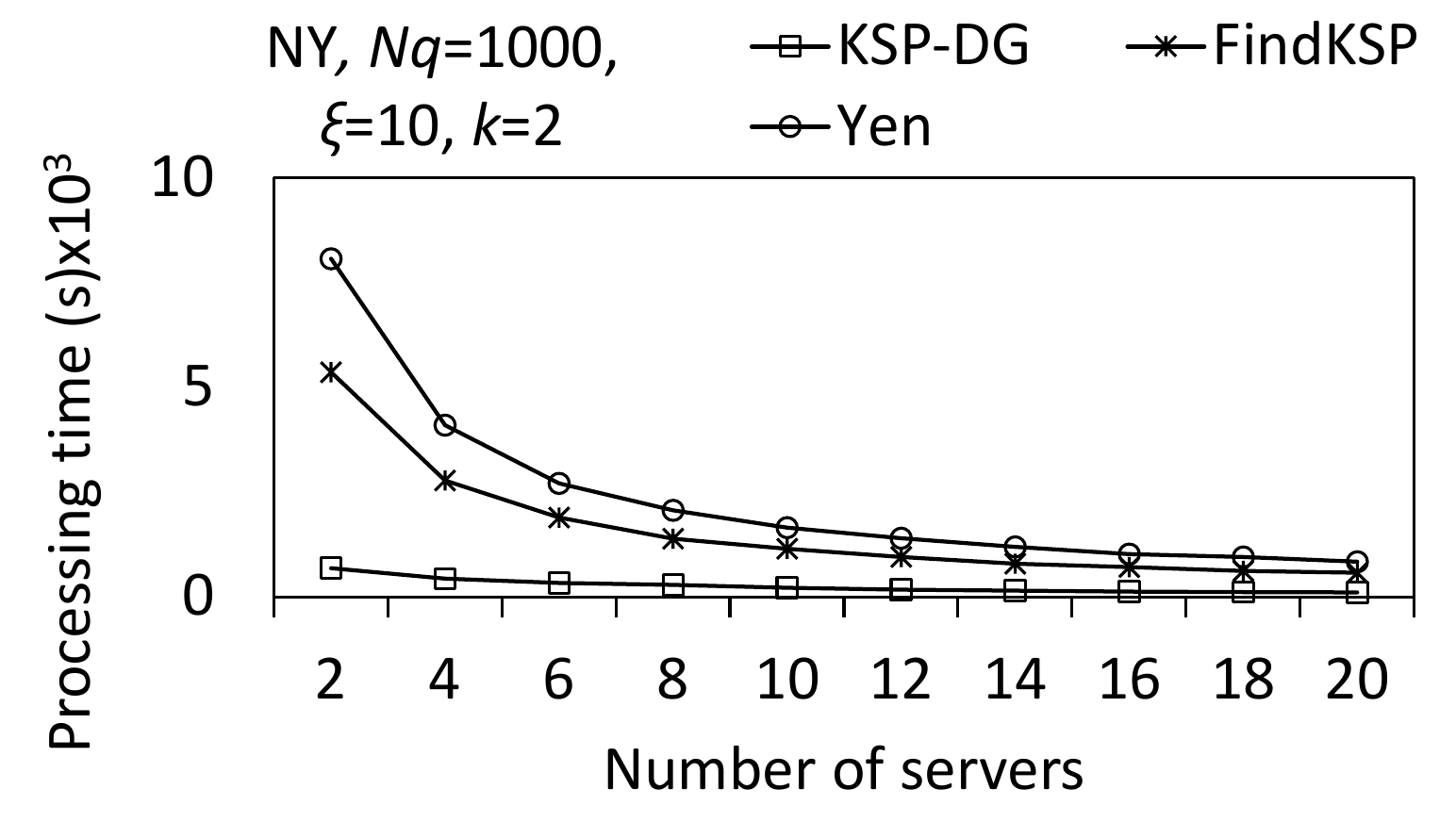}
\vspace{-0.8cm}
\caption{Scalability Comparison}\label{KSP-DG-Scalability-comparison}
\end{minipage}
\begin{minipage}[t]{0.235\textwidth}
\centering
\includegraphics[width=\textwidth]{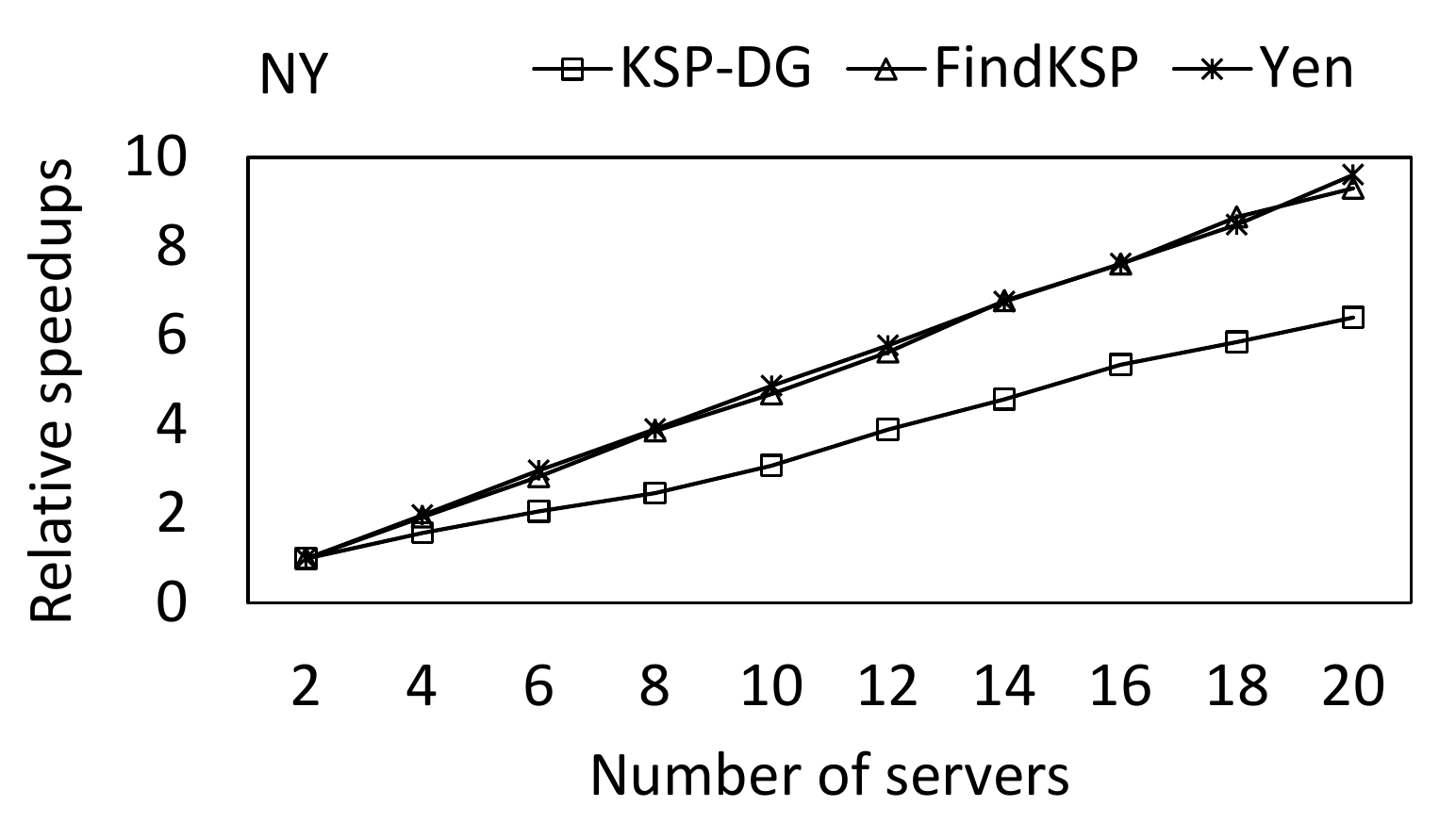}
\vspace{-0.8cm}
\caption{{Relative Speedups}}\label{relative-speedup}
\end{minipage}
\end{figure*}

The comparison of three algorithms w.r.t. $k$ is shown in Figure~\ref{Processing-time-comparison-k}. Given a batch of queries on FLA, we vary the values of {\em k} and measure the cost of the three algorithms for processing this batch of queries. As Figure~\ref{Processing-time-comparison-k} shows, the processing time of KSP-DG and FindKSP grows at a much slower rate than that of Yen as $k$ increases, and KSP-DG performs better than FindKSP. The results on the other datasets demonstrate similar trends and are omitted due to space limitations.

{Figures \ref{processing time with CANDS}-\ref{Maintenance-cands} present the comparison of KSP-DG and CANDS on the processing and maintenance cost for the same set of 1,000 queries ($k$=1) and varying edge weights. As shown in Figure \ref{processing time with CANDS}, CANDS outperforms KSP-DG for single shortest path queries, as the shortest paths between boundary vertices indexed by CANDS better facilitate identifying the single shortest path than the lower bounding paths in KSP-DG. However, the shortest paths indexes in CANDS incur a much greater maintenance cost in dynamic graphs, as shown in Figure~\ref{Maintenance-cands}. This is because almost all indexed shortest paths have to be recomputed when a significant number of edges (e.g., $\alpha=50\%$) are changing weights.}

\vspace{-0.5em}
\subsection{Scaling-out}\label{sec:scale-out}
{To further evaluate the horizontal scalability of DTLP and KSP-DG, we extend the cluster to 20 servers and the results are shown in Figures~\ref{DTLP-scalability}-\ref{relative-speedup}.} As is clear from Figure~\ref{DTLP-scalability}, the building time of DTLP decreases when more servers are introduced, as DTLP can distribute the load to all servers, and thus is able to utilize the power of a cluster.

Figure~\ref{KSP-DG-Scalability} demonstrates the performance of KSP-DG with a varying number of servers.  We feed a batch of 1,000 queries into KSP-DG with a different number of servers employed. In each figure, we observe a notable reduction of processing cost when more servers are used, which testifies to the horizontal scalability of KSP-DG. We also feed KSP-DG with 1,000 queries with different values of $k$ to further validate its scalability by varying the number of servers, and the results are illustrated in Figure~\ref{KSP-DG-Scalability-k}. It is observed that the running time of KSP-DG significantly decreases with more servers being used, regardless of the value of $k$.

Figure~\ref{KSP-DG-Scalability-comparison} depicts the scalability of the three algorithms when processing the same group of queries on a different number of servers. Since FindKSP and Yen's algorithm are centralized, we run these two algorithms on every server individually and then distribute all queries to the adopted servers randomly for fair comparison. The results show that KSP-DG always outperforms FindKSP and Yen's algorithm. {The relative speedups of the three algorithms with a varying number of servers are shown in Figure~\ref{relative-speedup}. It can be observed that the relative speedup of each algorithm grows linearly with the number of servers.}

{Finally, we measure the load on each server in terms of average CPU utilization and memory consumption when  the size of the cluster increases for building DTLP, updating DTLP, and processing queries on CUSA. The results show that the difference between the maximum and minimum CPU utilization across the cluster is consistently less than 6\%, and for memory utilization, less than 2\%, indicating load balancing can be achieved over a range of cluster sizes. }
\vspace{-0.5em}
\section{Related work}\label{sec:related-work}
 There exists a large body of work on the problem of identifying paths with certain properties in graphs \cite{Route-Planning-Hannah-2016,meng2015h,li2018resilient, li2020leveraging}. We briefly review literature relevant to at least one of the three components to our work and its proposed solutions, namely, {\em k shortest paths}, {\em dynamic graphs}, and {\em distributed processing}.  We provide a summary in the following table, where {\small $\otimes$} denotes the relevance of the work to a specific component.
\begin{table}[htbp]\scriptsize
\vspace{-0.18cm}
\centering
\begin{tabular}{|m{12em}|m{5em}|m{5em}|m{5em}|}
\hline
 &$k$ Shortest Paths& Dynamic Graphs & Distributed Processing\\
\hline
\makecell[cl]{Centralized KSP algorithms\\\cite{liu2018finding,yen1971finding,eppstein1998finding,hershberger2007finding,gao2010fast,gao2012holistic,hershberger2001vickrey,chang2015efficiently}} &$\otimes$& $-$& $-$\\
\hline
\makecell[cl]{Distributed SSP algorithms\\\cite{aridhi2015mapreduce,qiu2018parapll,chandy1982distributed, DistributedSP-Baruch-1989,Elkin2017Distributed,Ghaffari2017Improved}} &$-$& $-$& $\otimes$ \\
\hline
\makecell[cl]{Distributed SSP algorithms\\ in a dynamic graph \cite{yang2014cands}}&$-$&$\otimes$&$\otimes$\\
\hline
\end{tabular}
\vspace{-0.35cm}
\end{table}

\textbf{Centralized KSP Algorithms.}
Yen's algorithm \cite{yen1971finding} identifies KSPs based on a deviation paradigm, which first computes the shortest path, and then generates all candidate paths that deviate from this path by applying Dijkstra's algorithm repeatedly, from which the shortest is selected as the next shortest path. It repeats the above steps until KSPs have been determined. Some methods \cite{hershberger2001vickrey,hershberger2007finding,sun2009study,chang2015efficiently, eppstein1998finding, gao2010fast, gao2012holistic,liu2018finding} are proposed to further optimize the generation of candidate paths.  \cite{hershberger2007finding, hershberger2001vickrey, chang2015efficiently} partition the candidate shortest paths into equivalence classes and safely prune specific classes that are impossible to contain {\em k} shortest paths. Others adopt the Shortest Path Tree (SPT for short) \cite{eppstein1998finding,gao2010fast,gao2012holistic,chang2015efficiently,liu2018finding} to help identify candidate shortest paths, where SPT maintains the shortest paths from all vertices to the terminal vertex. 

All of the above proposals suffer from the following drawbacks if directly applied to our problem. First, most of them require access to the entire graph during their operation; as a result the only option is to replicate the entire dynamic graph on each server, allowing them to operate in a distributed fashion. This however has significant cost and scalability implications. Second, the majority of these algorithms adopt a sequential strategy that necessitates a search for the shortest paths one after the other, which limits their ability to handle many concurrent queries in a distributed setting. Finally, some of them require building a path index such as SPT \cite{eppstein1998finding,gao2010fast,gao2012holistic} for every query, which is too heavy-weight for dynamic graphs as the index often become invalid due to varying weights.

\textbf{Distributed SSP Algorithms.}
Past work focuses on identifying a Single Shortest Path (SSP) \cite{chandy1982distributed,DistributedSP-Baruch-1989,Elkin2017Distributed,Ghaffari2017Improved,Forster2017A,aridhi2015mapreduce,qiu2018parapll} over a {\em static} graph in a distributed fashion. \cite{chandy1982distributed, DistributedSP-Baruch-1989,Elkin2017Distributed,Ghaffari2017Improved} aim to determine the SSP in a communication network, which is distributed if each vertex represents a processor. 
Others investigate the distributed computation of the SSP on a cluster of servers \cite{aridhi2015mapreduce, qiu2018parapll}. Qiu et al. \cite{qiu2018parapll} identify the SSP based on the principle of pruned landmark labeling \cite{akiba2013fast}, while Aridhi et al. \cite{aridhi2015mapreduce} propose a distributed algorithm to identify an approximate shortest path in a large-scale network with MapReduce.

These distributed SSP algorithms suffer from the following problems when applied to our setting. First, they are designed for static graphs, and thus most index structures employed in those approaches such as  {\em k-shortcut hopset} \cite{Elkin2017Distributed,Ghaffari2017Improved,Forster2017A} and {\em node labels} \cite{qiu2018parapll} quickly become obsolete in a dynamic graph rendering them unusable. Second, all of them focus on the search for a single shortest path and its nontrivial to extend, to the case of finding KSPs.

\textbf{Distributed SSP in Dynamic Graphs.} Yang et al. \cite{yang2014cands} propose a distributed algorithm CANDS to identify the SSP in a dynamic graph. It partitions the entire graph into subgraphs residing on different servers, and indexes the shortest path between any pair of boundary vertices within each subgraph. Although seemingly similar to our solution, this approach would not work well when directly applied to our problem. First, the search procedure of CANDS is essentially sequential. For a given query, it starts from the subgraph covering the source vertex and iteratively expands to other subgraphs via the indexed shortest paths using Dijkstra's algorithm until reaching the subgraph containing the destination vertex. The examined subgraphs are not known initially and they have to be explored in order. Second, the indexed shortest paths in CANDS require frequent re-computation in a dynamic graph, which can be quite expensive. Finally, CANDS is  designed for identifying single shortest path, and it is non-trivial to adapt this paradigm to identify KSPs. 
\vspace{-1em}
\section{Conclusions and Future Work}\label{sec:con}
The problem of identifying KSPs over road networks is fundamental to many location-based services. Our work is the first to focus on this problem,  and proposes a suite of distributed solutions. Consisting of the bounding paths and the skeleton graph, DTLP provides reference paths assisting the identification of relevant subgraphs that need to be explored to identify KSPs. Since the bounding paths in DTLP do not change with varying weights, DTLP is light-weight for dynamic graphs with low maintenance cost. Based on DTLP, the KSP-DG algorithm is designed to run in a distributed setting. Each iteration of KSP-DG consists of filter and refine steps: the filter step first identifies the relevant subgraphs that are most likely to cover the KSPs, and the refine step then examines these subgraphs to generate candidate {\em k} shortest paths, which are used to update the list of shortest paths obtained so far. 

As future work, one may consider addressing variants of the problem studied in this paper. 
For example, a constrained version of the KSP query may require all shortest paths to pass through some designated vertices; another version may involve limiting the diversity of the shortest paths to below a certain threshold. These variants have important practical applications, and it is worthwhile to investigate solutions in a distributed environment.
\balance
\bibliographystyle{ACM-Reference-Format}
\bibliography{reference}

\end{document}